\pdfoutput=1
\documentclass[a4paper,11pt]{article}
\usepackage[margin=25mm]{geometry}
\usepackage[T1]{fontenc}
\usepackage{lmodern}
\usepackage{amsmath,amssymb,amsthm,booktabs,graphicx,xcolor}
\usepackage{array}
\newcommand{\exptag}[1]{\textsf{Exp-#1}}
\usepackage{enumitem}
\usepackage{float}
\usepackage{xurl}
\usepackage{algorithm}
\usepackage[noend]{algpseudocode}
\usepackage[colorlinks=true,linkcolor=blue!45!black,urlcolor=blue!45!black,citecolor=blue!45!black]{hyperref}
\setlist{nosep}

\numberwithin{equation}{section}
\numberwithin{figure}{section}
\numberwithin{table}{section}
\numberwithin{algorithm}{section}
\newtheorem{theorem}{Theorem}[section]
\newtheorem{proposition}[theorem]{Proposition}
\newtheorem{lemma}[theorem]{Lemma}
\newtheorem{corollary}[theorem]{Corollary}
\newtheorem{definition}[theorem]{Definition}

\newcommand{\E}{\mathbb E}
\newcommand{\Var}{\operatorname{Var}}

\newcommand{\Pois}{\operatorname{Pois}}
\newcommand{\Skellam}{\operatorname{Skellam}}

\let\plaineqref\eqref
\renewcommand{\eqref}[1]{Eq.~\plaineqref{#1}}
\title{Random-Sign IBLTs: Calibration-Free Self-Sizing After Decoding Failure}
\author{Min Wu\\ \normalsize\texttt{wumin@hdu.edu.cn}}
\begin{document}
\maketitle
\begin{abstract}
In a two-round set reconciliation protocol based on an Invertible Bloom Lookup Table (IBLT), the first round sends a fixed-capacity table. If an oversized difference causes decoding failure, the receiver must estimate a safe upper bound on the unknown difference size $d$ from the sum of squared counts in the difference table to size the second-round table adaptively. The standard IBLT uses each element's set direction as its accumulation sign. This couples the hypergraph-overlap bias induced by failure selection with the sign composition of the two-sided difference. A one-sided difference amplifies the estimation bias through contributions of the same sign, forcing existing schemes to rely on extensive offline simulation calibration over a grid of configurations.

We propose the \emph{random-sign IBLT}. A pseudorandom sign independent of the position mapping decouples failure selection from sign composition while preserving the number and width of cell fields and the peeling recovery capability. We prove that the estimator is exactly unbiased on every fixed failed hypergraph. For large differences, we prove a Gaussian limit for the estimation error conditioned on failure and derive a closed-form Gaussian upper bound. For small differences, we characterize a discrete limit distribution dominated by duplicate edge pairs and construct a Skellam upper-bound rule whose conditional miss rate is asymptotically at most the target risk, uniformly over all $d\ge2$ without offline calibration. Five hundred million paired trials and end-to-end experiments with the real protocol spanning several orders of magnitude in load show that the random-sign IBLT eliminates direction sensitivity. For small differences, the conditional miss rate of the Gaussian upper bound exceeds the target by more than a factor of 8, whereas the Skellam upper-bound rule keeps it at or below the target level within sampling error across all 189 load configurations. The total two-round communication cost approaches that of an oracle using the true difference size in the same protocol.
\end{abstract}

\section{Introduction}\label{sec:intro}
Many distributed systems require two network endpoints holding sets $A$ and $B$ to find their symmetric difference $A\triangle B$, the elements present on only one side. This problem is called set reconciliation. The Invertible Bloom Lookup Table (IBLT)\cite{gm,eppstein2011whatsdifference} is the standard data structure for this task. Each side builds a table of $M$ cells and hashes each element into the count, keySum, and checksum fields of $k$ cells. The receiver then subtracts the tables cell by cell. Shared elements contribute identical values to the same cells and cancel exactly, leaving only the contributions of the $d=|A\triangle B|$ differing elements. An element present only in $A$ has direction $+1$, and an element present only in $B$ has direction $-1$. When the difference size $d$ is small enough relative to the table size $M$ (the load $d/M$ is below the peeling threshold), the receiver repeatedly finds a pure cell containing a single differing element and recovers all differing elements through peeling. Once $d$ exceeds the table's capacity, peeling stalls and decoding fails.

An IBLT presents a basic tension: the table size $M$ must be fixed at construction, yet the difference size $d$ remains unknown until synchronization completes. In practice, $d$ for the same pair of nodes can vary by several orders of magnitude over time. During routine synchronization, the two sides are nearly identical and $d$ is small. After recovery from a network partition or a batch write, $d$ can grow sharply. Allocating a large table for the worst case wastes bandwidth in most routine synchronizations. Allocating a small table causes decoding failure when the difference is large. The two-round protocol\cite{wu2026selfsizing} addresses this tension. The first round sends a small table of fixed size, allowing routine small differences to complete in one round. If an oversized difference causes decoding failure, the count array of the difference table still retains the superimposed contributions of all $d$ differing elements. The receiver estimates an upper bound $U$ on $d$ from the sum of squared counts in this table. The second round allocates sufficient capacity based on $U$ and completes synchronization. The key requirement is to control the underestimate risk conditioned on failure (the conditional miss rate). If $U<d$ (a miss), insufficient capacity puts the second-round table at risk of another failure. An overly conservative $U$ wastes second-round bandwidth.

The method in \cite{wu2026selfsizing} uses the count sum and the sum of squared counts in the first-round difference table to construct an unbiased estimator $\widehat d$ and obtain an upper bound $U$ by inversion. Estimating $d$ after failure is difficult because decoding failure itself selects the sample, and this selection couples with the sign composition of the set difference (the ratio of differing elements from $A$ to those from $B$). When we estimate $d$ with this sum-of-squares method, each element's own contribution is fixed. The estimation error comes entirely from cross terms between differing elements that fall into the same cell. The standard IBLT uses each element's set-membership direction ($+1$ or $-1$) as its accumulation sign in count. Two elements with the same direction add when they share a cell, increasing the sum of squares. Elements with opposite directions cancel. Peeling success depends on how the position mappings of differing elements to cells overlap. Conditioning on failure therefore changes the overlap distribution. Together, these effects produce systematic bias in the failure sample. For a one-sided difference, all differing elements come from the same side. This is common in one-way synchronization, for example when a replica lacks only data newly written to the primary. All cross terms then have the same sign, so the change in overlap induced by failure is amplified in the sum of squared counts, substantially biasing the estimate. In two-way synchronization where each side writes new data, the numbers of differing elements on the two sides are comparable. For this balanced difference, positive and negative cross terms largely cancel, leaving the post-failure mean estimate nearly unchanged. The bias magnitude varies continuously with the ratio of differing elements on the two sides, and its sign reverses with the number of hashes $k$ (Section~\ref{sec:model}). After decoding failure, the receiver does not know whether the difference comes from one side or both. The observed sum of squared counts therefore cannot reveal how much is due to the true $d$ and how much to the direction bias induced by failure selection. Thus, except in deep overload, the upper-bound rule for the standard IBLT in \cite{wu2026selfsizing} relies in practice on offline simulation calibration for each configuration; its analytic bound is too conservative for small $M$ and may yield no usable quantile.

Our core insight is that separating an element's set direction from its accumulation sign removes the source of this coupling. We therefore propose the random-sign IBLT. The change is lightweight: the two sides share a random-sign hash independent of the position mapping, assigning each key a random $+1$ or $-1$. On insertion, the element's random sign multiplies its contributions to all three fields: count, keySum, and checksum.

This modification supports both decoding and estimation. Shared elements receive the same random sign in both tables and still cancel exactly after subtraction. A cell containing one unrecovered differing element has count $+1$ or $-1$. Multiplying keySum by count cancels the random sign exactly (since $(\pm1)^2=1$) and recovers the original key. Computing the key's random sign and comparing it with count then reveals whether the key belongs to set $A$ or $B$. Under ideal purity, the random-sign IBLT therefore peels exactly the same elements as the standard IBLT at every step. The decoding threshold, failure probability, and number and byte width of cell fields remain unchanged. For estimation, each differing element's coding coefficient in the difference table becomes an independent uniform random sign, whether the difference is one-sided or balanced. Peeling success depends only on which cells the elements select (the mapping), while cancellation of the estimation error depends on the random signs. These are independent. For every mapping selected by failure, the random signs remain independent and uniform, so their contributions cancel in expectation. Using random signs to cancel cross terms is a classical idea from the AMS sketch and CountSketch\cite{alon1999spacecomplexity,charikar2004countsketch}. We implement it in a peelable data structure, enabling one table to support both difference recovery and difference-size estimation and allowing rigorous statistical analysis of the post-failure estimator.

\paragraph{Main results and contributions.}
We prove that the random-sign IBLT preserves the standard IBLT's decoding behavior under ideal purity and that the joint distribution of the difference-table statistic and the mapping is identical for all sign compositions. Section~\ref{sec:design} presents the encoding and peeling algorithm that achieves this decoupling. We then prove that the estimator $\widehat d$ is exactly unbiased on every fixed mapping, giving conditional mean $d$ after failure, and fully characterize the conditional distribution of the estimation error across all load scales (Sections~\ref{sec:measurement} and \ref{sec:inversion}). For large differences, we derive a Gaussian rule from the conditional Gaussian limit. For small differences, we characterize a discrete limit distribution and construct a Skellam rule whose conditional miss rate is asymptotically at most the target risk, uniformly over all difference sizes without offline simulation calibration. We also lift the known duplicate-pair mechanism to the distribution level: at low load, the failure-conditioned distribution is close in total variation to the planted distribution, which reduces the conditional analysis of rare failures to direct computation under the planted distribution and yields an efficient planted sampling method. The experiments (Section~\ref{sec:exp}) verify that random-sign and standard encoding have identical decoding behavior. They show that the Skellam rule keeps the conditional miss rate at or below the target level within sampling error across all 189 configurations spanning several orders of magnitude in load, whereas the Gaussian rule exceeds the target by more than a factor of 8 for small differences. With this rule, the two-round protocol reliably completes synchronization at a total transmission cost close to that of an oracle using the true $d$ in the same protocol.

\paragraph{Structure of the paper.}
Section~\ref{sec:model} introduces the standard IBLT model and analyzes its post-failure estimation bias. Section~\ref{sec:design} presents the random-sign IBLT and the two-round protocol. Section~\ref{sec:measurement} proves exact unbiasedness given the mapping and derives the Gaussian rule for large differences. Section~\ref{sec:inversion} derives the discrete limit distribution for small differences and the Skellam rule for all scales. Section~\ref{sec:exp} reports the experimental evaluation. Sections~\ref{sec:related} and \ref{sec:conclusion} discuss related work, rule selection, and computational cost. The appendix contains complete proofs and supplementary experiments.

\section{Model and post-failure estimation bias}\label{sec:model}
We introduce the standard IBLT and the ideal peeling model, define the conditional miss rate required by the two-round protocol, characterize failure at low load, and explain how the standard IBLT's post-failure bias depends on sign composition. Appendix~\ref{app:notations} collects the notation used throughout the paper.

\subsection{The standard IBLT and the ideal peeling model}
The two sides hold sets $A$ and $B$ and seek their symmetric difference $A\triangle B$, whose size $d=|A\triangle B|$ is unknown in advance. The standard IBLT consists of $M$ cells, each with three fields: count, keySum, and checksum. To insert a key $x$, the cell hash selects $k$ distinct cells from $[M]$ and adds $1$, $x$, and $\mathrm{chk}(x)$ to their respective fields. Each side builds a table, and the tables are subtracted cell by cell. Shared keys contribute identical values to the same cells in both tables, so their contributions cancel in all three fields. Only the contributions of the $d$ differing keys remain. A key present only in $A$ has direction $s_i=+1$, and a key present only in $B$ has direction $s_i=-1$. We denote the numbers of keys in the two directions by $d_+$ and $d_-$. The ordered pair $(d_+,d_-)$ is the \emph{sign composition}, and $b=d_+-d_-=|A|-|B|$ is the difference between the set sizes. We use two extreme sign compositions throughout: a \emph{one-sided difference}, where all differing keys come from one side ($|b|=d$), and a \emph{balanced difference}, where both sides have the same number of differing keys ($d_+=d_-$, $b=0$). Both correspond to common synchronization scenarios. A one-sided difference arises in one-way synchronization when only one side inserts or deletes entries. A balanced difference arises in two-way synchronization when each side modifies a small number of entries and the insertions and deletions roughly cancel. Practical loads usually lie between these extremes, depending on how much each side changes. Decoding proceeds by peeling. We repeatedly find a pure cell containing one unrecovered differing key, recover the key by multiplying keySum by count, verify it with the checksum, and remove its contribution from all $k$ cells it selects. Removal may make other cells pure. Peeling continues until no pure cell remains.

Under ideal purity, the decoder accepts a cell as pure exactly when it contains one unrecovered differing key. Cells with count $\pm1$ that contain multiple differing keys are rejected by the checksum. Under this assumption, peeling success depends entirely on the mapping of differing keys to cells. View the $M$ cells as vertices and the $k$ cells selected by each differing key as a hyperedge. The difference table then corresponds to a random $k$-uniform hypergraph: the $d$ differing keys independently and uniformly choose $k$-subsets $S_i\subseteq[M]$, with different keys allowed to choose the same subset. We denote the unsigned mapping with key labels by $\mathcal G$ and the overlap between two edges by $H_{ij}=|S_i\cap S_j|$, the number of cells shared by differing keys $i$ and $j$. On $\mathcal G$, pure cells are exactly the unsigned degree-one vertices. Peeling repeatedly deletes edges incident to these vertices. The deletion order does not affect the terminal state, which is always the hypergraph's 2-core. A nonempty 2-core causes decoding failure, denoted by event $F$, with failure probability $\Pr(F)$. Thus $F$ depends only on the unsigned mapping $\mathcal G$. When the load $d/M$ is below the critical value $\alpha_c$ (about $0.818$ for $k=3$ and $0.772$ for $k=4$\cite{molloy,dw}), peeling succeeds with high probability. Rare failures arise from small cores such as duplicate edge pairs, and their probability decays polynomially in $M$. In the critical window of width of order $M^{-1/2}$ around $\alpha_c$, the failure probability rises from near zero to near one\cite{dm}. Above $\alpha_c$, the failure probability tends to one.

\subsection{The two-round protocol and the conditional miss rate}
Because the capacity $M$ must be fixed before $d$ is known, the two-round protocol\cite{wu2026selfsizing} first sends a small table of fixed size $M_1$. For a small difference, peeling succeeds in the first round. After a first-round failure, the receiver obtains an upper bound $U$ on $d$ from the difference table's count array and chooses the second-round size $m(U)$ based on $U$. The second round re-encodes and decodes using a fresh, independent cell hash.

Let $C_v$ denote the count of cell $v$ in the difference table. Following \cite{wu2026selfsizing}, we define the estimator $\widehat d$ of $d$ using the count sum $\sum_vC_v$ and sum of squares $\sum_vC_v^2$ saved before peeling begins:
\begin{equation}
\widehat d=\frac{1}{k(1-k/M)}\biggl(\sum_vC_v^2-\frac1M\Bigl(\sum_vC_v\Bigr)^2\biggr),\label{eq:dhat}
\end{equation}
and denote its unconditional standard deviation scale by $s(n)$:
\begin{equation}
s(n)=\sqrt{\frac{2n(n-1)}{M-1}}.\label{eq:s}
\end{equation}
For the standard IBLT, $\widehat d$ is exactly unbiased for $d$ and has variance $s(d)^2$ for every sign composition. This expectation is unconditional: it averages over all runs, regardless of failure $F$. If we evaluated the upper bound $U$ over all runs, adding a few standard deviations to $\widehat d$ would control the underestimate probability. In the two-round protocol, a successful first round ends the protocol immediately, without computing $U$. Only after decoding failure (event $F$) does the receiver compute $U$ to size the second-round table. If $U<d$, the second-round capacity falls below the prescribed margin, risking another failure. We call the event $U<d$ a \emph{miss} and its probability conditioned on first-round failure, $\Pr(U<d\mid F)$, the \emph{conditional miss rate}. The sizing rule must therefore satisfy
\begin{equation}
\Pr(U<d\mid F)\le u,\label{eq:target}
\end{equation}
where $u$ is the target risk. In deployment, we can choose $0.01$ or $0.05$, corresponding to a 1 \% or 5 \% underestimate risk.

The two-round protocol incurs a miss only when the first round fails and the upper bound underestimates $d$. This event has \emph{joint miss rate} $\Pr(F,\,U<d)=\Pr(F)\Pr(U<d\mid F)$. Controlling the conditional miss rate at $u$ bounds the joint miss rate by $u\Pr(F)$, both in overload, where failure is almost certain, and at low load, where it is rare. For failure samples without a miss ($U\ge d$), the second-round table has sufficient capacity margin to decode with high probability. Section~\ref{sec:protocol} gives the complete decomposition of the total two-round failure rate.

\subsection{Failure at low load and the planted distribution}\label{sec:plant}
\eqref{eq:target} must hold at all loads. Low load requires separate treatment. When $d$ is much smaller than $M$, the failure probability decays polynomially in $M$. Failure is then rare, making the distribution conditioned on $F$ difficult to analyze directly and requiring many simulated mappings for each failure sample. The structure of failure in this region is simple. Two edges with the same mapping form a duplicate edge pair. A single edge cannot form a 2-core, and two edges form a 2-core exactly when they are a duplicate edge pair. When $d=o(M)$, 2-cores with three or more edges are negligible relative to duplicate edge pairs. Failure is therefore approximately equivalent to the existence of a duplicate edge pair. Conditioning on $F$ is close to planting a duplicate edge pair in a random mapping. We thus define the \emph{planted distribution}: choose a pair of keys uniformly, assign both the same uniformly random $k$-subset, and let all remaining keys choose their mappings independently and uniformly.

\begin{proposition}[Failure is determined by duplicate edge pairs when $d=o(M)$]\label{prop:plant}
Fix $k\ge3$ and let $2\le d=o(M)$. The failure probability is asymptotically equal to the expected number of duplicate edge pairs: $\Pr(F)\sim\binom d2/\binom Mk$. The total variation distance between the distribution of $\mathcal G$ conditioned on $F$ and the planted distribution tends to zero.
\end{proposition}

Goodrich and Mitzenmacher\cite{gm} identify duplicate edge pairs as the main source of IBLT failure. The finite-length analysis of Yugawa and Wadayama\cite{yw} likewise shows that stopping sets of size 2 dominate. Proposition~\ref{prop:plant} extends this observation to a total variation approximation of the entire mapping for all $d=o(M)$. Lemma~\ref{lem:plant} in Appendix~\ref{app:graph} gives the quantitative error bound and proof. Proposition~\ref{prop:plant} supports both theory and experiments. It replaces conditioning on failure with the planted distribution, under which direct calculations are straightforward. The limit theorems in Sections~\ref{sec:measurement} and \ref{sec:inversion} use this approximation to handle conditioning on failure when $d=o(M)$. The risk decomposition in Section~\ref{sec:protocol} also uses it to estimate the first-round failure probability. For experiments, it enables direct sampling from the planted distribution, where every sample yields a failure. Section~\ref{sec:exp-setup} uses this planted sampling to evaluate low-load configurations beyond the reach of direct sampling and checks its finite-size accuracy where both methods are feasible.

\subsection{Post-failure bias of the standard IBLT}
Controlling \eqref{eq:target} is difficult because the failure event $F$ selects the mapping $\mathcal G$. This selection changes the distribution of the overlaps $H_{ij}$. At low load, failure samples mainly contain duplicate edge pairs (Proposition~\ref{prop:plant}), whose overlap is $k$, increasing the conditional mean of $H_{ij}$. In the critical window and at higher load, failure corresponds to mappings with large cores. Such a core is a local structure, and the direction of the change in the conditional mean of $H_{ij}$ depends on $k$, as measured by $c_F$ below. Statistics that depend on $H_{ij}$ therefore require separate analysis after conditioning on $F$.

In the standard IBLT, the difference-table count is $C_v=\sum_{i:v\in S_i}s_i$, and each differing key's coding coefficient is its fixed set direction $s_i\in\{+1,-1\}$. In the sum of $C_v^2$, the diagonal terms contribute exactly $kd$. The estimation error comes entirely from the cross terms $s_is_jH_{ij}$ between keys sharing cells. Excess overlap raises $\widehat d$ for same-sign pairs and lowers $\widehat d$ for opposite-sign pairs. The conditional relative bias of $\widehat d$ thus factors exactly into two terms.
\begin{proposition}[Bias decomposition of the standard IBLT]\label{prop:projection}
In the standard IBLT, for any sign composition $(d_+,d_-)$,
\begin{equation}
\frac{\E[\widehat d\mid F]}{d}-1=c_F\Bigl(\frac{b^2}{d}-1\Bigr),\qquad c_F=\frac{\E[H_{12}\mid F]-k^2/M}{k(1-k/M)}.\label{eq:projection}
\end{equation}
\end{proposition}
The proof is in Appendix~\ref{app:std}. Since $F$ is invariant under permutations of key labels, $\E[H_{ij}\mid F]$ is identical for all $i\ne j$. The formula uses differing keys $1$ and $2$ as a representative pair. The factor $c_F$ depends only on failure selection of the mapping and measures the excess of inter-edge overlap in the failure sample over its unconditional mean $k^2/M$. The factor $b^2/d-1$ depends only on the sign composition $(d_+,d_-)$. For a one-sided difference ($|b|=d$), every intersecting edge pair has sign product $+1$, so $b^2/d-1=d-1$. All excess overlap contributes with the same sign, fully amplifying its effect on $\widehat d$. For a balanced difference ($b=0$), there are $d/2$ fewer same-sign pairs than opposite-sign pairs, giving a relative bias of $-c_F$, whose magnitude is $1/(d-1)$ of the one-sided value.
\begin{figure}[!htbp]
\centering\includegraphics[width=\linewidth]{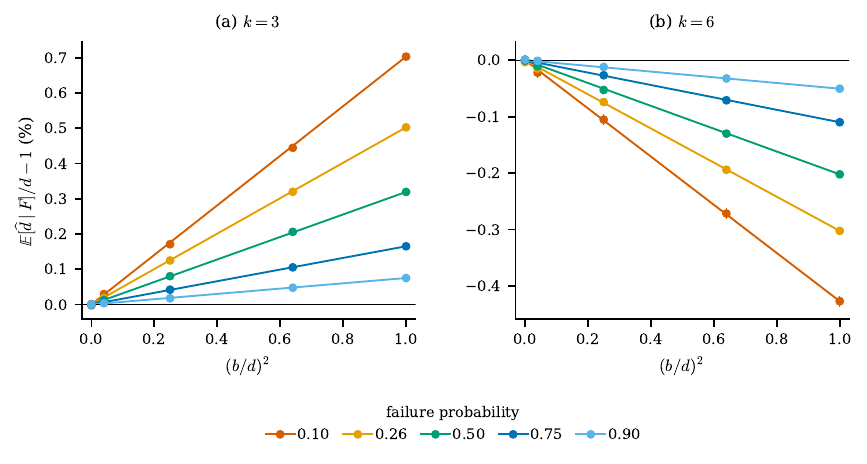}
\caption{Post-failure relative bias of the standard IBLT versus $(b/d)^2$ ($M=16384$). (a) $k=3$. (b) $k=6$. Each color denotes a failure probability. All configurations lie in the critical window: $d/M$ is 0.812--0.822 for $k=3$ and 0.633--0.639 for $k=6$. Points are measurements, and lines follow \eqref{eq:projection}.}
\label{fig:bias}
\end{figure}

Figure~\ref{fig:bias} shows this decomposition at $M=16384$ for five failure probabilities (0.10 to 0.90) each for $k=3$ and $k=6$. Within each configuration, all five sign compositions share the same batch of mappings and therefore identical peeling outcomes. Since $b^2/d-1=d(b/d)^2-1$, the relative bias at each failure probability is linear in $(b/d)^2$, with slope determined by $c_F$. We determine $c_F$ from the one-sided point. The remaining predictions of \eqref{eq:projection} differ from the measurements by at most twice the standard error of the difference. For $k=3$ at failure probability 0.10, the one-sided relative bias is $+0.70\%$, corresponding to a bias of about 0.64 standard deviations $s(d)$. As the failure probability increases, the failure sample approaches the full sample, weakening selection and reducing $|c_F|$. The sign of $c_F$ also depends on $k$. For $k=6$, $c_F<0$, so the one-sided failure sample is biased low, with relative bias $-0.43\%$ at failure probability 0.10 (Figure~\ref{fig:bias}(b)). The bias changes sign with $k$ and varies with sign composition according to \eqref{eq:projection}. A fixed correction therefore cannot offset it (Figure~\ref{fig:composition}).

In practice, the receiver does not know $(d_+,d_-)$ in advance. The difference table's count array alone cannot reveal whether a larger observed $\widehat d$ reflects a larger true $d$ or failure-induced bias for a one-sided difference. The upper-bound rule in \cite{wu2026selfsizing} therefore relies on simulation calibration by load interval outside deep overload, and its analytic bound is too conservative for small $M$. By \eqref{eq:projection}, the bias arises because the standard IBLT uses the set direction $s_i$ as the accumulation sign in count, with sign composition determining the factor that amplifies $c_F$. For a balanced difference, this factor is $-1$, leaving a relative bias of $-c_F$. This suggests encoding differing keys with independent, uniform random signs so that every cross term has zero mean for each mapping, while preserving peeling's ability to recover each key's original set direction. The post-failure estimator would then be unbiased for every sign composition. We develop this idea in the random-sign IBLT in Section~\ref{sec:design}.

\section{The random-sign IBLT}\label{sec:design}
We introduce the random-sign IBLT. The design multiplies each key by a random sign shared by both sides during insertion and deletion. The cell fields, field widths, and peeling process remain unchanged. We first give the encoding and decoding rules and prove ideal decoding equivalence with the standard IBLT and independence of the coding coefficients from the mapping. We then explain how the two-round protocol uses the lower bound $|b|$ to determine the second-round capacity.

\subsection{Encoding rules, decoding equivalence, and sign independence}
The two sides agree in advance on a sign hash $\varepsilon:\mathcal U\to\{-1,+1\}$ independent of the cell hash. Each sign is generated directly from the key and a shared seed, requiring no per-key storage. Algorithm~\ref{alg:encode} gives the complete encoding and decoding procedure. The lines marked $\star$ show all changes from the standard IBLT: applying the sign during insertion and deletion and recovering the direction on output. The number of cell fields, field widths, and peeling control flow remain unchanged. We describe encoding and decoding in turn.

\begin{algorithm}[!htbp]
\caption{Encoding and decoding of the random-sign IBLT. Lines marked $\star$ differ from the standard IBLT; field arithmetic is modulo $2^{64}$.}
\label{alg:encode}
\begin{algorithmic}[1]
\Procedure{Insert}{$x$}
\For{each cell $v$ of $x$}
\State $\mathrm{count}_v\gets\mathrm{count}_v+\varepsilon(x)$\Comment{$\star$ standard IBLT adds $1$}
\State $\mathrm{keySum}_v\gets\mathrm{keySum}_v+\varepsilon(x)\,x$\Comment{$\star$ standard IBLT adds $x$}
\State $\mathrm{checksum}_v\gets\mathrm{checksum}_v+\varepsilon(x)\,\mathrm{chk}(x)$\Comment{$\star$ standard IBLT adds $\mathrm{chk}(x)$}
\EndFor
\EndProcedure
\Procedure{Delete}{$x$}
\State as in \Call{Insert}{$x$}, subtract the same quantities from the three fields
\EndProcedure
\Procedure{Peel}{difference table}
\State add cells with count $\pm1$ to the candidate queue
\While{the candidate queue is nonempty}
\State remove cell $v$; if $C_v\ne\pm1$, process the next candidate
\State $x\gets C_v\,\mathrm{keySum}_v$
\If{$\mathrm{checksum}_v=C_v\,\mathrm{chk}(x)$ and the cells of $x$ include $v$}
\State output key $x$ and direction $C_v\,\varepsilon(x)$\Comment{$\star$ standard IBLT outputs $C_v$}
\State if the direction is $+1$, call \Call{Delete}{$x$}, otherwise call \Call{Insert}{$x$}
\State add cells whose count becomes $\pm1$ to the candidate queue
\EndIf
\EndWhile
\State \Return whether all three fields of every cell are zero
\EndProcedure
\end{algorithmic}
\end{algorithm}

Encoding follows lines 1--7 of Algorithm~\ref{alg:encode}. To insert a key $x$, we add $\varepsilon(x)$, $\varepsilon(x)\,x$, and $\varepsilon(x)\,\mathrm{chk}(x)$ to the three fields of each of its $k$ cells. Deletion subtracts the same quantities. We accumulate keySum and checksum modulo $2^{64}$. In the set model, these two fields also support the bitwise XOR implementation commonly used in the standard IBLT (Appendix~\ref{app:xor}). Our analysis and experiments use the additive form. Subtracting the two tables cell by cell cancels the identical contributions of each shared key in all three fields. A differing key $i$ has coding coefficient $s_i\varepsilon_i$ in the difference table, where $\varepsilon_i=\varepsilon(x_i)$. Thus, cell $v$ in the difference table has fields
\begin{equation}
C_v=\sum_{i:v\in S_i}s_i\varepsilon_i,\qquad
\mathrm{keySum}_v=\sum_{i:v\in S_i}s_i\varepsilon_i\,x_i,\qquad
\mathrm{checksum}_v=\sum_{i:v\in S_i}s_i\varepsilon_i\,\mathrm{chk}(x_i).\label{eq:fields}
\end{equation}
We treat $\varepsilon_i$ as iid uniform signs independent of the mapping $\mathcal G$.

Decoding follows lines 8--17 and uses the standard IBLT peeling process. A pure cell $v$ containing a single unrecovered differing key $x_i$ has count $C_v=s_i\varepsilon_i\in\{-1,+1\}$. Multiplying $\mathrm{keySum}_v$ by $C_v$ cancels the sign $s_i\varepsilon_i$ and yields the candidate key $x_i$ (line 12). Once the candidate passes the checksum and cell-membership checks (line 13), we multiply $C_v$ by $\varepsilon(x_i)$. Since $\varepsilon_i^2=1$, this recovers the original set direction $s_i=C_v\,\varepsilon(x_i)$ (line 14). We then call \textsc{Delete} or \textsc{Insert} according to the recovered direction $s_i$ (line 15). Both operations apply $\varepsilon(x_i)$ again in all $k$ cells of $x_i$. The contribution subtracted from each field is therefore exactly the key's contribution multiplied by its coding coefficient $s_i\varepsilon_i=C_v$. After each peeling step, the three fields still equal the encoded sum of the unrecovered keys. As in the standard IBLT, peeling succeeds when all three fields of every cell are zero (line 17).

Figure~\ref{fig:structure} compares the standard and random-sign IBLTs under the same mapping in a small one-sided example ($s_1=\dots=s_4=+1$). Both encodings have the same pure cells: the degree-$1$ cells 1, 4, 7, and 8, shaded in the figure. All have counts in $\{\pm1\}$. In pure cell 4, for example, the random sign gives $C_4=-1$. Multiplying by $\varepsilon_2=-1$ recovers the true direction of $x_2$ as $s_2=C_4\varepsilon_2=+1$. The two encodings follow identical peeling processes. Their sums of squared counts differ substantially. The single-key diagonal terms sum to $kd=12$ in both encodings. In the shared cells (cells 2, 3, 5, and 6), the standard IBLT accumulates contributions with the same sign. Each shared cell has count $2$, all cross terms are positive, and $\sum_vC_v^2=20$. In the random-sign IBLT, contributions along solid lines ($\varepsilon_i=+1$) cancel those along dashed lines ($\varepsilon_i=-1$) in the shared cells. Each shared cell in this example contains exactly one positive and one negative key. All cross terms are negative, the counts cancel to $0$, and $\sum_vC_v^2=4$. The next two propositions generalize the decoding equivalence and sign independence illustrated in Figure~\ref{fig:structure}.

\begin{figure}[!htbp]
\centering\includegraphics[width=.92\linewidth]{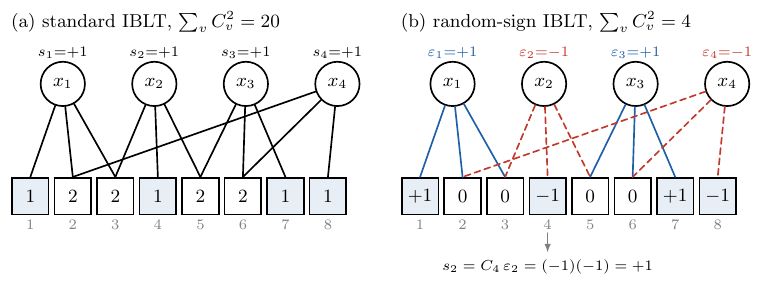}
\caption{The standard IBLT and the random-sign IBLT under the same mapping ($k=3$, $M=8$, one-sided difference $d=4$). The shaded cells have count $\pm1$; in (b), solid and dashed lines correspond to $\varepsilon_i=+1$ and $\varepsilon_i=-1$, respectively.}
\label{fig:structure}
\end{figure}

Random signs allow a degree-three cell with coding coefficients $(+1,+1,-1)$ to have count $+1$. This cannot occur in the standard IBLT under a one-sided difference. The candidate key read from such a cell fails the checksum and cell-membership checks, so peeling skips the cell until it contains a single key. Sign cancellation thus preserves correct key deletion and the stopping condition. Proposition~\ref{prop:decode} formalizes this property.

\begin{proposition}[Ideal decoding equivalence]\label{prop:decode}
Assume that neither the pure-cell test nor the termination check makes a false decision (the checksum error event $E$ does not occur) and that field arithmetic does not overflow. At every intermediate state, the random-sign IBLT can remove exactly the keys removable by unsigned degree-one peeling. Under the same cell hash, the random-sign and standard IBLTs therefore leave the same 2-core. The failure event $F$ depends only on the unsigned mapping $\mathcal G$. The decoding threshold, failure probability, and recovered result are identical to those of the standard IBLT.
\end{proposition}

\begin{proof}
We induct on the number of peeling steps. By the field invariant, every degree-one cell has count $\pm1$ and passes the check. Ideal purity excludes cells of any other degree, so the two processes can follow the same cell-selection order. The full proof is in Appendix~\ref{app:proof-decode}.
\end{proof}

In an implementation, the finite 64-bit checksum gives a nonzero false-acceptance probability. Under an ideal hash, each check of a non-pure cell falsely accepts with probability at most $2^{-64}$. Section~\ref{sec:protocol} accounts for this term separately in the risk upper bound for the two-round protocol. Section~\ref{sec:exp-decode} verifies the implemented decoding equivalence in five hundred million paired trials with real IBLTs. Table~\ref{tab:decode-paired} in the appendix summarizes success-failure agreement and recovery checks for each configuration. Section~\ref{sec:exp-cost} reports decoding time.

Random signs retain their distribution after conditioning on failure. Failure changes the overlap distribution of the mapping, so conditioning usually changes the distribution of random variables associated with the mapping. This selection causes the standard IBLT's bias to depend on sign composition. By Proposition~\ref{prop:decode}, failure depends only on the mapping. The signs come from a sign hash independent of the cell hash and remain iid uniform after we fix the mapping. Proposition~\ref{prop:composition} states this property precisely.

\begin{proposition}[Signs are independent of the mapping]\label{prop:composition}
For any fixed direction vector $(s_i)$, the coding coefficients $(s_i\varepsilon_i)$ are iid uniform signs independent of the mapping $\mathcal G$. Thus, failure $F$ is independent of the coding coefficients. Given $F$, or even the entire $\mathcal G$, the coding coefficients remain iid uniform signs. Statistics constructed from $(C_v)$, $\mathcal G$, and $F$ have the same joint distribution for all sign compositions $(d_+,d_-)$.
\end{proposition}

\begin{proof}
For fixed $(s_i)$, the map $(\varepsilon_i)\mapsto(s_i\varepsilon_i)$ is a bijection on $\{\pm1\}^d$ that preserves the uniform distribution. Since $(\varepsilon_i)$ is independent of $\mathcal G$, $(s_i\varepsilon_i)$ is an iid uniform sign sequence independent of $\mathcal G$. The joint distribution of $(s_i\varepsilon_i,\mathcal G)$ therefore does not depend on $(s_i)$. By Proposition~\ref{prop:decode}, $F$ is a function of $\mathcal G$, so conditioning on $F$ or $\mathcal G$ preserves the distribution of the coding coefficients. By \eqref{eq:fields}, $(C_v)$ is determined by $(s_i\varepsilon_i)$ and $\mathcal G$. Statistics constructed from $(C_v)$, $\mathcal G$, and $F$ thus have a joint distribution independent of $(s_i)$ and identical for all $(d_+,d_-)$.
\end{proof}

Section~\ref{sec:exp-e2e} and Table~\ref{tab:vldb} in the appendix illustrate this property numerically in the protocol. For a fixed configuration and calibration factor, random-sign encoding gives similar conditional miss rates under one-sided and balanced differences. Figure~\ref{fig:encoding} compares the post-failure distributions of the standard and random-sign encodings.

Propositions~\ref{prop:decode} and \ref{prop:composition} underpin our analysis. Decoding success depends only on $\mathcal G$, and the estimator's randomness comes from the independent signs. We can therefore analyze the post-failure distribution by first fixing the mapping and then studying the quadratic form in iid signs on the overlap graph. The results hold for all sign compositions. The joint distribution of the coding coefficients $s_i\varepsilon_i$ is independent of the direction vector $(s_i)$, so we can treat $s_i\varepsilon_i$ and $\varepsilon_i$ identically in the analysis. We write both as $\varepsilon_i$ below.

\subsection{\texorpdfstring{The two-round protocol flow and the lower bound $|b|$}{The two-round protocol flow and the lower bound |b|}}\label{sec:flow}
In the two-round protocol, the receiver subtracts the tables cell by cell and runs peeling on the difference table. For the standard IBLT, the count sum $\frac1k\sum_vC_v=b=d_+-d_-$ directly gives the deterministic lower bound $d\ge|b|$. For the random-sign IBLT, $\frac1k\sum_vC_v=\sum_i\varepsilon_i$ depends only on the random signs and no longer encodes the set-size difference $b$. The sender therefore attaches the set size $|A|$ (8 bytes) to the first-round table, allowing the receiver to compute $b=|A|-|B|$ directly.

If first-round peeling succeeds, the protocol ends immediately. After a first-round failure, the receiver computes an upper bound $U$ from the raw counts $(C_v)$ of the difference table. The computation uses only the two scalars $\sum_vC_v^2$ and $\sum_vC_v$, which the receiver accumulates while subtracting cells one by one, and peeling can modify the difference table in place. It incorporates the known lower bound $|b|$ by taking $U_b=\max(U,|b|)$ (for a one-sided difference, $|b|=d$). The second-round table has size $m(U_b)=\lceil U_b/(0.8\alpha_c)\rceil$. We re-encode and decode using a fresh cell hash independent of the first round. Whenever $U_b\ge d$, the second-round load is at most $0.8\alpha_c$, leaving a margin below the peeling threshold. The second round therefore succeeds with high probability. The analysis applies to any constant target load below $\alpha_c$; we use 0.8. Sections \ref{sec:measurement} and \ref{sec:inversion} analyze the post-failure distribution of the sum of squared counts and give the rule for computing $U$.

\section{Unbiased estimation after failure and the Gaussian upper bound}\label{sec:measurement}
We use the independence of signs and mapping established in Section~\ref{sec:design} to analyze $\widehat d$ after decoding failure and construct an upper bound for $d$. Cancellation between random signs of different keys makes $\widehat d$ exactly unbiased for every fixed mapping (Proposition~\ref{prop:moments}). It therefore remains unbiased even after arbitrarily rare failures. Controlling the conditional miss rate at the target risk $u$ also requires the distribution of post-failure fluctuations of $\widehat d$ around $d$. When $d\gg\sqrt M$, the standardized estimation error converges to the standard Gaussian distribution (Theorem~\ref{thm:clt}), yielding the Gaussian rule (Definition~\ref{def:gauss}).

\subsection{Conditional unbiasedness and conditional variance given the mapping}
We expand the sum of squares and the total sum in~\eqref{eq:dhat} by key. Each differing key occupies $k$ cells and satisfies $\varepsilon_i^2=1$, so the diagonal contributions of the two terms sum to exactly $d$. The estimation error consists of the cross terms between different keys:
\begin{equation}
\widehat d-d=\frac{2}{k(1-k/M)}\sum_{i<j}\varepsilon_i\varepsilon_j\Bigl(H_{ij}-\frac{k^2}{M}\Bigr).\label{eq:rs-expand}
\end{equation}
The right-hand side is a quadratic form in the random signs. Each coefficient $H_{ij}-k^2/M$ measures the overlap relative to its unconditional mean and depends entirely on the mapping $\mathcal G$. For the standard IBLT, the same expansion replaces the sign products with fixed values $s_is_j$. Selection of overlaps by failure then translates directly into bias (Proposition~\ref{prop:projection}). With random signs, each $\varepsilon_i\varepsilon_j$ takes $\pm1$ with equal probability even on a fixed mapping where failure occurs (Proposition~\ref{prop:composition}). The cross terms therefore have mean zero on every mapping. Intuitively, failure selects the mapping and hence the cross-term coefficients, while each sign product remains equally likely to be positive or negative. The error remains centered at zero, while its conditional variance can change. Proposition~\ref{prop:moments} expresses this intuition through the conditional mean and variance given the mapping.

\begin{proposition}[Conditional unbiasedness and conditional moments given the mapping]\label{prop:moments}
For every fixed mapping $\mathcal G$,
\begin{equation}
\E[\widehat d\mid\mathcal G]=d,\qquad
\Var(\widehat d\mid\mathcal G)=\frac{4}{k^2(1-k/M)^2}\sum_{i<j}\Bigl(H_{ij}-\frac{k^2}{M}\Bigr)^2\le2d(d-1).\label{eq:condmoments}
\end{equation}
Thus, for any mapping event $\mathcal S$ of positive probability, $\E[\widehat d\mid\mathcal S]=d$. In particular, $\E[\widehat d\mid F]=d$ holds for every finite $M$ and every sign composition, even for arbitrarily rare failures. Averaging over the random mapping $\mathcal G$ gives $\E\widehat d=d$ and $\Var\widehat d=s(d)^2$.
\end{proposition}

\begin{proof}[Proof sketch]
Given $\mathcal G$, $(\varepsilon_i)$ are iid uniform signs. Taking expectations term by term in~\eqref{eq:rs-expand} and using $\E[\varepsilon_i\varepsilon_j\mid\mathcal G]=0$ proves conditional unbiasedness. The sign products for distinct unordered pairs $(i,j)$ are pairwise uncorrelated, yielding the conditional variance in~\eqref{eq:condmoments}. The Cauchy--Schwarz inequality gives the upper bound $2d(d-1)$. The full proof is in Appendix~\ref{app:proof-moments}.
\end{proof}

For a fixed mapping, $\widehat d$ is a Rademacher quadratic form with the Gram matrix of the centered incidence vectors. Its conditional mean equals the trace, and its conditional variance comes only from the off-diagonal terms. These are the same identities used by the Hutchinson trace estimator\cite{hutchinson1990trace}. Proposition~\ref{prop:moments} expresses them in terms of overlap coefficients, establishing conditional unbiasedness for any mapping event.

Figure~\ref{fig:composition} illustrates this unbiasedness numerically. The post-failure bias of the standard encoding changes direction with $k$, while the mean for the random-sign encoding stays near zero. Section~\ref{sec:exp-decode} and Appendix~\ref{app:dist} compare conditional moments across additional configurations.

\subsection{\texorpdfstring{The Gaussian limit after failure when $d\gg\sqrt M$}{The Gaussian limit after failure when d >> sqrt(M)}}
Proposition~\ref{prop:moments} establishes the center of the estimator after failure. Constructing an upper bound also requires the shape of the error distribution around this center. Failure affects $\widehat d$ only through the mapping. By~\eqref{eq:condmoments}, the sum of squared overlaps $\sum_{i<j}(H_{ij}-k^2/M)^2$ determines the conditional error scale. Failure selects mappings with different values of this sum and may therefore change the error distribution. We standardize the error by the unconditional standard deviation $s(d)$:
\begin{equation}
X=\frac{\widehat d-d}{s(d)},\label{eq:X}
\end{equation}
and study the conditional distribution of $X$ given $F$.

We first fix a mapping. By~\eqref{eq:rs-expand}, the error is a weighted sum of sign products $\varepsilon_i\varepsilon_j$. The main contributions come from key pairs sharing a cell: the intersecting edge pairs with $H_{ij}\ge1$ in $\mathcal G$. Their expected number is about $k^2d^2/(2M)$. Sign products for distinct key pairs are pairwise uncorrelated, and no small set of terms dominates the sum. The number of intersecting edge pairs grows with $d^2/M$ and tends to infinity when $d\gg\sqrt M$. The error then sums many $\pm1$ perturbations, and its distribution approaches a Gaussian. When $d$ is of the same order as $\sqrt M$, only finitely many intersecting edge pairs contribute. The error consists of a few $\pm1$ terms, and its distribution remains discrete. We distinguish these regimes by the scale parameter
\begin{equation}
c=\frac{d}{\sqrt M}.\label{eq:c}
\end{equation}
The expected number of intersecting edge pairs is about $k^2c^2/2$, so $c$ plays the role of the collision scale in the birthday paradox\cite{arratia1989chen,dw}. It and the load $d/M$ are independent parameters.

Conditioning on failure preserves the sign structure. On any mapping where failure occurs, $(\varepsilon_i)$ remain iid uniform signs (Proposition~\ref{prop:composition}). The Gaussian shape persists as long as the mapping still contains many intersecting edge pairs. This also holds for rare failures. Under a subcritical load, a typical failure is caused by a duplicate edge pair. This raises a single $H_{ij}$ to $k$ and leaves the overall structure of intersecting edge pairs almost unchanged.

\begin{theorem}[Gaussian limit after failure]\label{thm:clt}
Fix $k\ge3$ and let $d/\sqrt M\to\infty$, allowing any limiting behavior of $d/M$. Then
\begin{equation}
\sup_x\bigl|\Pr(X\le x\mid F)-\Phi(x)\bigr|\longrightarrow0.\label{eq:clt}
\end{equation}
This includes the critical window, rare failures under a fixed subcritical load, and arbitrary overload.
\end{theorem}

\begin{proof}[Proof sketch]
The proof has two steps. First, we fix the mapping and take randomness only over the signs. Suppose the number of intersecting edge pairs is of the same order as $M$, the ratio of the sum of squared overlaps $\Omega_2=\sum_{i<j}H_{ij}^2$ to its exact mean tends to one, and each key participates in overlaps of size at most $O(\log M)$. The higher cumulants of the standardized quadratic form~\eqref{eq:rs-expand} then tend to zero, and the conditional distribution of $X$ converges to $\mathcal N(0,1)$ uniformly over such mappings. Second, we show that, conditional on $F$, the mapping satisfies these conditions with probability tending to $1$. We consider three cases according to the limit of $d/M$. When $d/M$ has a positive finite limit, $\Pr(F)\ge\binom Mk^{-1}$. The probability that a mapping violates the conditions is superpolynomially small and remains negligible after conditioning. When $d=o(M)$, Proposition~\ref{prop:plant} shows that the mapping distribution given failure is asymptotically identical to the planted distribution in TV distance. The Gaussian limit holds for the planted distribution. When $d>M$, failure occurs with certainty, so the conditional and unconditional distributions coincide. The full proof is in Appendix~\ref{app:clt}.
\end{proof}

Figure~\ref{fig:encoding} compares the distributions of $X$ for failure samples and all samples in a configuration with $c\approx104$. For the standard encoding, we compute $X$ in the same way from its counts $C_v=\sum_{i:v\in S_i}s_i$. With random-sign encoding, the two distributions almost coincide. With standard encoding, the failure distribution is shifted as a whole.

\subsection{The Gaussian rule and its conditional miss rate}\label{sec:gauss}
By Theorem~\ref{thm:clt}, $\widehat d$ is approximately Gaussian after failure, with mean $d$ and standard deviation $s(d)$. If $d$ were known, we could set the margin using the fact that $\widehat d$ falls below $d+s(d)\Phi^{-1}(u)$ with probability about $u$. However, the scale $s(d)$ (\eqref{eq:s}) depends on the unknown $d$. Direct substitution of $s(\widehat d)$ would make the margin fluctuate with the observation. A low $\widehat d$ gives a small $s(\widehat d)$, reducing the margin precisely for samples that need a larger one. The risk analysis would then need to account for this additional error.

Test inversion is a classical method for constructing confidence sets\cite{neyman1937}. We reverse the reasoning and test candidate values individually. For each candidate $n$, we assume the true value is $n$. The scale $s(n)$ is then known, and $\widehat d$ falls below $n+s(n)\Phi^{-1}(u)$ with probability about $u$. An observation below this lower bound indicates that $n$ is too large, so we exclude $n$. The largest remaining candidate gives the upper bound $U$. The risk is controlled by a single event: $U<d$ implies that the true value $d$ is excluded. Under the true value $d$, this exclusion has probability about $u$. Definition~\ref{def:gauss} follows this construction, replacing $\Phi^{-1}(u)$ with a more accurate quantile for finite $M$.

\begin{definition}[Gaussian rule]\label{def:gauss}
Given an observation $\widehat d$ and a target risk $u\in(0,1/2)$, we test each candidate integer $n\ge2$ as the true value, for which the scale $s(n)$ is known. Let $\chi^2_{M-1,u}$ denote the lower $u$ quantile of the chi-square distribution with $M-1$ degrees of freedom. Set the acceptance lower bound to
\begin{equation}
t(n)=n+s(n)\,\frac{\chi^2_{M-1,u}-(M-1)}{\sqrt{2(M-1)}}.\label{eq:gauss}
\end{equation}
We accept candidate $n$ if $\widehat d\ge t(n)$ and reject it otherwise. We set $U$ to the largest accepted candidate, returning $+\infty$ if none is accepted.
\end{definition}

For fixed $M,k$ as $d\to\infty$, $\widehat d/d$ converges in distribution to $\chi^2(M-1)/(M-1)$\cite{wu2026selfsizing}, which has mean 1. After centering the chi-square variable at its mean and scaling it to unit variance, its lower $u$ quantile is $(\chi^2_{M-1,u}-(M-1))/\sqrt{2(M-1)}$. Equation~\eqref{eq:gauss} uses this quantile in place of $\Phi^{-1}(u)$ to capture skewness more accurately for finite $M$. The two quantiles agree as $M\to\infty$. In the experimental range $M\ge64$, $u\ge0.01$, $t(n)$ is strictly increasing for $n\ge2$ (Appendix~\ref{app:gauss-mono}), so we can find $U$ by monotone binary search. Outside this range, we check candidates individually as specified in the definition. The rule depends only on the observation $\widehat d$ and the known parameters $(M,k,u)$.

\begin{corollary}[Conditional miss rate of the Gaussian rule]\label{cor:gauss}
Under the conditions of Theorem~\ref{thm:clt}, for any fixed $u\in(0,1/2)$, $\limsup\Pr(U<d\mid F)\le u$.
\end{corollary}

\begin{proof}[Proof sketch]
We take $n=d$. Theorem~\ref{thm:clt} and $(t(d)-d)/s(d)\to\Phi^{-1}(u)$ give a rejection probability of $u+o(1)$ for the true candidate. The conclusion follows from $\{U<d\}\subseteq\{\widehat d<t(d)\}$. The full proof is in Appendix~\ref{app:clt}.
\end{proof}

By Corollary~\ref{cor:gauss}, the Gaussian rule has conditional miss rate asymptotically at most $u$ throughout the range $d\gg\sqrt M$ of Theorem~\ref{thm:clt}. This range includes the critical window and the overload region. Figure~\ref{fig:coverage} shows that the Gaussian rule's miss/u approaches the target level as $c$ grows. When $d$ is of the same order as $\sqrt M$ or smaller, the Gaussian approximation no longer holds, and the actual conditional miss rate exceeds $u$. Section~\ref{sec:inversion} explains this behavior and gives the Skellam rule, whose conditional miss rate is asymptotically at most $u$ uniformly over all $d\ge2$.

\section{The Skellam rule for all loads}\label{sec:inversion}
The Gaussian rule's conditional miss rate guarantee requires $d\gg\sqrt M$. If we size the IBLT table for the expected load, $d$ can be far below $\sqrt M$. Failure is then rare, but the two-round protocol must still use the post-failure observation to size the second-round table. In this range, the conditional miss rate of the Gaussian rule can be several times the target $u$ (Table~\ref{tab:gauss-c} in the appendix).

The reason lies in the shape of the post-failure error at low load. Failure then has essentially one source: two differing keys select exactly the same $k$ cells. The duplicate edge pair itself forms a 2-core. Even after the remaining keys have been peeled, these two keys remain. Under random signs, the two keys are equally likely to have the same or opposite signs. If we place each key in the table alone, it leaves a count of $\pm1$ in each of its $k$ cells. Together, their squared counts sum to $2k$. Dividing by $k$ gives their exact contribution of $2$. With opposite signs, the counts in all $k$ shared cells cancel to $0$. The entire sum of squares of $2k$ disappears, leaving the estimate $2$ below the true value. With equal signs, each shared cell has count $\pm2$ and squared count $4$. The sum of squares becomes $4k$, leaving the estimate $2$ above the true value. This jump is equally likely to be upward or downward, and its size is independent of $M$. The remaining keys also occasionally share a cell, almost always exactly one. Each such overlap shifts the estimate by $2/k$, again with either direction equally likely. The number of such key pairs is approximately Poisson\cite{arratia1989chen} and remains finite when $c$ is bounded. The post-failure error thus consists of one large jump of size $2$ and several small jumps of size $2/k$.

The Gaussian approximation cannot capture these two types of jumps. When $d$ is small relative to $\sqrt M$, the standard deviation $s(d)$ of the Gaussian rule tends to zero. The downward jump of size $2$ remains unchanged and lies arbitrarily many standard deviations into the Gaussian tail. The probability of rejecting the true value therefore tends to one half. When $c$ is of order one, only finitely many small jumps occur. The error still has a lattice distribution, with a visible departure from the Gaussian in the lower tail corresponding to the target risk (Appendix~\ref{app:limit}). We construct a null distribution that retains both types of jumps and approaches the Gaussian used by the Gaussian rule as $c$ grows (Figure~\ref{fig:aux}). The resulting Skellam rule controls the conditional miss rate for all $d\ge2$. Appendices~\ref{app:inversion} and~\ref{app:search} give the full proof and numerical implementation.

\subsection{The lattice-preserving test statistic \texorpdfstring{$\widetilde d$}{d~}}\label{sec:raw}
To keep these jumps exactly on lattice points, we need a discrete test statistic. The correction term $(\sum_vC_v)^2/M$ in $\widehat d$ shifts its values off the lattice. We therefore use the normalized sum of squares without subtracting the mean term,
\begin{equation}
\widetilde d=\frac1k\sum_{v=1}^M C_v^2.\label{eq:raw-energy}
\end{equation}
Expanding the sum of squares gives $d$ from the diagonal terms, leaving only overlap contributions to the error:
\begin{equation}
\widetilde d-d=\frac2k\sum_{i<j}H_{ij}\,\varepsilon_i\varepsilon_j.\label{eq:raw-expand}
\end{equation}
Every term on the right is an integer multiple of $2/k$, so $\widetilde d-d$ takes only integer multiples of $2/k$. The expansion shows both types of jumps directly. A duplicate edge pair has $H_{ij}=k$ and contributes exactly $2\varepsilon_i\varepsilon_j$, or $\pm2$ with equal probability. An edge pair sharing a single cell has $H_{ij}=1$ and contributes $\pm2/k$.

Without the correction term, $\widetilde d$ remains exactly unbiased for every mapping.

\begin{proposition}[Conditional unbiasedness of $\widetilde d$]\label{prop:raw-unbiased}
For every fixed mapping $\mathcal G$, $\E[\widetilde d\mid\mathcal G]=d$. Thus, for any mapping event $\mathcal S$ with positive probability, $\E[\widetilde d\mid\mathcal S]=d$. In particular, $\E[\widetilde d\mid F]=d$ holds for all finite $M$ and all sign compositions, however rare the failure.
\end{proposition}

\begin{proof}
Given $\mathcal G$, $H_{ij}$ is constant. By Proposition~\ref{prop:composition}, the signs remain iid and uniform, so $\E[\varepsilon_i\varepsilon_j\mid\mathcal G]=0$ for $i\ne j$. Taking expectations term by term in~\eqref{eq:raw-expand} gives the result.
\end{proof}

This property also follows from the random signs. In the standard IBLT, the cross terms have fixed sign products $s_is_j$. The same expansion gives the unconditional mean of the sum of squares as $d+\frac kM(b^2-d)$, which depends on the unknown $b$. The correction term in $\widehat d$ removes this dependence. With random signs, the cross terms have mean zero for every mapping. The sum of squares is therefore unbiased without the correction term.

By \eqref{eq:raw-decomp} in Appendix~\ref{app:ams}, $\widetilde d$ is a weighted sum of $\widehat d$ and $(\frac1k\sum_vC_v)^2$, with weights $1-k/M$ and $k/M$, respectively. The latter depends only on the random signs and is the classical AMS second-moment estimator\cite{alon1999spacecomplexity}. After standardization by $s(d)$, the difference tends to zero, so Theorem~\ref{thm:clt} also holds for $\widetilde d$ (Appendix~\ref{app:clt}).

\subsection{Constructing the null distribution from the two jumps}\label{sec:aux}
For each candidate $n$, we construct a reference distribution for $\widetilde d$ after failure when the true value is $n$. This is the null distribution, which we use in Section~\ref{sec:skellam} to test each candidate. We construct it by counting the occurrences of the two types of jumps.

The large jump occurs exactly once and contributes $2\varepsilon_1\varepsilon_2$. Its direction is determined by the sign product of the duplicate edge pair.

The small jumps come from the $n-2$ keys outside the duplicate edge pair. They form $\binom{n-2}2$ pairs. Each pair shares a cell with probability approximately $k^2/M$, producing a small jump. Upward and downward jumps are equally likely, so each has expected count
\begin{equation}
\lambda_n=\frac12\binom{n-2}2\frac{k^2}M=\frac{k^2(n-2)(n-3)}{4M}.\label{eq:lambda}
\end{equation}
Key pairs sharing a cell are sparse. By Poisson splitting, the upward and downward counts are approximately independent $\Pois(\lambda_n)$ counts. The upward count minus the downward count gives the net number of steps and follows a Skellam distribution, denoted $Z_n\sim\Skellam(\lambda_n,\lambda_n)$. Each step has size $2/k$.

Adding both types of jumps to the candidate $n$ gives the null distribution.
\begin{definition}[null distribution]\label{def:null}
For $\lambda\ge0$, let
\begin{equation}
J_\lambda=2\varepsilon_1\varepsilon_2+\frac2kZ,\qquad Z\sim\Skellam(\lambda,\lambda),\label{eq:J}
\end{equation}
where $\varepsilon_1\varepsilon_2$ takes the values $\pm1$ with equal probability and is independent of $Z$. We call $2\varepsilon_1\varepsilon_2$ the \textit{$\pm2$ term} and $\frac2kZ$ the \textit{Skellam term}. For a candidate integer $n\ge2$, we take $\lambda_n$ from~\eqref{eq:lambda}. The null distribution for candidate $n$ is the distribution of
\begin{equation}
n+J_{\lambda_n}=n+2\varepsilon_1\varepsilon_2+\frac2kZ_n\label{eq:aux}
\end{equation}
with CDF denoted by $G_n$.
\end{definition}
In Section~\ref{sec:uniform}, we prove that the post-failure distribution of $\widetilde d$ approaches $G_d$ uniformly over all $d\ge2$ (Theorems~\ref{thm:cross} and~\ref{thm:uniform}).

When $n=2,3$, fewer than two keys remain, so $\lambda_n=0$ and the null distribution assigns probability one half to each of $n\pm2$. In general, its mean is $n$. Since $\Var Z_n=2\lambda_n$, its variance is $4+2(n-2)(n-3)/M$.

\begin{figure}[!htbp]
\centering\includegraphics[width=\linewidth]{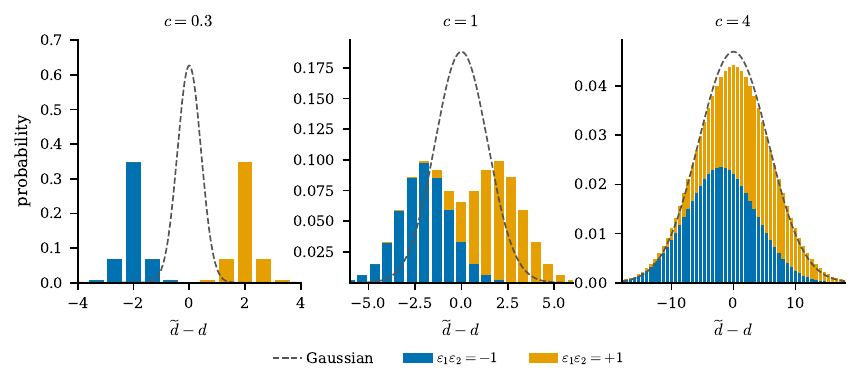}
\caption{Null distributions at $c=d/\sqrt M=0.3,1,4$, computed from the formula in the limit $M\to\infty$ ($k=3$, $\lambda_d=k^2c^2/4$). Blue and orange correspond to opposite and equal signs of the duplicate edge pair. The dashed line is the Gaussian used by the Gaussian rule.}
\label{fig:aux}
\end{figure}

The shape of this family explains why it applies at all loads. Figure~\ref{fig:aux} shows how the null distribution varies with $c$. When $c$ is small, the Skellam term is almost zero and the distribution is concentrated at $\pm2$. The Gaussian used by the Gaussian rule concentrates between these points, which lie beyond its tails. When $c\approx1$, the opposite-sign and equal-sign cases each spread into a lattice distribution with spacing $2/k$. These distributions overlap, but their combined shape still differs visibly from the Gaussian. As $c$ grows further, the width of the Skellam term far exceeds $2$, and the $\pm2$ term becomes a small contribution. The null distribution then approaches the Gaussian used by the Gaussian rule. This family thus connects the small-$c$ region to the region where the Gaussian rule applies. Section~\ref{sec:uniform} gives the rigorous approximation guarantee.

\subsection{The null distribution approximates the post-failure distribution at all loads}\label{sec:uniform}
Definition~\ref{def:null} constructs the null distribution from a counting argument. To use it to test candidates, we establish that it describes the post-failure distribution of $\widetilde d$ for all $d$. The receiver does not know the load region containing $d$ and cannot use the load to select a rule. The same family $G_d$ must therefore apply for both small and large $d$. We analyze these two cases and then combine them.

When $d$ is of the same order as $\sqrt M$ or smaller, overlaps are sparse. The model underlying $G_d$ consists of a duplicate edge pair and scattered single-cell overlaps among the remaining keys. The actual mapping may have additional structure: the duplicate edge pair may share a cell with a remaining key, two keys may share more than one cell, or one key may share cells with two other keys. Each case is much rarer than a single-cell overlap, and its probability tends to zero as $M$ grows. After excluding these cases, the actual mapping has exactly the structure underlying $G_d$. The actual distribution is therefore close to $G_d$ at every lattice point (Theorem~\ref{thm:cross}), as illustrated by the first two shapes in Figure~\ref{fig:aux}.

When $d\gg\sqrt M$, intersecting edge pairs are numerous and interwoven, so the sparse structure above no longer applies. The null distribution $G_d$ still applies because both distributions approach the same Gaussian. The statistic $\widetilde d$ is a sum of many small terms, each with equally likely signs, and approaches a Gaussian by Theorem~\ref{thm:clt}. The null distribution $G_d$ also approaches a Gaussian, as the third shape in Figure~\ref{fig:aux} shows, with width matching $s(d)$. The two distributions differ in detail but share the same Gaussian limit. Theorem~\ref{thm:uniform} combines these cases to give a uniform approximation over all $d\ge2$.

\begin{theorem}[Discrete limit distribution when $c$ is bounded]\label{thm:cross}
Fix $k\ge3$ and a finite constant $C$.
\begin{enumerate}[label=(\roman*)]
\item For $2\le d\le C\sqrt M$, the total variation distance between the conditional distribution of $\widetilde d$ given $F$ and the null distribution $G_d$ tends to zero uniformly and is at most
\begin{equation}
O_k\Bigl(\frac dM+\frac{d^3}{M^2}\Bigr).\label{eq:tv-error}
\end{equation}
\item If $d/\sqrt M\to c\in(0,\infty)$, the conditional distribution of $\widetilde d-d$ given $F$ converges in total variation to the distribution of $J_{k^2c^2/4}$, whose variance is $4+2c^2$.
\item The limiting distributions in (ii) satisfy
\begin{equation*}
J_{k^2c^2/4}\Rightarrow2\varepsilon_1\varepsilon_2\quad(c\to0),\qquad\frac{J_{k^2c^2/4}}{\sqrt2\,c}\Rightarrow\mathcal N(0,1)\quad(c\to\infty).
\end{equation*}
\end{enumerate}
\end{theorem}

Figure~\ref{fig:crosscdf} compares the lower tail of this discrete limit distribution with the Gaussian approximation, using planted samples with $c$ ranging from 0.25 to 8. Section~\ref{sec:exp-dist} also reports the distance between the empirical and limit distributions at fixed $c$ as $M$ grows.

\begin{proof}[Proof sketch]
We prove (i) in two steps. First, we replace the actual failure with a duplicate edge pair. By Proposition~\ref{prop:plant}, we can replace the mapping conditional on failure with the planted distribution, incurring error $O_k(d/M^2+M^{1-k})$ in the range of the theorem (Lemma~\ref{lem:plant}). Second, we exclude the three cases above from the planted mapping. Their probabilities are $O_k(d/M)$, $O_k(d^2/M^2)$, and $O_k(d^3/M^2)$, respectively. The remaining structure consists of a duplicate edge pair and mutually disjoint single-cell overlaps. The duplicate edge pair gives the $\pm2$ term. By the Chen--Stein local dependence bound\cite[Theorem~1]{arratia1989chen}, the overlap count is approximately $\Pois(2\lambda_d)$. Independent signs split this count into two independent $\Pois(\lambda_d)$ counts. Their difference gives the Skellam term, so $\widetilde d-d$ is close to $J_{\lambda_d}$. The Poisson approximation has error $O_k((d^2+d^3)/M^2)$. Since $2\le d\le M$ and $k\ge3$, each of $d/M^2$, $d^2/M^2$, and $M^{1-k}$ is at most $d/M$. Also, $(d^2+d^3)/M^2$ is at most $2d^3/M^2$. Combining these bounds gives~\eqref{eq:tv-error}.
(ii) By (i), it suffices to show that $J_{\lambda_d}$ converges in total variation to $J_{k^2c^2/4}$. When $d/\sqrt M\to c$, we have $\lambda_d\to k^2c^2/4$. The required convergence follows from continuity in the Poisson parameter. Independence of the two terms gives $\Var J_\lambda=4+8\lambda/k^2$. Substituting $\lambda=k^2c^2/4$ gives $4+2c^2$.
(iii) As $c\to0$, $\lambda\to0$ and the Skellam term tends to zero in probability, leaving only the $\pm2$ term of $J_\lambda$. As $c\to\infty$, $\Var J_{k^2c^2/4}\sim2c^2$ and the Skellam term satisfies the CLT. The $\pm2$ term divided by $\sqrt2\,c$ tends to zero. Appendix~\ref{app:proof-cross} gives the full argument.
\end{proof}

\begin{theorem}[Uniform CDF approximation]\label{thm:uniform}
Fix $k\ge3$. The Kolmogorov distance between the conditional distribution of $\widetilde d$ given $F$ and the null distribution $G_d$ satisfies
\begin{equation}
\sup_{d\ge2}\ \sup_{x\in\mathbb R}\bigl|\Pr(\widetilde d\le x\mid F)-G_d(x)\bigr|\longrightarrow0
\quad(M\to\infty).\label{eq:uniform}
\end{equation}
\end{theorem}

\begin{proof}[Proof sketch]
When $d\gg\sqrt M$, Gaussian limits let us handle a dense intersection graph without extending the matching structure. By Theorem~\ref{thm:clt}, the conditional distribution of $\widetilde d$, standardized by $s(d)$, approaches the standard Gaussian. The null distribution is that of $d+J_{\lambda_d}$. When $\lambda_d\to\infty$, the Skellam term of $J_{\lambda_d}$ satisfies the CLT. The ratio of $\Var J_{\lambda_d}=4+2(d-2)(d-3)/M$ to $s(d)^2$ tends to one, and the $\pm2$ term divided by $s(d)$ tends to zero. Both distributions approach the same continuous Gaussian CDF, so their Kolmogorov distance tends to zero. This argument applies to $d=o(M)$, a positive load ratio, and unbounded overload. We obtain uniformity by a subsequence argument. Take a subsequence along which $d/\sqrt M$ converges. Theorem~\ref{thm:cross} handles finite limits, and the Gaussian argument above handles infinite limits (Appendix~\ref{app:proof-uniform}).
\end{proof}

The same family of null distributions $G_d$ therefore describes the post-failure distribution of $\widetilde d$ for all $d\ge2$. For small $d$, the approximation holds pointwise on the lattice. For large $d$, both distributions approach the same Gaussian. We can therefore use this family without knowing the load region. In Section~\ref{sec:skellam}, we use it to construct the Skellam rule.

\subsection{The Skellam rule controls the conditional miss rate at all loads}\label{sec:skellam}
The Skellam rule uses the test inversion of the Gaussian rule (Section~\ref{sec:gauss}), replacing the Gaussian lower tail with the lattice lower tail of $G_n$. If the observation falls in the lower $u$ tail of $G_n$, candidate $n$ is too large to be the true value, so we reject candidate $n$. We take $U$ to be the largest accepted candidate.

\begin{definition}[Skellam rule]\label{def:skellam}
Let $\widetilde d$ be the observation and $u\in(0,1/2)$ the target risk. For each candidate integer $n\ge2$, we accept $n$ if $G_n(\widetilde d)>u$ and reject it otherwise. We define the upper bound as the supremum of all accepted candidates,
\begin{equation}
U=\sup\bigl\{n\ge2:\ G_n(\widetilde d)>u\bigr\},\label{eq:skellam-U}
\end{equation}
with $U=+\infty$ when the accepted candidates are unbounded.
\end{definition}

The null distribution for candidate $2$ takes only the values $0$ and $4$. Since $\widetilde d\ge0$, we always accept candidate $2$, so $U\ge2$. For practical table sizes, $U$ is also finite (for $u=0.01$, $M>18$ suffices). We can compute it exactly by a finite search, at a cost roughly proportional to $\lambda_d$ per computation (Appendix~\ref{app:search}).

The behavior of the Skellam rule at small and large loads explains why it applies throughout. Consider first small $c$. The Skellam term is then almost zero, the error after failure comes almost entirely from the duplicate edge pair, and the observation $\widetilde d$ equals $d-2$ or $d+2$ with probability one half each. The scale $s(n)$ of the Gaussian rule tends to zero in this region, so its upper bound $U$ stays close to the observation. When the duplicate edge pair has opposite signs, the observation is $d-2$, $U$ falls below $d$, and the conditional miss rate tends to one half. The null distribution $G_n$ of the Skellam rule retains this structure and assigns half its mass to $n-2$ and half to $n+2$. We accept candidate $n$ exactly when the smaller value $n-2$ does not exceed the observation, giving $U\approx\widetilde d+2$. The rule thus adds back a potentially lost large jump. With opposite signs, $U$ equals $d$ exactly; with equal signs, $U$ is about $d+4$. Neither case produces a miss. Consider next large $c$. The Skellam term then dominates and $G_n$ approaches the Gaussian used by the Gaussian rule, so the two rules give the same upper bound, and the Gaussian rule already controls the conditional miss rate in this region (Corollary~\ref{cor:gauss}).

Theorem~\ref{thm:coverage} gives the main result of this section. As the table size tends to infinity, the probability that the Skellam rule underestimates $d$ after failure is asymptotically at most the target risk $u$, uniformly over all $d\ge2$. This guarantee holds for any sign composition. The rule uses only $\sum_vC_v^2$ from the same table and the known parameters $(M,k,u)$. The Gaussian rule provides this guarantee only when $d\gg\sqrt M$ (Corollary~\ref{cor:gauss}). In Section~\ref{sec:protocol}, we apply this guarantee to the risk decomposition of the two-round protocol.

\begin{theorem}[Uniform asymptotic upper bound on the conditional miss rate]\label{thm:coverage}
Fix $k\ge3$ and $u\in(0,1/2)$. For any sign composition, the Skellam rule satisfies
\begin{equation}
\limsup_{M\to\infty}\ \sup_{d\ge2}\ \Pr(U<d\mid F)\le u.\label{eq:coverage}
\end{equation}
\end{theorem}

\begin{proof}[Proof sketch]
If the observation has distribution $G_n$ exactly, we reject $n$ when it falls in a lower tail with total probability at most $u$. By the definition of the CDF, the rejection probability is at most $u$, including for discrete distributions. The rejection region is a lower set, so the rejection probabilities under the actual distribution and $G_n$ differ by at most their Kolmogorov distance. A miss implies rejection of the true value: $\{U<d\}\subseteq\{G_d(\widetilde d)\le u\}$. Taking $n=d$ and applying Theorem~\ref{thm:uniform} gives~\eqref{eq:coverage}. By Proposition~\ref{prop:composition}, the joint distribution of the coding coefficients is independent of the sign composition. The result therefore holds for any sign composition. Appendix~\ref{app:proof-coverage} gives the full proof.
\end{proof}

Theorem~\ref{thm:coverage} is asymptotic in the table size. Two effects arise on finite tables. First, $G_n$ is a lattice distribution, so the rejection region contains only whole probability masses. The actual conditional miss rate can therefore be strictly below $u$. Second, the null distribution omits overlaps between the duplicate edge pair and the remaining keys. Each of the other $d-2$ keys has a cell among the $k$ cells of the pair with probability about $k^2/M$, so these overlaps occur at a total rate of about $(d-2)k^2/M$. They can push the conditional miss rate slightly above $u$ on small tables. Section~\ref{sec:exp-cover} evaluates finite-table behavior experimentally (Figures~\ref{fig:coverage} and~\ref{fig:confirm}). At $M=64$, the largest miss/u is 1.08, an excess beyond sampling error. All configurations with $M\ge256$ meet the target within sampling error. At fixed $c$, this occurrence rate is about $ck^2/\sqrt M$. The omitted overlaps become rarer as $M$ grows, and extra capacity decreases as well. When the first-round table must be small, the second-round sizing margin $1/(0.8\alpha_c)$ absorbs slight underestimation (Section~\ref{sec:protocol}). If more margin is needed, we can reduce the target load from 0.8, using a larger second-round table to lower the two-round failure rate.

\subsection{From the conditional miss rate to the failure probability of the two-round protocol}\label{sec:protocol}
We now connect the post-failure conditional miss rate to the overall success rate of two-round reconciliation. The protocol follows Section~\ref{sec:flow}. We denote first-round failure by $F_1$ and second-round failure by $F_2$. Let $E$ denote a checksum error: a non-pure cell passes the check, a key that passes the check lies outside the set difference, or all fields are zero at termination despite a nonempty core.

\begin{theorem}[Two-round risk decomposition]\label{thm:protocol}
Under the ideal hash model, let $\Pr_m(F)$ denote the probability of peeling failure for $d$ differing keys in a fresh table with $m$ cells. Then
\begin{equation}
\Pr\big((F_1\cap F_2)\cup E\big)
\le\Pr(F_1,\ U_b<d)+\Pr(F_1)\max_{m\ge m(d)}\Pr\nolimits_m(F)+\Pr(F_1,\ U_b=+\infty)+\Pr(E).\label{eq:risk}
\end{equation}
\end{theorem}

\begin{proof}
By the union bound, it suffices to bound $F_1\cap F_2$ and $E$ separately. The upper bound $U_b$ is a function of the first-round counts and the set sizes. The second-round mapping is independent of it. Under ideal purity, $F_2$ depends only on the second-round mapping and $d$. On $\{U_b<d\}$ and $\{U_b=+\infty\}$, we bound the conditional second-round failure probability by one. On $\{d\le U_b<+\infty\}$, we have $m(U_b)\ge m(d)$ and the conditional second-round failure probability is $\Pr_{m(U_b)}(F)\le\max_{m\ge m(d)}\Pr_m(F)$. Adding the parts gives~\eqref{eq:risk}.
\end{proof}

The first two of the four terms are the core risks. The first is a sizing error, the joint miss rate $\Pr(F_1)\Pr(U_b<d\mid F_1)$: for the Skellam rule it is asymptotically at most $u\Pr(F_1)$ uniformly (Theorem~\ref{thm:coverage}), and it is zero for a one-sided difference; the Gaussian rule has a larger conditional miss rate when $c$ is bounded, but failure itself is then rare. The second is a failure of second-round peeling with sufficient capacity; the second-round load is at most $0.8\alpha_c$, and this term decays polynomially in $d$ (Lemma~\ref{lem:propL}). The last two terms are marginal. The third is the absence of a finite upper bound: the Skellam rule never produces it when $M_1>4\log(1/u)$, and the Gaussian rule essentially only for $d=2$ or $3$. The fourth is a checksum error, of order $2(M_1+m+kd)\cdot2^{-64}$ under a 64-bit checksum. Appendix~\ref{app:risk} gives the derivation and numerical values of each term.

As an example, take $k=3$, $M_1=1024$, $d=838$ ($d/M_1\approx0.82$, corresponding to first-round load $d/(\alpha_cM_1)=1$ in Section~\ref{sec:exp}), $u=0.01$, and a balanced difference. In the protocol experiments (\exptag{PROTO}) of Section~\ref{sec:exp-e2e}, 19,187 of 26,667 trials fail in the first round, giving $\Pr(F_1)\approx0.72$. Substituting empirical frequencies into the terms of~\eqref{eq:risk}, the Skellam rule's first term, the joint miss rate, is about $6.6\times10^{-3}$ (177 misses). For the second term, we use the oracle's ($U=d$) measured final failure rate at second-round size $m(d)$, about $5.6\times10^{-4}$. The third term is zero, and the fourth is about $5\times10^{-16}$ by the formula in Appendix~\ref{app:risk}, giving a total of about $7\times10^{-3}$. The first term counts every miss as a two-round failure, whereas most misses are slight underestimates, and the second-round table sized at target load $0.8\alpha_c$ retains enough margin to succeed. At the same configuration, the measured two-round failure rates are $3.8\times10^{-4}$ for the Skellam rule, $7.5\times10^{-4}$ for the Gaussian rule, and $5.6\times10^{-4}$ for the oracle.

\section{Experimental evaluation}\label{sec:exp}
The experiments address three questions. First, does random-sign encoding remove the dependence of the post-failure estimate on sign composition while preserving decoding (Section~\ref{sec:exp-decode})? Second, what is the shape of the post-failure statistic's lower tail (Section~\ref{sec:exp-dist}), and how does this shape determine the conditional miss rate and extra capacity of the two rules at finite size (Section~\ref{sec:exp-cover})? Third, when the upper bound is used in a two-round protocol on a complete IBLT, what are the reliability and communication cost (Section~\ref{sec:exp-e2e}), what is the computational cost, and how do these results guide rule choice (Section~\ref{sec:exp-cost})?

The experiments fall into two categories. The first uses only a peeling sampler, without constructing a complete IBLT, to study the distribution of the post-failure statistic. Its main component is the statistical experiments (\exptag{STAT}), configured as in Section~\ref{sec:exp-setup}; the planted-sampling control (\exptag{PLANT}), conditional-unbiasedness test (\exptag{UNB}), and worst-configuration test (\exptag{WORST}) verify planted sampling, conditional unbiasedness, and the configurations with the highest risk, respectively. The second category uses complete IBLTs. The protocol experiments (\exptag{PROTO}) run the two-round protocol, configured as in Section~\ref{sec:exp-e2e}; the decoding-equivalence test (\exptag{DEC}) and computational-cost measurement (\exptag{COST}) verify decoding results and measure time and memory, respectively. Appendix~\ref{app:exp-overview} summarizes each experiment's purpose and scale. The main conclusions are measured on failure samples, and failure probabilities vary by orders of magnitude across configurations. Before launch, a single feasibility pilot estimates each configuration's failure probability. We use these estimates to set the trial counts needed to collect enough failures, targeting about 20,000 for the statistical and protocol experiments, and use planted sampling when direct sampling is too costly (Section~\ref{sec:exp-setup}). We then fix configurations, seeds, stopping rules, and sample sizes, with no adjustments during the formal runs. A few figures and tables use an additional single configuration to illustrate the distribution's shape, with parameters given in the caption.

\subsection{Configurations and sampling for the statistical experiments}\label{sec:exp-setup}
The statistical experiments use a sampler that independently generates $d$ uniform random hyperedges and runs peeling, without constructing a complete IBLT. On each mapping, we compute $X=(\widehat d-d)/s(d)$ for the standard encoding under one-sided and balanced differences, together with the exact integers $\sum_vC_v$ and $\sum_vC_v^2$ for the random-sign encoding. We then compute both rules' upper bounds from the same observation.

Each configuration is specified by $(k,M,d)$. We describe its position along two axes: $c=d/\sqrt M$ (\eqref{eq:c}) and the first-round load $d/(\alpha_cM)$. These axes control different effects. The parameter $c$ determines the statistical scale of the estimation error. For small $c$, the post-failure $\widehat d-d$ is dominated by the $\pm2$ offsets of a few duplicate edge pairs. For large $c$, it approaches a continuous distribution. The load determines the position relative to the peeling phase transition: subcritical below 0.95, near the critical load from 0.95 to 1.05, and overloaded above 1.05. For fixed $k,M$, the axes satisfy $d/(\alpha_cM)=c/(\alpha_c\sqrt M)$. Thus, changing $M$ is necessary to vary the load while keeping $c$ fixed and separating the two effects.

The statistical experiments use $k=3,4,5$ and $M\in\{256,1024,4096,16384\}$. For each pair $(k,M)$, $d$ ranges over the union of three sets: the 8 values corresponding to $c\in\{0.5,1,2,4,8,16,32,64\}$, the 5 values corresponding to load $\in\{0.95,1,1.05,1.5,3\}$, and the small differences $d=2,3,6$. Excluding loads above 3 leaves 189 configurations, with about 20,000 failures collected for each. Applying the same rule at $M=64$ gives 42 configurations, which we use separately to test the rules' limits of applicability. Figure~\ref{fig:exp-grid} in the appendix shows the configurations on both axes.

\paragraph{Planted sampling.}
We obtain failure samples by direct sampling or planted sampling, depending on the load. Near or above the critical load, failure is common. We generate random mappings and retain failures, obtaining samples from the true failure-conditioned distribution. At low load, the failure probability is small and direct sampling is too costly. By Proposition~\ref{prop:plant}, when $d=o(M)$, failure is almost always caused by a pair of keys with identical mappings. We can construct such samples directly: choose a random pair of difference keys to share the same set of $k$ cells, then draw the mappings and signs of the remaining keys independently. We use direct sampling for the 84 statistical configurations with high failure rates and planted sampling for the remaining 105. The figures distinguish the two sampling methods (Appendix~\ref{app:sampling-impl}).

The planted-sampling control (\exptag{PLANT}) compares the two sample types on 12 planted configurations where direct sampling remains affordable. Among directly sampled failures, the fraction caused by a duplicate edge pair ranges from 0.988 to 1.000. The KS distance for $X$ ranges from 0.0048 to 0.0079, comparable to sampling error at this sample size, and the Skellam rule's miss/u values also agree. The planted approximation deviates only near the critical load at small $M$. The planted configurations in the statistical experiments lie outside this region (Appendix~\ref{app:planted}).

\subsection{Random-sign encoding removes dependence on sign composition}\label{sec:exp-decode}
We verify two properties of random-sign encoding: the post-failure estimate is independent of sign composition, and decoding gives exactly the same result as the standard encoding.

\paragraph{Conditional bias.}
Under random-sign encoding, the post-failure $\widehat d$ is unbiased for every sign composition. Under the standard encoding, conditioning the mapping on failure introduces a bias in $\widehat d$ that depends on sign composition. Figure~\ref{fig:encoding} compares three conditions on the same batch of 4 million mappings, with identical peeling outcomes in each trial. Under the standard encoding with a one-sided difference, the failure samples shift right, with mean $+0.599$ and KS distance $0.242$ from all samples (Figure~\ref{fig:encoding}(a)). With a balanced difference, the failure samples and all samples almost coincide (Figure~\ref{fig:encoding}(b)). Under random-sign encoding, the failure samples have mean $+0.001$ and KS distance $0.001$, almost coinciding with all samples and $N(0,1)$ (Figure~\ref{fig:encoding}(c)).

\begin{figure}[!htbp]
\centering\includegraphics[width=\linewidth]{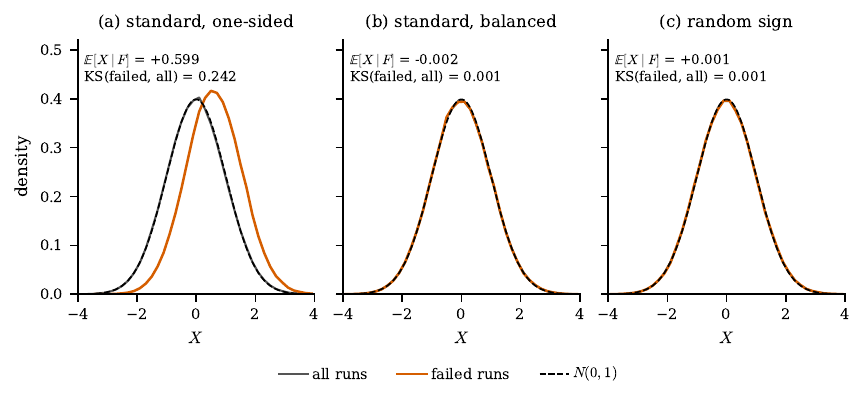}
\caption{$X=(\widehat d-d)/s(d)$ under three conditions on the same batch of mappings ($k=3$, $M=16384$, $d=13314$, failure rate 0.125). (a)(b) Standard encoding with one-sided and balanced differences; (c) random-sign encoding. The dashed curve is $N(0,1)$.}
\label{fig:encoding}
\end{figure}

The standard encoding's bias depends on both sign composition and $k$. The horizontal axis in Figure~\ref{fig:composition}(a) covers five sign compositions, with $d_-/d$ ranging from $0$ (one-sided) to $0.5$ (balanced). The solid lines for the standard encoding are near zero at the balanced point and deviate further as the difference becomes more one-sided. The direction depends on $k$: for a one-sided difference, the bias is $+0.638$ at $k=3$ and $-0.387$ at $k=6$ (mechanism in Appendix~\ref{app:std}). The receiver knows $k$ and hence the direction of the bias. Determining the correction's magnitude also requires $(d_+,d_-)$, which is unknown in advance. The dashed lines for random-sign encoding remain near zero. Even with the vertical axis magnified about 40 times (Figure~\ref{fig:composition}(b)), the five sign compositions show no discernible difference. This agrees with Proposition~\ref{prop:composition}: the coding coefficients $s_i\varepsilon_i$ are iid random signs with a distribution independent of $(d_+,d_-)$. The 84 direct-sampling configurations in the statistical experiments give the same conclusion (Appendix~\ref{app:dist}).

\begin{figure}[!htbp]
\centering\includegraphics[width=\linewidth]{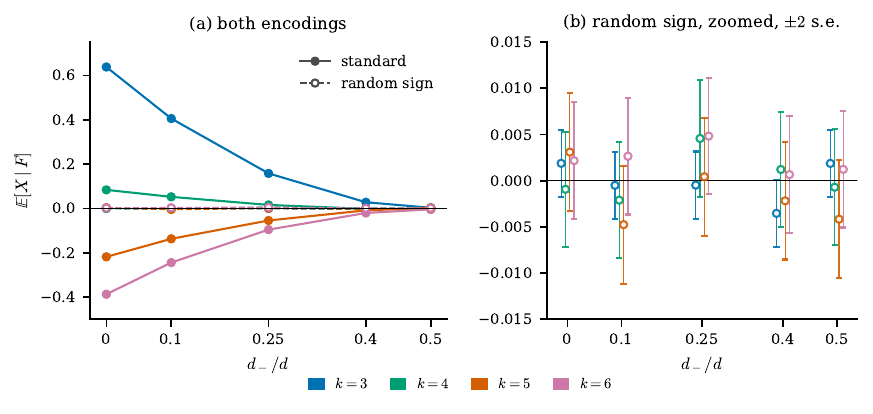}
\caption{Mean post-failure $X$ versus sign composition ($M=16384$, failure probability about 0.1; both encodings share the same mappings for each configuration). The horizontal positions $d_-/d=0$ and $0.5$ correspond to one-sided and balanced differences, respectively. Color indicates $k$. (a) Standard encoding: solid lines and filled markers. Random-sign encoding: dashed lines and open markers. (b) Random-sign data from (a), with the vertical axis magnified about 40 times and error bars of two standard errors.}
\label{fig:composition}
\end{figure}

\paragraph{Decoding equivalence.}
Random-sign encoding preserves the set of keys recovered by decoding (Proposition~\ref{prop:decode}) and therefore inherits the standard IBLT's threshold and failure probability. The decoding-equivalence test (\exptag{DEC}) verifies this in paired trials on a complete IBLT. Both encodings share the same keys and cell hash. Each trial checks the first-round outcome, every accepted key, and the final recovered set against an independent degree-one peeling reference. We observe zero disagreements in more than 500 million paired trials (Appendix~\ref{app:decode-exp}). Table~\ref{tab:decode} reports failure probabilities near the threshold. After normalizing the load by the respective $\alpha_c$, the curves for the five values of $k$ almost coincide. The transition band narrows as $M$ grows, consistent with the critical window of width of order $M^{-1/2}$ in Section~\ref{sec:model}.

\begin{table}[!htbp]
\centering
\caption{First-round decoding failure probability. Columns give the load $d/(\alpha_cM)$. Both encodings agree in every trial, so the values apply to both.}
\label{tab:decode}
\begin{tabular}{rr@{\hspace{1.2em}}ccc@{\hspace{1.2em}}ccc@{\hspace{1.2em}}ccc}
\toprule
& & \multicolumn{3}{c}{$M=1024$} & \multicolumn{3}{c}{$M=16384$} & \multicolumn{3}{c}{$M=65536$}\\
\cmidrule(lr){3-5}\cmidrule(lr){6-8}\cmidrule(lr){9-11}
$k$ & $\alpha_c$ & 0.98 & 1 & 1.02 & 0.98 & 1 & 1.02 & 0.98 & 1 & 1.02\\
\midrule
3 & 0.818 & 0.298 & 0.715 & 0.958 & 0.000 & 0.648 & 1.000 & 0.000 & 0.619 & 1.000\\
4 & 0.772 & 0.266 & 0.743 & 0.977 & 0.000 & 0.654 & 1.000 & 0.000 & 0.620 & 1.000\\
5 & 0.702 & 0.268 & 0.756 & 0.975 & 0.000 & 0.661 & 1.000 & 0.000 & 0.629 & 1.000\\
6 & 0.637 & 0.284 & 0.740 & 0.971 & 0.000 & 0.665 & 1.000 & 0.000 & 0.627 & 1.000\\
10 & 0.461 & 0.383 & 0.759 & 0.966 & 0.000 & 0.663 & 1.000 & 0.000 & 0.630 & 1.000\\
\bottomrule
\end{tabular}
\end{table}

Proposition~\ref{prop:decode} assumes that the checksum makes no errors. In actual decoding, the checksum must reject cells of degree greater than one with count $\pm1$. Random-sign encoding produces such cells even with one-sided differences: a degree-three cell with coding coefficients $+1,+1,-1$ has count $+1$. All such candidates are rejected at the checksum step, with zero false acceptances. Their frequency under random-sign encoding is independent of sign composition and equals that under standard encoding with a balanced difference (Figure~\ref{fig:decode} in the appendix). Section~\ref{sec:exp-cost} reports the resulting additional decoding time.

\subsection{The discrete limit distribution accurately describes the post-failure lower tail at finite size}\label{sec:exp-dist}
The discrete limit distribution $G_d$ coincides with the empirical post-failure lower tail without fitting, allowing the finite-size risk of both rules to be read directly from the figure. Figure~\ref{fig:crosscdf} uses six configurations with $c$ ranging from 0.25 to 8 at $k=3$, $M=16384$. It shows the empirical CDF of post-failure $(\widetilde d-d)/s(d)$ (orange), $G_d$ from Theorem~\ref{thm:cross} (black), and $N(0,1)$ (dashed), on a logarithmic vertical axis to display the lower tail. We first examine the shape and then the risks of the two rules.

\begin{figure}[!htbp]
\centering\includegraphics[width=\linewidth]{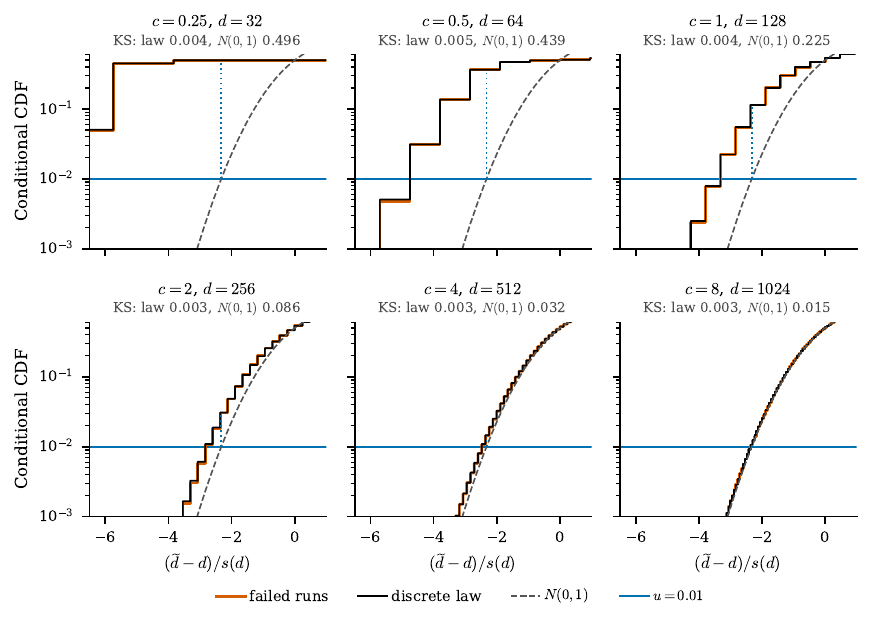}
\caption{Post-failure lower tail of $(\widetilde d-d)/s(d)$ at different $c$ ($k=3$, $M=16384$; 400,000 planted samples per configuration for $c=1,2,4$, and 100,000 for the rest). Below each title are the KS distances from the empirical CDF to $G_d$ and to $N(0,1)$. The blue horizontal line marks CDF $u=0.01$; the blue dotted vertical line sits at the 1\% quantile $-2.33$ of $N(0,1)$ and ends at the empirical CDF there.}
\label{fig:crosscdf}
\end{figure}

\paragraph{Lower-tail shape.}
The black and orange curves almost coincide in all six configurations, with KS distances from the empirical CDF to $G_d$ of 0.003--0.005. The quantiles of $G_d$ thus match the true post-failure quantiles. The lower tail changes with $c$: at small $c$, it has only a few coarse steps; as $c$ grows, the steps become denser and approach the dashed curve. The steps arise from duplicate edge pairs. At small $c$, failure is almost always caused by one such pair. If the edges have opposite signs, the observation is 2 below $d$. This $-2$ offset is large relative to $s(d)$: at $c=0.25$, it lies at a standardized error of about $-5.7$, and about four tenths of the failure samples fall on this step, visible as the large step in the $c=0.25$ panel. As $c$ grows, overlap among the remaining edges contributes many $\pm2/k$ offsets that fill in the steps (Appendix~\ref{app:limit}). The KS distance to $N(0,1)$ drops from 0.496 at $c=0.25$ to 0.015 at $c=8$. Finite-size error decreases as $M$ grows: across $k=3,4,5$, the KS distance to $G_d$ at $c=1,2,4$ drops from at most 0.034 at $M=256$ to at most 0.012 at $M=16384$ (Table~\ref{tab:limit-ks} in the appendix).

\paragraph{Risk of the two rules.}
The horizontal line marks CDF $u=0.01$. The horizontal coordinate where a curve meets it is that distribution's $u$ quantile. On this scale, the Gaussian rule's acceptance lower bound~\eqref{eq:gauss} is approximately the $u$ quantile of $N(0,1)$, or $-2.33$, where the dashed curve meets the horizontal line. A standardized error below this value causes rejection of the true $d$. The height of the orange curve at $-2.33$, marked by the top of the dotted vertical line, therefore approximates the Gaussian rule's conditional miss rate. At $c=8$, all three curves meet the horizontal line at the same point, and the conditional miss rate is about $1.1u$. At $c=1$, the orange curve has height about 0.11 at $-2.33$, or $11u$. At $c=0.25$, $-2.33$ lies to the right of the large step at $-5.7$, and the orange curve has height about 0.49: roughly half the failure samples lead to underestimation. These readings agree with the miss/u values measured in the statistical experiments (\exptag{STAT}) in Section~\ref{sec:exp-cover} (Table~\ref{tab:gauss-c} in the appendix). For candidate $d$, the Skellam rule uses the black curve as its null distribution. Its rejection region is determined by the intersection of the black curve and the horizontal line, placing the risk near $u$. At small $c$, the horizontal line often crosses a vertical segment between two steps. The rule can reject only observations on the lower step and below, giving risk below $u$ and explaining its conservatism at small $c$. Section~\ref{sec:exp-cover} tests both observations across all 189 configurations.

\subsection{Conditional miss rate and extra capacity at finite size}\label{sec:exp-cover}
We test both rules' risk at finite table sizes and measure the extra capacity required by the upper bound. Theorem~\ref{thm:coverage} guarantees that the Skellam rule's conditional miss rate is asymptotically at most $u$ for all $d\ge2$. The Gaussian rule has the same guarantee only when $c\to\infty$ (Corollary~\ref{cor:gauss}). The conditional miss rate is the fraction of failure samples with a miss ($U<d$), reported as miss/$u$. Values at most 1 indicate risk within the target. Extra capacity is $(U-d)/d$. We first examine all configurations in Figure~\ref{fig:coverage}, then retest those closest to the risk upper bound in Figure~\ref{fig:confirm}.

\paragraph{Conditional miss rate.}
Figure~\ref{fig:coverage} reports miss/u for all 189 configurations of the statistical experiments (\exptag{STAT}). The columns correspond to the two values of $u$; the top and bottom rows use direct and planted sampling, respectively. At very small $c$, failure is so rare that only planted sampling supplies samples (Section~\ref{sec:exp-setup}). The circles for the Skellam rule lie near the dashed line miss/u $=1$ when $c\ge8$ and fall below it as $c$ decreases. Values near 1 indicate that the rule uses the risk budget without exceeding it: the maximum is 1.04, while with about 20,000 failures per configuration, the standard error of miss/u at $u=0.01$ is about 0.07. Values below 1 result from the steps in Figure~\ref{fig:crosscdf}: when $u$ falls between two steps, the rule can reject only the lower step and below, giving risk below $u$. This is especially visible at $u=0.01$, where circles in the bottom row lie on the lower boundary, indicating zero misses.

\begin{figure}[!htbp]
\centering\includegraphics[width=\linewidth]{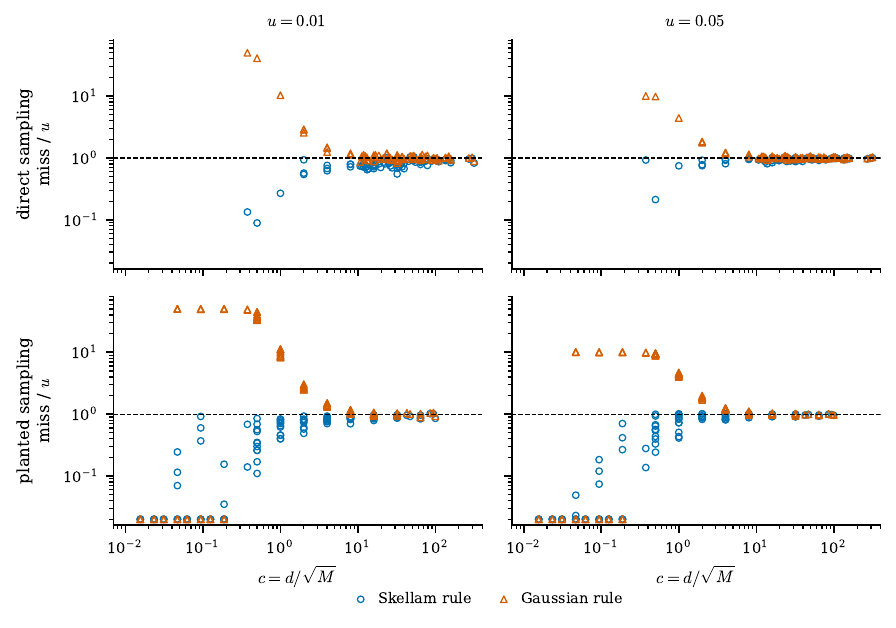}
\caption{miss/u versus $c=d/\sqrt M$ in the statistical experiments ($k=3,4,5$, $M=256$--16384, about 20,000 failures per configuration; configurations with zero misses appear on the lower boundary). Top: direct sampling. Bottom: planted sampling. Circles: Skellam rule. Triangles: Gaussian rule.}
\label{fig:coverage}
\end{figure}

\begin{figure}[!htbp]
\centering\includegraphics[width=\linewidth]{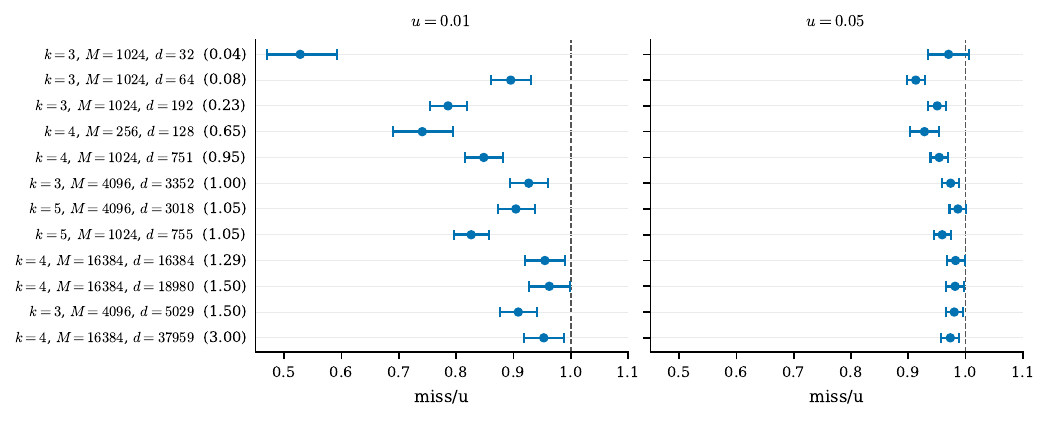}
\caption{Estimated miss/u for the Skellam rule with two-sided 95\% intervals (Clopper--Pearson) on the \exptag{WORST} configurations. Configurations are ordered by load, shown in parentheses. The dashed line marks miss/u $=1$.}
\label{fig:confirm}
\end{figure}

The triangles for the Gaussian rule coincide with the circles only when $c\ge8$ and rise rapidly as $c$ decreases. At $u=0.01$, miss/u is about 8--11 at $c=1$ and exceeds 33 at $c=0.5$ (Table~\ref{tab:gauss-c} in the appendix). In the bottom row, most triangles with $c<0.3$ form a horizontal line: the miss rate stays at $1/2$, or miss/u $=1/(2u)$, giving 50 and 10 in the two columns. At very small $c$, the post-failure observation falls near $d+2$ or $d-2$ with probability one half each (Proposition~\ref{prop:smalld}). The latter lies to the left of the Gaussian quantile, as illustrated by the large step at $c=0.25$ in Figure~\ref{fig:crosscdf}. The few triangles on the lower boundary correspond to $d=2,3$: the observation is then below the acceptance thresholds of all candidates, so the Gaussian rule returns $+\infty$. This excess risk persists as table size grows, consistent with the finite limit in Theorem~\ref{thm:cross}.

\paragraph{Worst-configuration retest.}
About 20,000 failures per configuration in Figure~\ref{fig:coverage} are insufficient to confirm that configurations closest to the risk upper bound also meet the target. The worst-configuration test (\exptag{WORST}) selects 12 such configurations using a rule fixed before sampling, covering bounded $c$, the critical window, and overload. Independent samples provide 56,000 to 330,000 failures per configuration. Figure~\ref{fig:confirm} orders these configurations by load from top to bottom. All 24 estimates are at most 0.99, and 22 two-sided 95\% intervals lie entirely to the left of 1. The largest one-sided simultaneous confidence upper bound is $1.024u$, meeting the $1.1u$ criterion set before sampling (Appendix~\ref{app:cover-more}). The lowest-load configuration at $u=0.01$ lies farthest from 1, consistent with the conservatism at small $c$ in Figure~\ref{fig:coverage}.

\paragraph{Extra capacity.}
Controlling misses requires an upper bound above $d$. When $c\ge8$, extra capacity is determined by table size and decreases as $M^{-1/2}$, with nearly identical values for both rules: $U-d$ is about 2.3--2.8 times $s(d)$ and $s(d)/d\approx\sqrt{2/M}$. At $u=0.01$, the Skellam rule's median extra capacity drops from 25\% at $M=256$ to 2.6\% at $M=16384$ (Table~\ref{tab:extra} in the appendix). When $c$ is bounded, $d$ itself is small, making the ratio less informative, and few extra elements are needed: at $c=4$, the median $U-d$ is 15. In the two-round protocol, this extra capacity appears as the increase in second-round table size over the oracle (Section~\ref{sec:exp-e2e}).

We recommend $M\ge1024$ for deployment. The null distribution in Section~\ref{sec:inversion} omits cells shared by the duplicate edge pair and the remaining keys. Its occurrence rate, $(d-2)k^2/M$, grows as table size decreases. At $M=64$, this omission becomes nonnegligible: the Skellam rule's maximum miss/u is 1.08, beyond sampling error. At $M=256$, we observe no excess in miss/u beyond sampling error. At the same $c$, the omitted structure is more frequent than at $M=1024$, and median extra capacity reaches 25\%.

\subsection{The two-round protocol on a complete IBLT}\label{sec:exp-e2e}
Theorem~\ref{thm:protocol} decomposes the two-round failure probability into two core risks, the joint miss rate and peeling failure with sufficient capacity, and two additional risks, an infinite upper bound and checksum error. We run the two-round protocol on a complete IBLT to compare the first two terms in practice, measure the random-sign IBLT's communication cost, and compare it with the existing calibration rule.

The protocol experiments (\exptag{PROTO}) use $k=3,4,5$, first-round table sizes $M_1\in\{256,1024,4096,16384\}$, and first-round loads 0.95, 1, 1.05, 1.5, 3. We test two sign compositions, one-sided (all on the $A$ side) and balanced, for a total of 120 configurations. The two sides also hold 1000 common elements. We fix each configuration's trial count in advance using a pilot estimate of the first-round failure probability. We target about 20,000 first-round failures, capped at 3,000,000 trials. Configurations with no first-round failure in the pilot receive a fixed 100,000 trials. Both values of $u$ use the same trials. After each first-round failure, we run the second round separately with table sizes determined by three upper bounds: the Skellam rule, the Gaussian rule, and an oracle using the true $d$. Both rules use $U_b=\max(U,|b|)$ as in Section~\ref{sec:flow}. The second-round table size is $\max\{\lceil U_b/(0.8\alpha_c)\rceil,64\}$, with a fresh cell hash. All three upper bounds share the keys and hashes within each trial, allowing paired comparisons that differ only in the upper bound (implementation details in Appendix~\ref{app:protocol}).

The two-round protocol achieves reliability comparable to the oracle. In Figure~\ref{fig:protocol}(a), both rules' final failure rates fall in the same range as the oracle's, all at most 0.13\% when $M_1\ge1024$. At $u=0.01$, their maximum is at most the oracle's (Table~\ref{tab:e2e} in the appendix). In paired comparisons by configuration, only the configuration used in the example of Section~\ref{sec:protocol} has a 95\% paired interval excluding zero: the Skellam rule has 10 final failures and the Gaussian rule has 20. Final failures arise mainly from second-round peeling, consistent with the expected dominance of the second term in Theorem~\ref{thm:protocol}. At $u=0.01$, after pooling configurations by $M_1$, at most 6\% of final failures follow a miss. The second-round table size includes a margin of $1/(0.8\alpha_c)$, which absorbs most slight underestimates. This margin also explains why the rules perform at least as well as the oracle: their upper bounds usually exceed $d$, yielding slightly larger second-round tables and fewer peeling failures.

Extra communication cost depends only on $M_1$ and decreases as $M_1^{-1/2}$. In Figure~\ref{fig:protocol}(b), configurations with different $k$, loads, and sign compositions almost coincide at each $M_1$ and follow the dotted line with slope $-1/2$. At $u=0.01$, the Skellam rule's increase in second-round table size over the oracle drops from 25.6\% at $M_1=256$ to 2.6\% at $M_1=16384$, consistent with the extra capacity in Section~\ref{sec:exp-cover}. In total bytes, the overhead drops from 14.6\% to 1.5\%.

\begin{figure}[!htbp]
\centering\includegraphics[width=\linewidth]{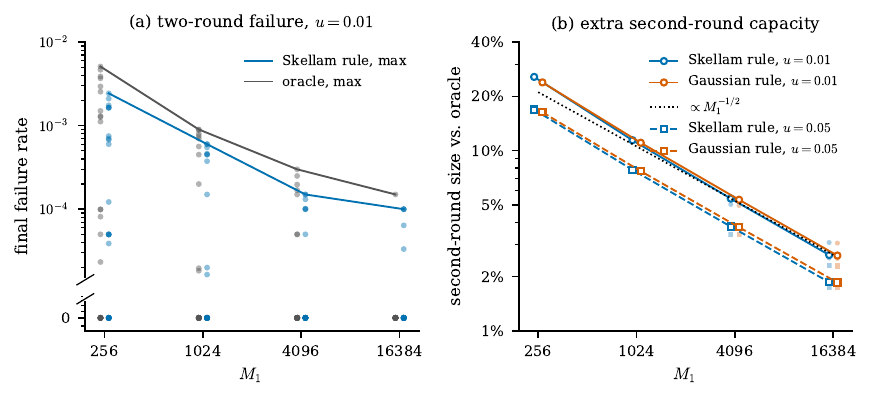}
\caption{Two-round protocol as first-round table size $M_1$ varies (120 configurations, 13.94 million trials per $u$). (a) Final two-round failure rate for each configuration at $u=0.01$. Configurations with zero failures appear at 0 below the axis break. (b) Increase in second-round table size over the oracle after first-round failure, with one point per configuration. The dotted line has slope $-1/2$.}
\label{fig:protocol}
\end{figure}

Figure~\ref{fig:scale} shows the same trend at fixed load (first-round target load 1.5, balanced, $u=0.01$). As table size grows, the Skellam rule's miss/u approaches 1 from below, while the Gaussian rule stays near 1. Each point in the top row pools all first-round failures for one configuration (Appendix~\ref{app:cover-more}). Total byte overhead relative to the oracle drops from 16.8\% at $M_1=256$ to 1.7\% at $M_1=16384$ and becomes negligible as table size grows.

\begin{figure}[!htbp]
\centering\includegraphics[width=\linewidth]{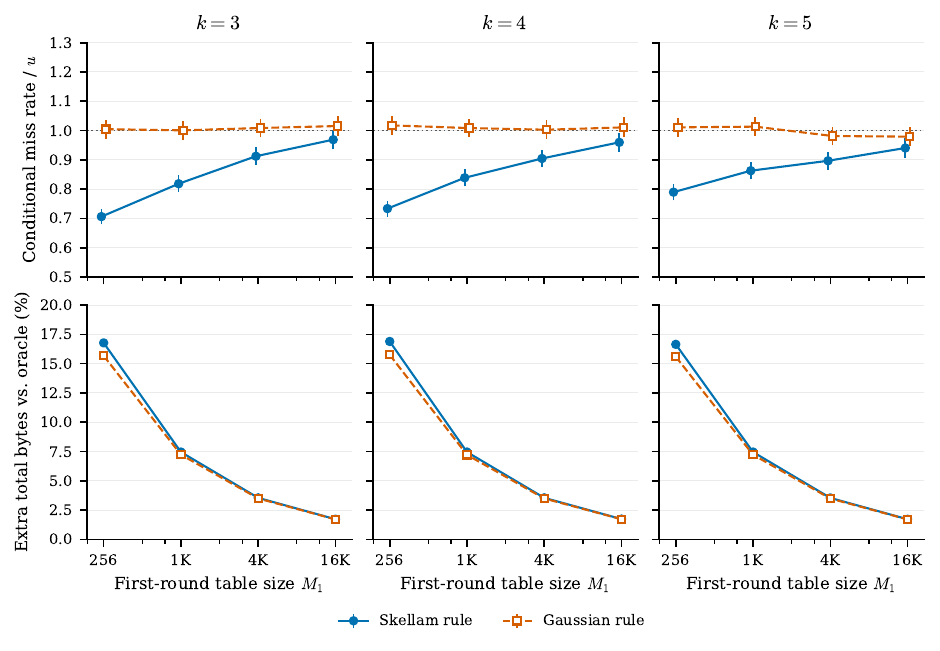}
\caption{Two-round protocol at fixed load as table size varies (balanced, $u=0.01$, target load 1.5). Each column corresponds to one $k$. Top: all first-round failures pooled by configuration (331,000 to 372,000 per point). Bottom: communication overhead relative to the oracle for the same configuration.}
\label{fig:scale}
\end{figure}

\paragraph{Comparison with the existing calibration rule.}
Prior work~\cite{wu2026selfsizing} also sets second-round capacity from the post-failure $\widehat d$ in a standard IBLT. Because the post-failure bias depends on the unknown sign composition, it obtains a calibration factor from offline simulation. For $k=3$, it groups configurations by first-round table size and takes the worst post-failure 1\% quantile $q_{0.01}$ of $\widehat d/d$ over several loads and sign compositions, then sets $U=\lceil\widehat d/q_{0.01}\rceil$. We use this upper bound in the same protocol, keeping all other settings identical to those for our two rules; Table~\ref{tab:vldb} in the appendix gives the complete results. Two findings emerge. First, under standard encoding, the risk of the same factor varies widely with sign composition: at loads 0.95--1.0, conditional miss rates are $0.05$--$0.40u$ for one-sided differences and $0.86$--$0.94u$ for balanced differences, an order-of-magnitude difference. This is the protocol-level consequence of the bias in Section~\ref{sec:exp-decode}, and the factor must be calibrated for the worst sign composition. Random-sign encoding brings both into $0.80$--$1.01u$. Second, the factor requires separate simulation calibration for each parameter configuration. The published values cover only $k=3$, $M_1\le4096$, and $u=0.01$. For these parameters, the calibration rule and the Skellam rule have similar final failure rates and total byte costs. We compute the Skellam rule directly from the theoretical distribution in Section~\ref{sec:inversion}, without calibration for any $k$, $M$, load, or $u$.

\subsection{Computational cost and rule choice}\label{sec:exp-cost}\label{sec:exp-deploy}
Almost all computation beyond standard encoding comes from sizing with the Skellam rule. The overhead of random-sign encoding itself is negligible. With random-sign encoding, the decoder processes every difference as a balanced difference: cells containing multiple keys with count $\pm1$ also enter the candidate queue and are then rejected by the checksum. For a one-sided difference, the number of candidates checked per recovered key therefore rises from 1.00 to 1.77 ($k=3$). The resulting increase in decoding time is 1--2\% of total protocol time at loads 1.0--1.05 and less than 1\% at higher loads.

In deployment, sizing with the Skellam rule requires only a table lookup. The parameters $(M_1,k,u)$ are fixed, and $U$ depends only on the single observed integer $\sum_vC_v^2$. We can therefore precompute $U$ for each possible observation within a chosen range and store the results in a table indexed by the observation. After computing $\sum_vC_v^2$, the receiver retrieves $U$ directly. The results agree entry by entry with online computation (Appendix~\ref{app:search-lut}). Table~\ref{tab:proto} reports total protocol time at $M_1=16384$, including first-round encoding and decoding, sizing, and the second round. With table lookup, each sizing takes less than 1 $\mu$s, and total protocol time differs from the Gaussian rule by less than 0.5\%.

When parameters vary across synchronizations, precomputation is inconvenient, or the observation exceeds the table's range, the Skellam rule computes the upper bound online. Online computation evaluates the exact CDF for each candidate, with cost growing roughly as $c^2$. Total protocol time is 15--24\% longer than for the Gaussian rule at loads 1.0--1.05 and 40--45\% longer at load 1.5. At load 3, sizing dominates, and total time is 3.1--3.8 times that of the Gaussian rule. A receiver can cache Poisson probability tables across synchronizations to reduce this cost. Appendix~\ref{app:timing} gives the conditions for the computational-cost measurement (\exptag{COST}), table construction and verification, and per-sizing time and memory use under different cache conditions.

\begin{table}[!htbp]
\centering
\caption{Total protocol time per trial (ms, median of 10 repetitions; $M_1=16384$, one-sided difference). Lookup and online are two implementations of the Skellam rule; the online implementation caches Poisson probability tables between successive trials.}
\label{tab:proto}
\begin{tabular}{lrrrrrrrrr}
\toprule
& \multicolumn{3}{c}{Skellam rule (lookup)} & \multicolumn{3}{c}{Skellam rule (online)} & \multicolumn{3}{c}{Gaussian rule}\\
\cmidrule(lr){2-4}\cmidrule(lr){5-7}\cmidrule(lr){8-10}
Load & $k=3$ & $k=4$ & $k=5$ & $k=3$ & $k=4$ & $k=5$ & $k=3$ & $k=4$ & $k=5$\\
\midrule
1.0 & 11.0 & 14.0 & 15.3 & 12.7 & 16.8 & 17.7 & 11.0 & 14.0 & 15.3\\
1.05 & 13.1 & 15.6 & 18.2 & 16.2 & 19.4 & 21.5 & 13.1 & 15.6 & 18.2\\
1.5 & 14.0 & 16.7 & 19.6 & 20.1 & 24.2 & 27.6 & 14.0 & 16.7 & 19.6\\
3 & 18.9 & 22.4 & 25.7 & 58.8 & 85.4 & 90.1 & 18.8 & 22.3 & 25.6\\
\bottomrule
\end{tabular}
\end{table}

The Gaussian rule meets the target only when $d\gg\sqrt M$. Table size alone does not exclude bounded $c$ or fixed small $d$, so the rule is suitable only for applications whose workloads ensure this condition. The Skellam rule applies to all $d\ge2$ without knowing the load region, and with table lookup its sizing cost is no higher than that of the Gaussian rule. For deployment, we recommend the Skellam rule with $M_1\ge1024$ (Section~\ref{sec:exp-cover}).

\section{Related work}\label{sec:related}
\paragraph{IBLT and set reconciliation.} IBLT\cite{gm} and the invertible Bloom filter\cite{eppstein2011whatsdifference} list set differences by peeling. Decoding success is determined by the 2-core of a random hypergraph\cite{molloy,dw,dm}. At low load, duplicate edge pairs dominate failure\cite{gm,yw}, and the probability of a linear-sized core at subcritical load is exponentially small\cite{kim2008poisson,dn08}. Subsequent work improves the hypergraph structure\cite{lazaro2021irregular}, develops rateless codes\cite{yang2024rateless,lazaro2023ratecompatible}, adds a stash to lower the failure probability\cite{belazzougui2024better}, or predicts and breaks the 2-core after decoding failure\cite{gadurek2026breaking}. These works seek to reduce failure or resume decoding after failure. We read the same table after failure to determine the next-round capacity.

\paragraph{Estimating the difference size.} A common approach to estimating $d$ is to send a separate estimator, such as the strata estimator\cite{eppstein2011whatsdifference}, min-wise hashing\cite{broder2000minwise}, the odd sketch\cite{mitzenmacher2014odd}, or GXBits\cite{jia2024gxbits}. PBS\cite{gong2021pbs} first estimates $d$ using independent ToW sketches, then sets capacity to a fixed multiple of the estimate. It transmits and accounts for the estimator and reconciliation structure separately. We obtain the estimate from the same counts used for decoding, avoiding a separate estimator. This reuse requires analysis conditional on decoding failure: the estimate is needed only after failure, and failure selects the mapping.

\paragraph{Random-sign sketches and second-moment estimation.} AMS\cite{alon1999spacecomplexity} estimates $F_2$ using a sum of random signs. CountSketch\cite{charikar2004countsketch} supports point queries with random-sign counters. Thorup and Zhang\cite{thorup2004tabulation} analyze its second-moment extension and give a centered second-moment estimate for a single counter array. Fast-AGMS\cite{cormode2005fastagms} uses the same structure for distributed queries. $\widehat d$ has the same algebraic form as these estimators. Its conditional unbiasedness given the mapping also follows the structure of the Hutchinson trace estimator\cite{hutchinson1990trace}. These sketches carry no decodable keySum and do not address the distribution conditional on decoding failure. Fischer and Nakos\cite{fischer2025sparse} combine hypergraph peeling and a random-sign CountSketch in a single linear sketch, with the tail norm estimated by a separate short sketch. We encode the shared random signs in the IBLT's count, keySum, and checksum. The same table supports peeling and, after failure, provides an estimate independent of the difference direction, together with a post-failure upper bound valid at all loads.

\paragraph{Reading the same table after failure.} Prior work\cite{wu2026selfsizing} reads count energy after first-round failure to set the second-round capacity. Its analysis models sign composition with iid signs and obtains a conditional mean equal to $d$. Its upper-bound rule uses a chi-square approximation under deep overload and per-configuration simulation calibration at other loads, where its analytic bound is too conservative for small $M$. We encode random signs in all three IBLT fields, making these properties hold exactly for every fixed mapping. We also derive a post-failure distribution and upper bound valid at all loads.

\paragraph{Conditional limits and discrete approximations.} For bounded $c$ (\eqref{eq:c}), we derive the limit distribution using the Chen--Stein method\cite{arratia1989chen} for the Poisson approximation of sparse overlap. The Skellam distribution describes the difference between two types of counts. At low load, we replace the failure-conditioned distribution with the planted distribution, which is the distribution size-biased by the duplicate edge pair count\cite{arratia2019sizebias}. When $d\gg\sqrt M$, we analyze the random-sign quadratic form given the mapping using the cumulant method on a dependency graph\cite{janson1988cumulant}. This requires only concentration bounds\cite{warnke2016typical,combes2015mcdiarmid} and degree control for the overlap graph. We prove a Gaussian limit uniform over deterministic coefficients using an edge-exposure martingale CLT. The upper-bound rule uses classical test inversion\cite{neyman1937}.

\section{Discussion and conclusion}\label{sec:conclusion}
The random-sign IBLT uses a single table for both decoding and difference-size estimation, preserving the cells and decoding capability. The mapping alone determines decoding success, while the signs supply the estimate's randomness. This separation makes the estimate exactly unbiased for every fixed mapping and identically distributed across sign compositions. It also lets us fully characterize the post-failure distribution: Gaussian when $d\gg\sqrt M$, and discrete with a $\pm2$ term plus a Skellam term when $c=d/\sqrt M$ is bounded. One family of distributions describes all scales. Using this family as the null distribution, the Skellam rule keeps the conditional miss rate asymptotically at most the target risk, uniformly over all $d\ge2$.

For deployment, we recommend the Skellam rule, which retains the uniform guarantee for all $d\ge2$. When $(M_1,k,u)$ is fixed, its upper bounds can be precomputed and stored in a table, reducing sizing to a single lookup. Section~\ref{sec:exp-deploy} gives the guarantees and costs of the two rules. The Skellam rule meets the target within sampling error on all experimental configurations with $M\ge256$, and we recommend $M\ge1024$ for deployment (Section~\ref{sec:exp-cover}). Shortening the checksum to 32 bits yields narrower cells while preserving the sizing risk. The verification error probability is of order $10^{-4}$ at $M_1=16384$ and load at most 3 (Section~\ref{sec:protocol}). Storing counts for several independent signs per cell reduces absolute error while preserving the ratio of the post-failure conditional scale to the unconditional scale. Inversion still uses the post-failure distribution (Appendix~\ref{app:multi}). Random-sign encoding preserves cell space overhead, and the additional decoding time is at most about 2\% of total protocol time (Section~\ref{sec:exp-cost}).

Three questions remain for future work. First, rigorous finite-parameter tail upper bounds could reject more candidates before constructing the Poisson tables, reducing the Skellam rule's computational cost. Second, making the error constants in Theorem~\ref{thm:cross} and Theorem~\ref{thm:clt} computable would give a finite risk bound for a specified $M$. Third, we seek to extend the method to multiset differences whose weights can exceed one. In this setting, the random-sign count gives an unbiased estimate of the difference vector's $F_2$. Decoding and post-failure analysis for this extension remain to be developed.

\section*{Code and data availability}
The implementation, frozen experimental configurations, summary data, and the programs that generate every figure and table are available at \url{https://github.com/whitewum/random-sign-artifact}. The source is archived at Zenodo under DOI \href{https://doi.org/10.5281/zenodo.23754604}{10.5281/zenodo.23754604}, and all raw experiment records are published as assets of the repository's releases.

\bibliographystyle{plain}
\bibliography{references}

\appendix
\section{Notation}\label{app:notations}
We collect the notation used in the main text and principal proofs, ordered by section for reference. We suppress dependence on $M$ except when comparing quantities across scales. We fix $k\ge3$ and let $M\to\infty$. A subscript $k$ denotes a constant depending only on $k$. The symbol $\Rightarrow$ denotes convergence in distribution. As in Section~\ref{sec:design}, we write the coding coefficient $s_i\varepsilon_i$ as $\varepsilon_i$ throughout.

\begin{table}[!htbp]
\centering\footnotesize
\caption{Principal notation: model, encoding, and measurement.}
\begin{tabular}{@{}l p{0.72\linewidth}@{}}
\toprule
Notation & Meaning\\
\midrule
\multicolumn{2}{@{}l}{\bfseries Model and peeling}\\
$A$, $B$ & sets held by the two sides; $|A|$ and $|B|$ are the set sizes\\
$M$, $M_1$ & number of cells; number of cells in the first-round table\\
$k$ & number of cells selected by each key\\
$x$, $x_i$ & key; the $i$-th differing key\\
$\mathrm{chk}(x)$ & checksum of key $x$\\
$d$ & difference size $|A\triangle B|$\\
$d_+$, $d_-$ & numbers of keys in $A$ only and in $B$ only, respectively\\
$b$ & difference between the two set sizes, $d_+-d_-=|A|-|B|$\\
$S_i$ & the $k$-subset selected by key $i$\\
$\mathcal G$ & unsigned mapping with key labels\\
$H_{ij}$ & overlap $|S_i\cap S_j|$ of two edges\\
$F$ & peeling failure on a single table; $\Pr(F)$ is the failure probability\\
$F_1$, $F_2$ & peeling failure in the first and second rounds\\
$E$ & checksum misjudgment event\\
$\alpha_c$ & peeling threshold (about $0.818$ for $k=3$ and $0.772$ for $k=4$)\\
\midrule
\multicolumn{2}{@{}l}{\bfseries Encoding and measurement}\\
$s_i$ & set direction of key $i$, taking values $\pm1$\\
$\mathcal U$, $\varepsilon(\cdot)$ & universe of keys; sign hash shared by the two sides\\
$\varepsilon_i$ & random sign $\varepsilon(x_i)$ of key $i$; in the analysis the coding coefficient $s_i\varepsilon_i$ is written as $\varepsilon_i$\\
$C_v$ & count of cell $v$, $\sum_{i:v\in S_i}\varepsilon_i$ (in the standard encoding, $\sum_{i:v\in S_i}s_i$)\\
$\widehat d$ & estimator of $d$, see~\eqref{eq:dhat}; used by the Gaussian rule\\
$s(n)$ & standardization scale $\sqrt{2n(n-1)/(M-1)}$, see~\eqref{eq:s}\\
$X$ & standardized estimation error $(\widehat d-d)/s(d)$, see~\eqref{eq:X}\\
$\widetilde d$ & estimator of $d$, $\frac1k\sum_vC_v^2$, see~\eqref{eq:raw-energy}; used by the Skellam rule\\
$c_F$ & normalized excess of the overlap over its unconditional mean in the failure sample, see~\eqref{eq:projection}\\
$\Phi$, $\mathcal N(0,1)$ & CDF and distribution of the standard Gaussian, respectively\\
\bottomrule
\end{tabular}
\end{table}

\begin{table}[!htbp]
\centering\footnotesize
\caption{Principal notation: inversion, protocol, and proofs.}
\begin{tabular}{@{}l p{0.72\linewidth}@{}}
\toprule
Notation & Meaning\\
\midrule
\multicolumn{2}{@{}l}{\bfseries Inversion and protocol}\\
$u$ & target risk, taken as $0.01$ and $0.05$ in the experiments\\
$n$ & candidate difference size, an integer $n\ge2$\\
$U$, $U_b$ & upper bound on $d$ given by the selected rule; $\max(U,|b|)$ after enforcing the lower bound $|b|$\\
$t(n)$ & lower acceptance bound for the Gaussian rule, see~\eqref{eq:gauss}\\
$\chi^2_{M-1,u}$ & lower $u$ quantile of the chi-square distribution with $M-1$ degrees of freedom\\
$\lambda_n$ & parameter of the Skellam term in the null distribution, see~\eqref{eq:lambda}\\
$Z_n$ & $\Skellam(\lambda_n,\lambda_n)$, the difference of two independent $\Pois(\lambda_n)$ variables\\
$J_\lambda$ & $2\varepsilon_1\varepsilon_2+\frac2k\Skellam(\lambda,\lambda)$, see~\eqref{eq:J}\\
$G_n$ & CDF of the null distribution $n+J_{\lambda_n}$ of candidate $n$ (Definition~\ref{def:null})\\
$c$ & $d/\sqrt M$ (\eqref{eq:c}); its limit in Theorem~\ref{thm:cross}\\
$m(\cdot)$ & second-round table-size rule $m(U_b)=\lceil U_b/(0.8\alpha_c)\rceil$\\
$\Pr_m(F)$ & failure probability on a new table with $m$ cells\\
\midrule
\multicolumn{2}{@{}l}{\bfseries Proofs}\\
$w_i$ & incidence vector of key $i$ minus the mean $k/M$, with $\langle w_i,w_j\rangle=H_{ij}-k^2/M$ (Appendix~\ref{app:proof-moments})\\
$N$, $\mu$ & number of duplicate edge pairs and its expectation $\binom d2/\binom Mk$\\
$Q$ & planted distribution: choose a pair of edge labels uniformly, then condition on this pair having the same mapping\\
$D_v$, $\ell_i$ & degree of cell $v$; $\sum_{j\ne i}H_{ij}$\\
$\Omega_2$ & sum of squared overlaps $\sum_{i<j}H_{ij}^2$\\
$\mathcal G_M$ & set of good graphs in Appendix~\ref{app:graph}, whose complement has probability smaller than any fixed negative power of $M$\\
\bottomrule
\end{tabular}
\end{table}

\section{Post-failure bias of the standard IBLT}\label{app:std}
We prove Proposition~\ref{prop:projection} and verify it numerically.

In the standard IBLT, the count is $C_v=\sum_{i:v\in S_i}s_i$, with the set direction as the coefficient. Expanding $\widehat d$ gives
\begin{equation}
\widehat d-d=\frac{2}{k(1-k/M)}\sum_{i<j}s_is_j\Bigl(H_{ij}-\frac{k^2}{M}\Bigr).\label{eq:std-expand}
\end{equation}
Failure changes the overlap distribution. Overlap between same-sign pairs raises $\widehat d$, whereas overlap between opposite-sign pairs lowers $\widehat d$. This yields Proposition~\ref{prop:projection}.

\begin{proof}[Proof of Proposition~\ref{prop:projection}]
By~\eqref{eq:std-expand},
\begin{equation*}
\E[\widehat d\mid F]-d=\frac{2}{k(1-k/M)}\sum_{i<j}s_is_j\,\E\Bigl[H_{ij}-\frac{k^2}{M}\Bigm|F\Bigr].
\end{equation*}
The event $F$ depends only on the unsigned mapping and is invariant under permutations of the key labels. Since the random mapping $\mathcal G$ is exchangeable in the labels, $\E[H_{ij}\mid F]$ is the same for all $i<j$ and equals $k^2/M+k(1-k/M)c_F$. The number of same-sign pairs minus the number of opposite-sign pairs is $\binom{d_+}2+\binom{d_-}2-d_+d_-=(b^2-d)/2$. Substituting gives $\E[\widehat d\mid F]-d=c_F(b^2-d)$, and dividing by $d$ yields~\eqref{eq:projection}. The proof uses only these two properties of $F$, so it holds for any event that depends only on the unsigned mapping, is invariant under label permutation, and has positive probability.
\end{proof}

We directly verify Proposition~\ref{prop:projection} using the simulations in Figure~\ref{fig:bias} of Section~\ref{sec:model} and Figure~\ref{fig:composition} of Section~\ref{sec:exp-decode}. Within each configuration, the five sign compositions share the same batch of mappings. Since $F$ depends only on the unsigned mapping, their failure samples coincide trial by trial and share the same $c_F$. The other factor on the right, $b^2/d-1$, depends only on the sign composition and equals $d-1$ for a one-sided difference. Thus the conditional bias for any sign composition equals the one-sided bias multiplied by $(b^2-d)/(d^2-d)$. Once we determine $c_F$ from a single one-sided point, we can predict the bias at every other level without fitting. Figure~\ref{fig:bias} shows this prediction against $(b/d)^2$. The one-sided point at $(b/d)^2=1$ determines the line for each configuration, shown in a different color. The remaining four points fall on that line, confirming~\eqref{eq:projection}.

The full verification covers 20 configurations with $k=3,4,5,6$ and failure probabilities ranging from 0.10 to 0.90 on the $M=16384$ grid. Each configuration has five sign composition levels, $d_-/d\in\{0,0.1,0.25,0.4,0.5\}$. Figure~\ref{fig:bias} shows ten of these configurations with $k=3$ and $k=6$. The standard-encoding curves in Figure~\ref{fig:composition}(a) show four configurations with failure probability about 0.1, including $k=4$ and $k=5$.

The four non-one-sided sign composition levels in each of the 20 configurations give 80 predictions in total. In units of $s(d)$, the differences between predictions and measurements are at most 0.0064 and all lie within two standard errors. The standard error includes sampling error in both the one-sided mean and the measured value.

\section{Encoding properties and proofs of exact moments}\label{app:encoding}
We describe the XOR implementation of keySum and checksum and prove the finite identities in Sections~\ref{sec:design} and~\ref{sec:measurement}. We then summarize their exact verification for small parameter values and discuss the AMS observation obtained from the count sum.

\subsection{XOR implementation}\label{app:xor}
In the set model, each side inserts each key at most once. We can therefore accumulate keySum and checksum by bitwise XOR, as in a common implementation of the standard IBLT, and multiply only the count by $\varepsilon(x)$. Under XOR, $-x=x$. Thus, in each difference-table cell, keySum is the XOR of the differing keys that select the cell, and checksum is the XOR of their $\mathrm{chk}$ values. Both are independent of the coding coefficients. In a pure cell, keySum and checksum equal $x_i$ and $\mathrm{chk}(x_i)$ directly, so we check whether checksum equals $\mathrm{chk}(x)$. To delete $x_i$, we subtract $C_v$ from the count and XOR the other two fields once more with $x_i$ and $\mathrm{chk}(x_i)$, respectively. The induction argument for the additive form in the next subsection also applies here. We recover the direction as $C_v\,\varepsilon(x)$, and peeling follows the same process and gives the same result as in the additive form. In a multiset, the same key may be inserted repeatedly. XOR cancels repeated contributions while the count retains them, so multisets require the additive form. If peeling modifies the difference table in place, the receiver can accumulate the two scalars $\sum_vC_v^2$ and $\sum_vC_v$ during cell-by-cell subtraction, without retaining the original difference table.

\subsection{\texorpdfstring{Proof of Proposition~\ref{prop:decode} (ideal decoding equivalence)}{Proof of ideal decoding equivalence}}\label{app:proof-decode}
\begin{proof}
We induct on the number of peeling steps. After each step, the three fields of every cell equal the encoded sum of the remaining keys: the count is $\sum s_i\varepsilon_i$, keySum is $\sum s_i\varepsilon_i x_i$, and checksum is $\sum s_i\varepsilon_i\,\mathrm{chk}(x_i)$. Each sum ranges over the remaining keys that select the cell. This invariant holds initially by~\eqref{eq:fields}.

Suppose that at some step cell $v$ has degree one in the unsigned hypergraph, with unique key $x_i$. Its count is $s_i\varepsilon_i\in\{\pm1\}$, count times keySum equals $x_i$, and checksum equals count times $\mathrm{chk}(x_i)$. The cell hash of $x_i$ also contains $v$, so $v$ passes the pure-cell test. Conversely, a cell of any other degree that passes the test is a false positive, which the ideal-purity assumption excludes. Thus the cells removable under random-sign encoding are exactly the unsigned degree-one cells. To delete $x_i$, we subtract its three-field contribution multiplied by the coding coefficient $s_i\varepsilon_i$ from its $k$ cells. This preserves the invariant and removes the corresponding edge from the unsigned hypergraph. With the same cell-selection order, the two processes reach the same endpoint. Peeling always ends at the 2-core, independently of the selection order, so $F$ depends only on $\mathcal G$.

At termination, we check all fields. Under random signs, a nonempty core can have all-zero counts, so counts alone cannot certify success. The keySum fields can cancel at the same time. Success is misjudged only if all checksums are also zero. This event is included in $E$ in Section~\ref{sec:protocol}.
\end{proof}

\subsection{\texorpdfstring{Proof of Proposition~\ref{prop:moments} (conditional unbiasedness and conditional moments)}{Proof of conditional unbiasedness and conditional moments}}\label{app:proof-moments}
\begin{proof}
Using $H_{ii}=k$ and $\varepsilon_i^2=1$, we obtain the sum of squared counts and the squared sum:
\begin{align*}
\sum_vC_v^2&=\sum_{i,j}\varepsilon_i\varepsilon_jH_{ij}=kd+2\sum_{i<j}\varepsilon_i\varepsilon_jH_{ij},\\
\frac1M\Bigl(\sum_vC_v\Bigr)^2&=\frac{k^2}M\Bigl(\sum_i\varepsilon_i\Bigr)^2=\frac{k^2}Md+2\sum_{i<j}\varepsilon_i\varepsilon_j\frac{k^2}M.
\end{align*}
Dividing the first identity by $k$ gives~\eqref{eq:raw-expand}. Subtracting the second identity from the first gives $kd(1-k/M)+2\sum_{i<j}\varepsilon_i\varepsilon_j(H_{ij}-k^2/M)$. Dividing this by $k(1-k/M)$ gives~\eqref{eq:rs-expand}. For a fixed mapping, $\E[\varepsilon_i\varepsilon_j]=0$ ($i\ne j$), so the conditional expectations of $\widehat d-d$ and $\widetilde d-d$ are both zero.

We now express the coefficients of the cross terms as inner products. Let $w_i\in\mathbb R^M$ be the incidence vector of key $i$ minus its mean $k/M$. Then $\langle w_i,w_j\rangle=H_{ij}-k^2/M$ and $\|w_i\|^2=k(1-k/M)$, so~\eqref{eq:rs-expand} becomes
\begin{equation*}
\widehat d-d=\frac{2}{k(1-k/M)}\sum_{i<j}\langle w_i,w_j\rangle\,\varepsilon_i\varepsilon_j.
\end{equation*}
Given $\mathcal G$, the inner products are constant. All randomness comes from the sign products $\varepsilon_i\varepsilon_j$. The expectation of a product of sign products is nonzero if and only if each subscript occurs an even number of times, in which case the expectation is one.

For the variance, we expand the square and consider two unordered pairs $\{i,j\}$ and $\{i',j'\}$. When the pairs coincide, the product is $\varepsilon_i^2\varepsilon_j^2=1$. When they differ, they share at most one subscript, so some subscript occurs exactly once and the expectation is zero. Thus only the diagonal terms remain:
\begin{equation*}
\Var(\widehat d\mid\mathcal G)=\frac{4}{k^2(1-k/M)^2}\sum_{i<j}\langle w_i,w_j\rangle^2,
\end{equation*}
which gives the equality in~\eqref{eq:condmoments}.

The exhaustive verification in Appendix~\ref{app:exact} also checks the third central moment; here we give its formula. For the third central moment, we expand the cube and consider ordered triples of unordered pairs. If all three pairs coincide, the product is $\varepsilon_i\varepsilon_j$. If exactly two coincide, the product is the sign product of the third pair. Both cases have expectation zero. If the pairs are distinct, every subscript occurs an even number of times exactly when the pairs form a triangle $\{i,j\},\{i,h\},\{j,h\}$. Each triangle corresponds to $3!=6$ ordered triples. Hence
\begin{equation*}
\E[(\widehat d-d)^3\mid\mathcal G]=\frac{8\cdot6}{k^3(1-k/M)^3}\sum_{i<j<h}\Bigl(H_{ij}-\frac{k^2}{M}\Bigr)\Bigl(H_{ih}-\frac{k^2}{M}\Bigr)\Bigl(H_{jh}-\frac{k^2}{M}\Bigr).
\end{equation*}

The upper bounds in~\eqref{eq:condmoments} follow from Cauchy--Schwarz: $\langle w_i,w_j\rangle^2\le\|w_i\|^2\|w_j\|^2=[k(1-k/M)]^2$. Substitution into the variance formula shows that each of the $\binom d2$ unordered pairs contributes at most $4$, giving $\Var(\widehat d\mid\mathcal G)\le2d(d-1)$. The component of the vector $\sum_i\varepsilon_iw_i$ on cell $v$ is $C_v-\frac1M\sum_uC_u$, so
\begin{equation*}
\widehat d=\frac{\|\sum_i\varepsilon_iw_i\|^2}{k(1-k/M)}\ge0,
\end{equation*}
and the triangle inequality gives $\|\sum_i\varepsilon_iw_i\|\le\sum_i\|w_i\|=d\sqrt{k(1-k/M)}$, so $\widehat d\le d^2$.

Under the random mapping $\mathcal G$, the conditional mean is always $d$. The law of total variance therefore gives $\Var\widehat d=\E[\Var(\widehat d\mid\mathcal G)]$. The intersection size $H_{ij}$ of two independent uniform $k$-subsets has a hypergeometric distribution with mean $k^2/M$ and variance $k^2(1-k/M)^2/(M-1)$. Substitution into the variance formula shows that each unordered pair contributes $4/(M-1)$, giving $\Var\widehat d=2d(d-1)/(M-1)=s(d)^2$.
\end{proof}

\subsection{Exact verification for small parameter values}\label{app:exact}
We verify the finite identities in Sections~\ref{sec:design} and~\ref{sec:measurement} by exhaustive enumeration for small parameter values. For the six triples $(M,k,d)$ in Table~\ref{tab:exact}, we enumerate all mappings and signs, covering 155,640 states. All arithmetic uses integers and rationals, so the results are exact. In every state, common elements cancel in the difference table. Signed peeling recovers the correct keys and directions, and the unrecovered keys form exactly the unsigned 2-core. For every mapping, $\E[\widehat d\mid\mathcal G]=d$, the conditional variance and third central moment match the formulas above, and $\Var(\widehat d\mid\mathcal G)\le2d(d-1)$. The distribution of $\widehat d$ over the random signs is the same for every direction vector. Table~\ref{tab:exact} summarizes the unconditional and post-failure statistics.

\begin{table}[!htbp]
\centering
\caption{Exact enumeration results for small parameter values.}
\label{tab:exact}
\setlength{\tabcolsep}{4.5pt}
\begin{tabular}{@{}crcccccc@{}}
\toprule
$(M,k,d)$ & states & $p_F$ & $\Var\widehat d$ & $\E[\widehat d\mid F]$ & $\Var(\widehat d\mid F)$ & $\E[(\widehat d-d)^3\mid F]$ & $\Pr(\widehat d=0\mid F)$\\
\midrule
$(4,3,2)$ & 40 & $1/4$ & $4/3$ & 2 & 4 & 0 & $1/2$\\
$(4,3,3)$ & 160 & 1 & 4 & 3 & 4 & $16/3$ & 0\\
$(4,3,4)$ & 560 & 1 & 8 & 4 & 8 & $64/3$ & $27/512$\\
$(5,3,3)$ & 1,760 & $2/5$ & 3 & 3 & 4 & $20/3$ & 0\\
$(5,3,4)$ & 11,440 & 1 & 6 & 4 & 6 & 12 & $93/8000$\\
$(6,3,4)$ & 141,680 & $2029/4000$ & $24/5$ & 4 & $11256/2029$ & $23104/2029$ & $75/4058$\\
\bottomrule
\end{tabular}
\end{table}

In every row of Table~\ref{tab:exact}, $\Var\widehat d$ equals $s(d)^2=2d(d-1)/(M-1)$ and $\E[\widehat d\mid F]$ equals $d$, consistent with conditional unbiasedness given the mapping in~\eqref{eq:condmoments}. Failure preserves the mean and increases the scale. For the three parameter sets with $p_F<1$, $\Var(\widehat d\mid F)$ is $3$ times, $4/3$ times, and about $1.16$ times the unconditional variance, respectively. When $p_F=1$, failure is certain and the two variances coincide. After failure, $\widehat d$ equals zero with positive probability. For $(4,3,2)$, this probability is $1/2$. The third central moment vanishes for $d=2$, because the sum in Appendix~\ref{app:proof-moments} requires three distinct keys, and is positive for $d\ge3$. The three rows with $p_F=1$ are unconditional values, so the right skew is a property of $\widehat d$ itself, consistent with the skew described by chi-square quantiles in Section~\ref{sec:measurement}.

The values in the $(4,3,2)$ row come from a duplicate edge pair. Each key independently and uniformly chooses one of the $\binom43=4$ $3$-subsets. When the two keys choose the same subset, each of these three cells has degree two. There is no pure cell, and peeling stops immediately, so $p_F=1/4$. The counts of the three cells then all equal $\varepsilon_1+\varepsilon_2$, and the remaining cells have zero counts. When the two signs are equal, $\sum_vC_v^2=12$ and $\sum_vC_v=\pm6$, so~\eqref{eq:dhat} gives $\widehat d=4$. When they differ, all counts are zero and $\widehat d=0$. The two cases each have probability $1/2$. Thus $\E[\widehat d\mid F]=2=d$ and $\Var(\widehat d\mid F)=4$. This variance is $3$ times the unconditional variance $4/3$, and it also equals the bound $2d(d-1)=4$ in~\eqref{eq:condmoments}: the duplicate edge pair makes $w_1=w_2$, so Cauchy--Schwarz is an equality. Also, $\Pr(\widehat d=0\mid F)=1/2$. This is the smallest instance of the duplicate edge pair structure in Section~\ref{sec:inversion}.

\subsection{An AMS observation independent of failure}\label{app:ams}
Each key enters $k$ cells with the same sign. Thus, under random signs, the normalized count sum equals the sign sum in every sample:
\begin{equation}
\frac1k\sum_vC_v=\sum_{i=1}^d\varepsilon_i.\label{eq:ams}
\end{equation}
The right-hand side depends only on the signs, and its square is the standard AMS second-moment estimator with unit weights\cite{alon1999spacecomplexity}. Rearranging~\eqref{eq:dhat} and~\eqref{eq:raw-energy}, we relate $\widetilde d$ and $\widehat d$ through this AMS term:
\begin{equation}
\widetilde d=\Bigl(1-\frac kM\Bigr)\widehat d+\frac kM\Bigl(\frac1k\sum_vC_v\Bigr)^2.\label{eq:raw-decomp}
\end{equation}

\begin{proposition}[AMS on the same table]\label{prop:ams}
The full distribution of $(\frac1k\sum_vC_v)^2$ is independent of the mapping and hence of ideal failure. Specifically, $\frac1k\sum_vC_v$ takes the value $d-2j$ with probability $2^{-d}\binom dj$, $j=0,\ldots,d$. The squared statistic has mean $d$ and variance $2d(d-1)$, which is $M-1$ times the unconditional variance $s(d)^2$ of $\widehat d$. By~\eqref{eq:raw-decomp}, $\widetilde d$ is a weighted sum of this squared statistic and $\widehat d$, with deterministic weights. Given the mapping, the conditional covariance of the squared statistic with $\widehat d$ is $\frac{4}{k(1-k/M)}\sum_{i<j}(H_{ij}-k^2/M)$.
\end{proposition}

\begin{proof}
The right-hand side of~\eqref{eq:ams} depends only on the signs and is independent of $\mathcal G$. It is therefore independent of every event that depends only on the mapping. $\sum_i\varepsilon_i=d-2j$ if and only if exactly $j$ signs are negative, with probability $2^{-d}\binom dj$. $\E(\sum_i\varepsilon_i)^2=d$, and the fourth moment is $3d^2-2d$, so the variance is $2d(d-1)$. \eqref{eq:raw-decomp} follows by directly rearranging~\eqref{eq:dhat} and~\eqref{eq:raw-energy}. For the covariance formula, we expand both statistics as quadratic forms in the signs and use the fact that sign products for distinct unordered pairs are uncorrelated. We multiply $(\sum_i\varepsilon_i)^2-d=2\sum_{i<j}\varepsilon_i\varepsilon_j$ by~\eqref{eq:rs-expand} and take expectations pair by pair. Averaging over failure gives $2d(d-1)c_F$, which is generally nonzero.
\end{proof}

This observation connects the design to classical AMS. The count array is a random-sign linear sketch whose sum gives exactly the AMS sketch. The AMS sketch is independent of failure, so its distribution remains unchanged however rare failure is. This independence comes at a cost in precision: the AMS estimator has variance $M-1$ times that of $\widehat d$. The estimator $\widehat d$ replaces the square of a single sum with the sum of squares over the $M$ cells, reducing the variance to $2d(d-1)/(M-1)$. Its correlation with failure requires the post-failure analysis. The covariance in Proposition~\ref{prop:ams} shows that the two estimators are generally correlated after failure, so we must allocate risk separately when combining them.

We verify Proposition~\ref{prop:ams} using the exhaustive enumeration procedure in Appendix~\ref{app:exact} for the four parameter sets $(4,3,2)$, $(4,3,4)$, $(5,3,3)$, and $(6,3,4)$, covering 144,040 states. For every mapping, $(\frac1k\sum_vC_v)^2$ has the same distribution, with mean $d$ and variance $2d(d-1)$.

\section{Hypergraph lemmas}\label{app:graph}
We prove the hypergraph results used in Theorem~\ref{thm:clt} and the discrete limit distribution: small-core counting, the characterization of failure by a duplicate edge pair when $d=o(M)$, the failure structure at fixed subcritical load, and concentration of the overlap graph. These results depend only on the unsigned mapping $\mathcal G$, which consists of $d$ labeled, independent, uniform $k$-subsets, with duplicate edges allowed.

Fix $k\ge3$. Let $K=\binom Mk$ and let $N=\sum_{i<j}\mathbf1_{\{S_i=S_j\}}$ be the number of duplicate edge pairs, with mean $\mu=\E N=\binom d2/K$. Let $Q$ be the planted distribution obtained by choosing a pair of edge labels uniformly and conditioning on the two edges mapping identically. Let $b_k=\lceil3(k-2)/2\rceil$, let $D_v$ be the degree of cell $v$, and let $\mathcal D=\{\max_vD_v\le\lceil\log M\rceil\}$.

\subsection{Small-core counting and duplicate edge pairs}
\begin{lemma}[Small-core counting]\label{lem:smallcore}
A 2-core subgraph with $e$ edges uses at most $v_e=\lfloor ke/2\rfloor$ vertices. For $d\le A_0M$ and $e\le\varepsilon M$ (with $\varepsilon<1/k$), a first-moment bound on the probability that such a subgraph occurs is
\begin{equation*}
B_e=\binom de\binom M{v_e}\Bigl(\binom{v_e}k\Big/\binom Mk\Bigr)^e\le\bigl[C(e/M)^{(k-2)/2}\bigr]^e,
\end{equation*}
where $C$ depends only on $k$ and $A_0$. Keeping the dependence on $d$ gives
\[
B_e\le(C_kd/M)^e(e/M)^{\lceil(k-2)e/2\rceil}.
\]
 For fixed $r\ge0$ and sufficiently small $\varepsilon$, $\sum_{e=3}^{\lfloor\varepsilon M\rfloor}e^rB_e=O(M^{-b_k})$.
\end{lemma}

The proof of Theorem 1 in \cite{gm} uses the same first-moment counting argument to obtain a failure probability $O(t^{-k+2})$. Analogous counting arguments for stopping sets appear in \cite{yw}. We retain the factor in $d$ for Lemma~\ref{lem:plant}, where $d=o(M)$.

\begin{proof}
Fix a set $I$ of $e$ edge labels and a set $V$ of $v_e$ vertices. The event $A_{I,V}=\{S_i\subseteq V,\ i\in I\}$ includes every core on these edge labels contained in $V$. Summing over $I$ and $V$ gives the first-moment bound $B_e$ in the lemma. Using $\binom de\le(3d/e)^e$, $\binom M{v_e}\le(3M/v_e)^{v_e}$, and $\binom{v_e}k/\binom Mk\le(v_e/M)^k$, together with $ke-v_e=\lceil ke/2\rceil$ and $v_e\le ke/2$, we obtain
\[
B_e\le3^{e+v_e}\Bigl(\frac de\Bigr)^e\Bigl(\frac{v_e}M\Bigr)^{\lceil ke/2\rceil}\le\Bigl(\frac{C_kd}M\Bigr)^e\Bigl(\frac eM\Bigr)^{\lceil(k-2)e/2\rceil},
\]
which retains the factor in $d$. Substituting $d\le A_0M$ and dropping the ceiling, since $e/M<1$, gives the first bound. Let $a=(k-2)/2$. Keeping the ceiling at $e=3$ gives $B_3=O(M^{-b_k})$. For $e\ge4$, the ratio of consecutive terms of the envelope $[C(e/M)^a]^e$ is at most $Ce^a((e+1)/M)^a$. For fixed $r\ge0$, choose $\varepsilon$ small enough that the ratio of consecutive terms of $e^rB_e$ is at most $1/2$ on $4\le e<\varepsilon M$, so that
\[
\sum_{e=3}^{\lfloor\varepsilon M\rfloor}e^rB_e=O(M^{-b_k}).
\]
This sum also covers the intermediate range with $e\to\infty$ and $e=o(M)$.
\end{proof}

The following lemma quantifies Proposition~\ref{prop:plant}. When $d=o(M)$, failure is almost always caused by a duplicate edge pair, so the failure-conditioned mapping is close to the planted distribution $Q$ in TV. Duplicate edge pairs are the bottleneck of IBLT failure, as identified in \cite{gm}. The density of $Q$ with respect to the original distribution is $N/\mu$, corresponding to size biasing by $N$ \cite{arratia2019sizebias}.

\begin{lemma}[Failure is determined by a duplicate edge pair when $d=o(M)$]\label{lem:plant}
Let $2\le d=o(M)$, regardless of whether $d\to\infty$. Then
\begin{align*}
\Pr(F\setminus\{N\ge1\})&\le C_k\Bigl(\frac{d^3}{M^{b_k+3}}+\frac{d^4}{M^{2k}}\Bigr)=o(\mu),\\
\bigl\|\mathcal L(\mathcal G\mid F)-Q\bigr\|_{TV}&\le C_k\Bigl\{\mu+dM^{k-b_k-3}+\frac{d^2}{M^k}\Bigr\}\to0.
\end{align*}
\end{lemma}

\begin{proof}
A single edge cannot form a core. Two edges form a core if and only if they map identically. Hence $\{N\ge1\}\subseteq F$, and the core has at least three edges on $F\setminus\{N\ge1\}$. Since $d=o(M)$, eventually $d\le\varepsilon M$. Thus every core of size $e\le d$ lies within the range of Lemma~\ref{lem:smallcore}, and we use the bound that retains the factor in $d$. The $e=3$ term is $O_k(d^3M^{-3-b_k})$. For $e\ge4$, we use the envelope without the ceiling, $(C_kd/M)^e(e/M)^{(k-2)e/2}$. Its ratio of consecutive terms is at most $C_k(d/M)((e+1)/M)^{(k-2)/2}$, which tends to zero uniformly over $e\le d=o(M)$. Thus $\sum_{e=4}^dB_e\le C_kd^4/M^{2k}$. Adding the two contributions gives the first bound. Since $\mu\asymp_kd^2/M^k$ and $b_k\ge k-1$, the ratio of the first bound to $\mu$ is $O_k(dM^{k-b_k-3}+d^2/M^k)\to0$.

A specified pair of edges maps identically with probability $1/K$, and two distinct edge pairs are simultaneously duplicated with probability $K^{-2}$: if the pairs share an edge, three edges map identically, and otherwise the two pairs are independent. Hence $\E[N(N-1)]\le\mu^2$, and $\mathbf1_{\{N\ge1\}}\ge N-N(N-1)/2$ gives $\mu-\mu^2/2\le\Pr(N\ge1)\le\mu$. The density of $Q$ with respect to the original distribution is exactly $N/\mu$. By the triangle inequality,
\[
\bigl\|\mathcal L(\mathcal G\mid N\ge1)-Q\bigr\|_{TV}\le\frac12\E\Bigl|\frac{\mathbf1_{\{N\ge1\}}}{\Pr(N\ge1)}-\frac{\mathbf1_{\{N\ge1\}}}{\mu}\Bigr|+\frac1{2\mu}\E(N-\mathbf1_{\{N\ge1\}})=\frac{\mu-\Pr(N\ge1)}{\mu}\le\frac\mu2.
\]
Since $\{N\ge1\}\subseteq F$, $\|\mathcal L(\mathcal G\mid F)-\mathcal L(\mathcal G\mid N\ge1)\|_{TV}=\Pr(F\setminus\{N\ge1\})/\Pr(F)\le\Pr(F\setminus\{N\ge1\})/\Pr(N\ge1)$. Since $\mu\to0$, eventually $\Pr(N\ge1)\ge\mu/2$, and the first bound makes this term $O_k(dM^{k-b_k-3}+d^2/M^k)$. Adding the two bounds gives the second bound. The proof uses only small-core counting and therefore applies to all $2\le d=o(M)$.
\end{proof}

\subsection{Failure structure at fixed subcritical load}
\begin{lemma}[Fixed subcritical load]\label{lem:propL}
Let $d/M\to\alpha\in(0,\alpha_c)$. Then $\Pr(F\setminus\{|\mathrm{core}|=2\})=O(M^{-b_k})$, $\Pr(F)=\mu\{1+O(M^{-\delta_k})\}$, and $\|\mathcal L(\mathcal G\mid F)-Q\|_{TV}=O(M^{-\delta_k})$, where $|\mathrm{core}|$ is the number of edges of the 2-core and $\delta_k=\min\{k-2,\,b_k-(k-2)\}\ge1$.
\end{lemma}

\begin{proof}
We first show that for every fixed $\varepsilon>0$ and $\beta>0$, $\Pr(|\mathrm{core}|\ge\varepsilon M)=o(M^{-\beta})$. This is a variant of a known result. For the binomial hypergraph $H(n,p;k)$, Kim \cite{kim2008poisson}, Theorem 1.7, proves that the probability that the 2-core has at least $\delta n$ vertices is exponentially small at subcritical load. The same conclusion is stated in \cite{gm}, citing \cite{dn08}, whose model assumes a given bounded degree distribution. Our mapping allows duplicate edges, which occur with polynomially small probability at linear load. We therefore prove the required bound directly using fixed-depth message passing. We follow the finite-round neighborhood analysis and concentration argument for LDPC codes \cite{richardson2001capacity}, treating cell degrees separately because they have no fixed upper bound in our model. Let $W_T$ be the number of edges whose incoming messages are all $1$ at round $T$. Then $|\mathrm{core}|\le W_T$ on every graph. Local convergence to a Poisson branching tree gives $\E W_T/M\to\alpha q_T^k$, where $q_0=1$ and $q_{t+1}=1-\exp(-k\alpha q_t^{k-1})$. When $\alpha<\alpha_c$, this recursion has no positive fixed point, so we can choose $T$ such that $\alpha q_T^k<\varepsilon/4$. On $\mathcal D$, changing $h$ edge coordinates changes $W_T$ by at most $O((k\log M)^T)h$, because the state of each edge depends only on edges within distance $T$. By the typical bounded differences inequality \cite{combes2015mcdiarmid}, $\Pr(W_T\ge\varepsilon M)\le\Pr(\mathcal D^c)+\exp\{-cM/(k\log M)^{2T}\}=o(M^{-\beta})$.

On $F\setminus\{|\mathrm{core}|=2\}$, the actual core has at least three edges. Lemma~\ref{lem:smallcore} controls the part of size below $\varepsilon M$, and the linear core bound above controls the rest, giving $\Pr(F\setminus\{|\mathrm{core}|=2\})=O(M^{-b_k})$. We complete the proof as in the second paragraph of the proof of Lemma~\ref{lem:plant}. The Bonferroni inequality gives $\Pr(F)=\mu+O(\mu^2+M^{-b_k})$. We also have $\|\mathcal L(\mathcal G\mid N\ge1)-Q\|_{TV}\le\mu/2$. Since $\{|\mathrm{core}|=2\}\subseteq\{N\ge1\}\subseteq F$, we have $F\setminus\{N\ge1\}\subseteq F\setminus\{|\mathrm{core}|=2\}$, and hence $\|\mathcal L(\mathcal G\mid F)-\mathcal L(\mathcal G\mid N\ge1)\|_{TV}\le\Pr(F\setminus\{N\ge1\})/\Pr(N\ge1)\le\Pr(F\setminus\{|\mathrm{core}|=2\})/\Pr(N\ge1)$. Using $\mu\asymp_kM^{2-k}$, the two terms combine to give $O(M^{-\delta_k})$. This total variation bound applies to the entire labeled mapping.
\end{proof}

\subsection{Superpolynomial concentration of the overlap graph}
\begin{lemma}[Superpolynomial graph concentration]\label{lem:mdconc}\label{lem:conc}
Fix $k\ge3$, and let $d/M$ lie in a fixed compact interval of $(0,\infty)$. Write $\Omega=\sum_{i<j}H_{ij}$, $\Omega_2=\sum_{i<j}H_{ij}^2$, $\ell_i=\sum_{j\ne i}H_{ij}$, and $\Omega_\ell=\sum_i\ell_i^2$. There is a deterministic set of graphs $\mathcal G_M$, called good graphs. Its complement has probability smaller than any fixed negative power of $M$. On a good graph, all three statistics deviate from their exact means by at most $O(\sqrt M(\log M)^4)$, and
\begin{equation*}
\Omega_2=\binom d2\left\{\frac{k^2}M+\frac{k^2(k-1)^2}{M(M-1)}\right\}+O(\sqrt M(\log M)^4),\qquad
\max_i\ell_i\le k\lceil\log M\rceil.
\end{equation*}
\end{lemma}
Intuitively, a good graph is a typical realization of the random mapping: each key intersects only a logarithmic number of other keys, and the aggregate overlap quantities are all close to their means. Graphs outside $\mathcal G_M$ are called \emph{bad graphs}, and their probability is superpolynomially small. Appendix~\ref{app:clt} applies a CLT to the random signs on good graphs, using precisely the fact that no single key can dominate $\sum_{i<j}H_{ij}\varepsilon_i\varepsilon_j$.

\begin{proof}
\emph{Deterministic means.} For the overlap $H$ of independent uniform $k$-subsets, write $h_1=\E H=k^2/M$ and $h_2=\E H^2=k^2/M+k^2(k-1)^2/[M(M-1)]$. The exact means are
\[
\E\Omega=\tbinom d2h_1,\qquad\E\Omega_2=\tbinom d2h_2,\qquad\E\Omega_\ell=d\{(d-1)h_2+(d-1)(d-2)h_1^2\}.
\]
The last identity holds because, given $S_i$, the $H_{ij}$ are independent and their conditional first and second moments do not depend on $S_i$.

\emph{Degree truncation and the Lipschitz bound.} Since $D_v\sim\mathrm{Bin}(d,k/M)$, we have $\Pr(\mathcal D^c)\le M(eC/\lceil\log M\rceil)^{\lceil\log M\rceil}$, which is superpolynomially small. For two graphs in $\mathcal D$, we successively replace each edge that differs by the edge with the same label in the other graph. The incidence set of every intermediate graph lies in the union of the two endpoint graphs, so every vertex has degree at most $2\lceil\log M\rceil$. Replacing one edge in such an intermediate graph affects only the $O(\log M)$ edge labels overlapping the old or new edge, and each $H_{ij}\le k$, so both $\Omega$ and $\Omega_2$ change by $O(\log M)$. For the replaced edge, $\ell_i=O(\log M)$ and its square changes by $O((\log M)^2)$; the $\ell_j$ values of the other $O(\log M)$ affected labels change by $O(1)$ each, and their squares by $O(\log M)$. Thus the total change in $\Omega_\ell$ is $O((\log M)^2)$. All three statistics are therefore $C(\log M)^2$-Lipschitz on $\mathcal D$ with respect to the Hamming distance over edge labels.

\emph{Concentration.} The typical bounded differences inequality \cite{warnke2016typical} gives concentration for functions that are Hamming Lipschitz on a high-probability set; we use the form due to Combes \cite[Proposition~2]{combes2015mcdiarmid}. We take a Lipschitz extension of $f$ from $\mathcal D$, truncate it to a deterministic nonnegative polynomial range, and apply McDiarmid's inequality to the independent edges. The bad set has superpolynomially small probability. Both $f$ and its truncated extension are at most $O(M^3)$, so the change in expectation is also superpolynomially small. Hence
\[
\Pr\bigl(|f-\E f|>C\sqrt M(\log M)^4\bigr)\le\Pr(\mathcal D^c)+2e^{-c(\log M)^4},\qquad f=\Omega,\Omega_2,\Omega_\ell.
\]

We take $\mathcal G_M$ to be the intersection of $\mathcal D$ with the three deviation events, which gives the stated concentration bound. Since $\ell_i\le k\lceil\log M\rceil$, the lemma follows.
\end{proof}

\section{The Gaussian limit after failure and the Gaussian rule}\label{app:clt}
We prove Theorem~\ref{thm:clt} and Corollary~\ref{cor:gauss}. We use the notation $D_v$, $\ell_i$, and $\Omega_2$ from Appendix~\ref{app:graph}. We first prove a Gaussian limit for the quadratic form in the random signs on a deterministic good graph, uniformly over good graphs (Lemma~\ref{lem:cumulant}). We then show that good graphs retain almost all probability after conditioning on failure. For the $d=o(M)$ and overload regimes, we use the Gaussian limit under deterministic coefficients (Lemma~\ref{lem:plantclt}). This limit is uniform in the coefficients and holds under both the random mapping and the planted distribution. Finally, we consider the possible limits of $d/M$ and combine the cases through a subsequence argument.

\subsection{CLT for the random signs on a good graph}
\begin{lemma}\label{lem:cumulant}
Fix $k\ge3$. Let a deterministic sequence of graphs satisfy $\Omega_2/M\to\beta\in(0,\infty)$ and $\max_i\ell_i=O(\log M)$. Then $\Omega_2^{-1/2}\sum_{i<j}H_{ij}\varepsilon_i\varepsilon_j\Rightarrow\mathcal N(0,1)$, uniformly over graphs satisfying the same constant bounds.
\end{lemma}

\begin{proof}
There are at most $\Omega_2=O(M)$ nonzero summands, each bounded by $k$. Two summands can be dependent only if they share a key, so the dependency graph has maximum degree at most $2\max_i\ell_i=O(\log M)$. Distinct summands have zero covariance (Appendix~\ref{app:encoding}), so the sum has variance exactly $\Omega_2$. For any fixed order $r\ge3$, a joint cumulant can be nonzero only if the summands form a connected group in the dependency graph \cite{janson1988cumulant}. There are $O_r(M(\log M)^{r-1})$ connected groups, each with a bounded cumulant. Dividing by $\Omega_2^{r/2}\asymp M^{r/2}$ gives $O_r(M^{1-r/2}(\log M)^{r-1})\to0$. The first cumulant is zero and the second is one. The standard Gaussian is determined by its moments, so the method of moments gives convergence. These bounds depend only on the range of $\beta$ and the constant in $\max_i\ell_i$. They are therefore uniform over graphs satisfying the same constant bounds. If CDF convergence were not uniform, we could extract a sequence of deterministic good graphs along which convergence fails, contradicting the argument above.
\end{proof}

Lemma~\ref{lem:cumulant} transfers to $X$ and $\widetilde d$. The difference between $\widehat d-d$ and $\frac{2}{k(1-k/M)}$ times this sum is
\begin{equation*}
-\frac{k}{M(1-k/M)}\Bigl[\Bigl(\sum_i\varepsilon_i\Bigr)^2-d\Bigr],
\end{equation*}
which depends only on the signs and has $L^2$ norm $O(d/M)$ for every graph. After division by $s(d)\asymp d/\sqrt M$, this norm is $O(M^{-1/2})$. On a good graph from Lemma~\ref{lem:conc}, $\Omega_2$ equals its exact mean plus $O(\sqrt M(\log M)^4)$. Thus the ratio of $\frac{4\Omega_2}{k^2(1-k/M)^2}$ to $s(d)^2$ tends to one, and we can replace normalization by $\sqrt{\Omega_2}$ in this lemma with normalization by $s(d)$. The standardized statistic $X$ for $\widehat d$ therefore has the same uniform Gaussian limit on good graphs. By \eqref{eq:raw-decomp}, the difference between $\widetilde d$ and $\widehat d$ is $\frac kM[(\sum_i\varepsilon_i)^2-\widehat d]$. Given the graph, $(\sum_i\varepsilon_i)^2-d$ has $L^2$ norm $\sqrt{2d(d-1)}$, and $\widehat d$ is conditionally unbiased with conditional variance at most $(M-1)s(d)^2=2d(d-1)$. By the triangle inequality, the difference has conditional $L^2$ norm at most $\frac{2k}M\sqrt{2d(d-1)}$. After division by $s(d)$, this is at most $2k\sqrt{M-1}/M\to0$. Hence the same Gaussian limit holds for $\widetilde d$.

\subsection{Conditioning on failure at a positive load ratio}\label{app:clt-cond}
In this subsection, we assume that $d/M$ lies in a fixed compact interval of $(0,\infty)$. For $d\ge2$, a specified duplicate edge pair forms a core, so
\begin{equation*}
\Pr(F)\ge\Pr(S_1=S_2)=\binom Mk^{-1}.
\end{equation*}
By Lemma~\ref{lem:conc}, $\Pr(\mathcal G\notin\mathcal G_M\mid F)\le\Pr(\mathcal G\notin\mathcal G_M)\binom Mk\to0$. The conditional CDF error is at most the uniform error on good graphs plus this conditional exceptional probability. Thus, at a positive load ratio,
\begin{equation*}
\sup_x\bigl|\Pr(X\le x\mid F)-\Phi(x)\bigr|\to0.
\end{equation*}
The same argument applies to any event that depends only on the mapping and has probability at least a fixed negative power of $M$. We divide the superpolynomially small exceptional probability by $\Pr(F)$ and average the CLT error on good graphs under the conditional distribution. The good-graph CLT therefore requires only convergence to zero, regardless of how rare failure is. A typical failure at subcritical load arises from a duplicate edge pair, which sets just one $H_{ij}$ to $k$. The resulting change in $\Omega_2$ is $O(1)$, far below the deviation allowed by Lemma~\ref{lem:conc}, so the graph remains in $\mathcal G_M$.

The conditional variance is also asymptotic to $s(d)^2$. On good graphs, the ratio tends to one by \eqref{eq:condmoments} and Lemma~\ref{lem:conc}. On any graph, the ratio is at most $M-1$, and the superpolynomially small probability of bad graphs controls their contribution.

\subsection{Gaussian limit under deterministic coefficients}
The following Gaussian limit holds uniformly over deterministic coefficient vectors $(s_i)\in\{\pm1\}^d$. We apply it to random sign encoding by conditioning on the iid coding coefficients and then averaging over the signs.

\begin{lemma}[Gaussian limit under deterministic coefficients]\label{lem:plantclt}
For deterministic coefficients $(s_i)\in\{\pm1\}^d$, let $C_v=\sum_{i:v\in S_i}s_i$, compute $\widehat d$ by \eqref{eq:dhat}, and define $X$ by \eqref{eq:X}. If $d/\sqrt M\to\infty$, then $X\Rightarrow\mathcal N(0,1)$, uniformly in the coefficients, with $d/M$ allowed to be unbounded. Under the planted distribution $Q$, we also have $X\Rightarrow\mathcal N(0,1)$, uniformly in the coefficients and the label pair.
\end{lemma}

\begin{proof}
We use $w_i$ from Appendix~\ref{app:proof-moments}. We expose edges in label order. The martingale differences are $\Delta_j=\frac{2}{k(1-k/M)}s_jw_j^\top\sum_{i<j}s_iw_i$, with $\widehat d-d=\sum_j\Delta_j$. Their conditional variances sum to
\begin{gather*}
\mathcal V_d=s(d)^2+\frac{8}{k(1-k/M)(M-1)}\sum_{i<h}(d-h)s_is_h\langle w_i,w_h\rangle,\\
\Var\mathcal V_d=\frac{64}{(M-1)^3}\sum_{h=2}^d(h-1)(d-h)^2=O\Bigl(\frac{d^4}{M^3}\Bigr).
\end{gather*}
Since $s(d)^4\approx4d^4/M^2$, the normalized bracket $\mathcal V_d/s(d)^2$ has variance $O(1/M)$ and converges to one in $L^2$. Given the first $j-1$ edges, the components of $\sum_{i<j}s_iw_i$ sum to zero. Thus $w_j^\top\sum_{i<j}s_iw_i$ is the sum of this vector's components over the $k$ cells of $S_j$. Its conditional fourth moment is at most $k^3\cdot\frac kM$ times the sum of the fourth powers of all components. The component at cell $v$, $\sum_{i<j}s_i(\mathbf1\{v\in S_i\}-k/M)$, is a sum of $j-1$ independent, centered, bounded terms, each with variance at most $k/M$. Rosenthal's inequality \cite{rosenthal1970} bounds its fourth moment by $C_k(j/M+j^2/M^2)$. Hence $\E\Delta_j^4\le C_k(j/M+j^2/M^2)$. Summing over $j\le d$ and dividing by $s(d)^4$, we bound the fourth-order Lyapunov remainder by
\begin{equation*}
\frac{\sum_j\E\Delta_j^4}{s(d)^4}\le C_k\Bigl(\frac M{d^2}+\frac1d\Bigr)\to0,
\end{equation*}
and the martingale CLT gives $X\Rightarrow\mathcal N(0,1)$. Both bounds depend only on $|s_i|=1$, so convergence is uniform in the coefficients and holds for unbounded $d/M$.

Under the planted distribution, the statistic is invariant under label permutation. We may therefore fix the duplicate edge pair as 1, 2, with $S_1=S_2$. Let $Y_{\ge3}$ be $\widehat d$ computed from the remaining $d-2$ edges alone, minus $d-2$. Then
\begin{equation*}
\widehat d-d=Y_{\ge3}+2s_1s_2+\frac{2}{k(1-k/M)}(s_1+s_2)\sum_{j\ge3}s_j\Bigl(|S_1\cap S_j|-\frac{k^2}{M}\Bigr).
\end{equation*}
The last term has variance $O(d/M)$. After division by $s(d)$, the deterministic term $2s_1s_2$ and the last term have mean squares $O(M/d^2)$ and $O(1/d)$, respectively. Both tend to zero. We apply the unconditional result to $Y_{\ge3}$.
\end{proof}

\subsection{\texorpdfstring{Proof of Theorem~\ref{thm:clt} (the Gaussian limit after failure) and Corollary~\ref{cor:gauss}}{Proof of the Gaussian limit after failure and the Gaussian rule}}
Assume $d/\sqrt M\to\infty$. From any subsequence, we take a further subsequence along which $d/M$ converges in $[0,\infty]$.

If the limit is positive and finite, we apply the conclusion of Appendix~\ref{app:clt-cond}.

If the limit is zero, Lemma~\ref{lem:plant} shows that the TV distance between the failure-conditional distribution and the planted distribution $Q$ tends to zero for all $2\le d=o(M)$. This graph-theoretic result continues to hold after we add independent signs. Lemma~\ref{lem:plantclt} holds uniformly in deterministic signs under both the random mapping and the planted distribution. We condition on the iid coding coefficients, apply these uniform results, and average over the signs to obtain the conditional Gaussian limit.

If the limit is infinite, eventually $d>M$. We can assign each edge removed during full peeling a cell whose degree is one at the time of removal. That cell is then empty and can witness no further removal, so full peeling requires $d\le M$. Thus $F$ occurs with certainty in this branch, and we apply Lemma~\ref{lem:plantclt} directly under the random mapping.

All three branches give uniform convergence of the conditional CDF to $\Phi$, and a subsequence contradiction establishes \eqref{eq:clt} for the full sequence. For Corollary~\ref{cor:gauss}, $U$ is the largest candidate that is not rejected, so $U\ge d$ whenever $d$ is not rejected, and hence $\{U<d\}\subseteq\{\widehat d<t(d)\}=\{X<(t(d)-d)/s(d)\}$. The standardized chi-square quantile $(t(d)-d)/s(d)$ converges to $\Phi^{-1}(u)$, so~\eqref{eq:clt} gives $\limsup\Pr(U<d\mid F)\le u$. \qed

\subsection{\texorpdfstring{Monotonicity of the acceptance lower bound $t(n)$ in Definition~\ref{def:gauss}}{Monotonicity of the acceptance lower bound of the Gaussian rule}}\label{app:gauss-mono}
We have $t(n)=n+q\,s(n)$ with $q=[\chi^2_{M-1,u}-(M-1)]/\sqrt{2(M-1)}$. For $q\ge0$, $t$ is increasing. For $q<0$, differentiation for $n\ge2$ gives $\frac{d}{dn}s(n)\le\frac{\sqrt2(2n-1)}{2\sqrt{n(n-1)(M-1)}}\le\frac{3}{2\sqrt2}\cdot\frac{\sqrt2}{\sqrt{M-1}}$. A sufficient condition is therefore $|q|\sqrt{2/(M-1)}<\sqrt8/3$. The Chernoff bound for the chi-square lower tail \cite[Lemma~1]{laurent2000adaptive} gives $|q|\le\sqrt{2\log(1/u)}\le3.035$ for $u\ge0.01$. The condition holds when $M\ge64$. \qed

By monotonicity, the accepted candidates form $\{2,\dots,U\}$. For $q<0$, $s(n)\le n\sqrt{2/(M-1)}$ gives $t(n)\ge n\,\chi^2_{M-1,u}/(M-1)$, and for $q\ge0$, $t(n)\ge n$; in both cases $t(n)\to\infty$, so $U$ is finite. The case $+\infty$ in Definition~\ref{def:gauss} therefore arises only when no candidate is accepted, that is, $\widehat d<t(2)=2+2q/\sqrt{M-1}$. For $u=0.01$, $t(2)$ is about 1.73, 1.86, and 1.96 at $M=256$, 1024, and 16384, respectively. A post-failure observation below $t(2)$ essentially occurs only when $d=2$ or $3$ and the duplicate edge pair has opposite signs. For $d=2$, $\Pr(F)=1/\binom Mk$; for $d=3$ and $k=3$, enumerating the union of three edges and running peeling shows that the exact $\Pr(F)$ differs from the duplicate-pair term $3/\binom M3$ by a relative error below $10^{-3}$ for $M\ge256$. At $M=1024$, the failure probability is $1.68\times10^{-8}$.

\section{Proofs of the discrete limit distribution and the Skellam rule}\label{app:inversion}
We prove Theorems~\ref{thm:cross}, \ref{thm:uniform}, and \ref{thm:coverage}. Write the iid coding coefficients as $\varepsilon_i$. By \eqref{eq:raw-expand}, $\widetilde d-d=\frac2k\sum_{i<j}H_{ij}\varepsilon_i\varepsilon_j$ holds for every sample. The final subsection discusses an extension with multiple sign counts per cell.

\subsection{\texorpdfstring{Proof of Theorem~\ref{thm:cross} (discrete limit for bounded $c$)}{Proof of Theorem: discrete limit for bounded c}}\label{app:proof-cross}
\begin{proof}
\emph{From failure to the planted distribution.} Lemma~\ref{lem:plant} gives $\|\mathcal L(\mathcal G\mid F)-Q\|_{TV}\to0$ for all $2\le d=o(M)$. Here $Q$ selects a pair of edge labels uniformly and then conditions on the two edges having the same mapping. For $2\le d\le C\sqrt M$, both $\mu$ and $d^2/M^k$ are $O_k(M^{1-k})$. Since $b_k=\lceil3(k-2)/2\rceil$, we have $k-b_k-3=-\lceil k/2\rceil\le-2$ and hence $dM^{k-b_k-3}=O_k(d/M^2)$. The distance is therefore at most $O_k(d/M^2+M^{1-k})$. Adding independent signs and applying \eqref{eq:raw-expand} preserve this TV bound. The statistic is invariant under label permutation, so we may fix the duplicate edge pair as 1, 2. Their common edge and the remaining $d-2$ edges are independent and uniform. For $d=2$, the null distribution is exact under the planted distribution. We treat empty sums below as zero.

\emph{Three negligible event classes for bounded $c$.} By the union bound, the common edge of the duplicate edge pair intersects a remaining edge with probability $O_k(d/M)$. Any specified pair of remaining edges shares at least two cells with probability $O_k(M^{-2})$. Summing over all such pairs gives $O_k(d^2/M^2)$. In the intersection graph of the remaining edges, for any three specified edges, the probability that one specified edge intersects both others is
\begin{equation*}
\Bigl(1-\binom{M-k}{k}\Big/\binom Mk\Bigr)^2=O_k(M^{-2}),
\end{equation*}
This equality follows by conditioning on the edge shared by the two intersection tests. The other two edges are independent, and their intersection probability with the shared edge is independent of its position. Summing over ordered triples bounds the probability that the intersection graph is not a matching by $O_k(d^3/M^2)$. Thus, outside an event with probability bounded by \eqref{eq:tv-error}, the common edge of the duplicate edge pair is disjoint from the remaining edges, and all overlaps among the remaining edges are single-cell overlaps forming a matching.

\emph{Poisson approximation for overlaps among the remaining edges.} For each pair of remaining edges, consider the indicator that they intersect. These indicators have a dependency graph defined by shared edge labels. Collections of indicators with disjoint label sets are independent. Each indicator has mean $p=1-\binom{M-k}k/\binom Mk=k^2/M+O_k(M^{-2})$ and $O(d)$ neighbors. By the calculation in the previous step, two indicators sharing an edge have joint probability exactly $p^2$. We apply the local dependence bound of Arratia, Goldstein, and Gordon \cite[Theorem~1]{arratia1989chen} with the neighborhood containing the indicator itself and its neighbors. Indicators outside the neighborhood depend only on the remaining independent edges, so the third term of the bound is zero. The bound is a universal constant times the sum of products of means within neighborhoods plus the sum of joint probabilities over distinct neighbors. Both sums are $O_k((d^2+d^3)/M^2)$ here. Thus the number of intersecting edge pairs among the remaining edges has TV distance $O_k((d^2+d^3)/M^2)$ from $\Pois(\binom{d-2}2p)$. We replace the parameter by $\binom{d-2}2k^2/M=2\lambda_d$. The TV distance between two Poisson distributions is at most the difference between their parameters: we couple the larger-parameter variable to the smaller one by adding an independent Poisson increment, and the variables differ only when this increment is nonzero. The replacement adds an error of $O_k(d^2/M^2)$.

\emph{Independence of the overlap sign marks and the Skellam limit.} On the matching event, each overlap uses a pair of iid signs shared with no other overlap. Their products are therefore mutually independent and symmetric, and independent of the duplicate edge pair's sign product $\varepsilon_1\varepsilon_2$. Given the mapping, $\widetilde d$ depends only on the number of overlaps among the remaining edges and these independent $\pm1$ marks. To formalize the coupling, we retain the true overlap count and generate that many new independent marks for each mapping. On the matching event, the old and new marked sums have the same conditional distribution. On all other mappings, we couple them arbitrarily, with a distributional error at most the exceptional event's probability. We then apply the same marking kernel to the Poisson approximation of the count, preserving the TV bound. By Poisson thinning, $\Pois(2\lambda_d)$ independent $\pm1$ marks split by sign into two independent $\Pois(\lambda_d)$ counts. Their difference, $Z_d$, has a Skellam distribution. Adding the duplicate edge pair's contribution $2\varepsilon_1\varepsilon_2$ gives $J_{\lambda_d}$ in \eqref{eq:J}. Thus $\widetilde d$ is close to the null distribution $d+J_{\lambda_d}$ in \eqref{eq:aux}. We use the Poisson approximation for the unconditional count and the matching event to bound the coupling error.

The four steps give a total error of $O_k(d/M+d^2/M^2+d^3/M^2+d/M^2+M^{1-k})$. Since $d\le M$ and $k\ge3$, the terms $d/M$ and $d^3/M^2$ absorb all others, giving \eqref{eq:tv-error}. When $d/\sqrt M\to c$, we have $\lambda_d\to k^2c^2/4$. Continuity in the Poisson parameter gives convergence of $J_{\lambda_d}$ in total variation to $J_{k^2c^2/4}$, proving Theorem~\ref{thm:cross}(ii). The $\pm2$ and Skellam terms are independent, so $\Var J_\lambda=4+8\lambda/k^2$ and the limiting variance is $4+2c^2$. For (iii), as $\lambda\to0$, we have $\Pr(Z\ne0)\le2\lambda\to0$ and $J_\lambda\Rightarrow2\varepsilon_1\varepsilon_2$. As $\lambda\to\infty$, we have $Z/\sqrt{2\lambda}\Rightarrow\mathcal N(0,1)$ and $\sqrt{8\lambda}/k=\sqrt2\,c$. The $\pm2$ term divided by $\sqrt2\,c$ tends to zero, so $J_{k^2c^2/4}/(\sqrt2\,c)\Rightarrow\mathcal N(0,1)$. The estimator $\widehat d$ has the same weak limit: the AMS term in \eqref{eq:raw-decomp} has coefficient $k/M$, and $(\sum_i\varepsilon_i)^2=O_p(d)=O_p(\sqrt M)$.
\end{proof}

The exact moments of $\widehat d$ under the planted distribution agree with the limiting variance $4+2c^2$ in (ii). Let $S_1=S_2$, with the remaining $d-2$ edges independent and uniform. If the two coding coefficients have opposite signs, the edges cancel completely, leaving the remaining $d-2$ elements. If they are equal, the edges combine into a single edge with coefficient $\pm2$. For $d\ge2$,
\begin{equation*}
\E[\widehat d\mid S_1=S_2]=d,\qquad
\Var(\widehat d\mid S_1=S_2)=\frac{2d(d-1)}{M-1}+4-\frac4{M-1}.
\end{equation*}
The component with opposite coefficients has variance $2(d-2)(d-3)/(M-1)$, and the component with equal coefficients has variance $2(d-2)(d+5)/(M-1)$. Averaging these variances and adding $4$ from the difference between the component means $d\mp2$ gives the result.

\subsection{Sparse limit for fixed small \texorpdfstring{$d$}{d}}
\begin{proposition}[Limit when $d=o(\sqrt M)$]\label{prop:smalld}
Fix $k\ge3$ and let $2\le d=o(\sqrt M)$. Conditional on failure, both $\widehat d-d$ and $\widetilde d-d$ converge weakly to the distribution assigning equal probability to $\pm2$. The Gaussian rule's conditional miss rate tends to zero for $d=2,3$, and to $1/2$ for fixed $d\ge4$ or $d\to\infty$.
\end{proposition}

\begin{proof}
Theorem~\ref{thm:cross} gives the equally weighted $\pm2$ limit for the fluctuation of $\widetilde d$ when $d=o(\sqrt M)$. By \eqref{eq:raw-decomp}, the correction for $\widehat d$ is $O_p(d^2/M)=o_p(1)$, so its fluctuation has the same weak limit. Consider the event that the remaining edges are pairwise disjoint and disjoint from the duplicate edge pair. On this event, $\widetilde d$ equals $d+2\varepsilon_1\varepsilon_2$ exactly, and the event has conditional probability tending to one. The Gaussian rule has $t(d)=d-O(s(d))$ with $s(d)\to0$. The negative component leads to rejection of the true candidate, and the positive component leads to acceptance. For $d=2,3$, the negative component gives an observation below $t(2)$, so the rule returns $+\infty$, which is not a miss. For $d\ge4$, the negative component still leads to acceptance of candidate $d-2$, causing a miss. For $d=2$, failure is equivalent to a duplicate edge pair, and $\Pr(\widehat d=0\mid F)=\Pr(\widehat d=4\mid F)=1/2$ for every $M>k$. Thus multiplying $\widehat d$ by any fixed constant to obtain an upper bound leaves a post-failure risk of at least $1/2$.
\end{proof}

\subsection{\texorpdfstring{Proof of Theorem~\ref{thm:uniform} (uniform CDF approximation)}{Proof of Theorem: uniform CDF approximation}}\label{app:proof-uniform}
\begin{proof}
First consider $d/\sqrt M\to\infty$. Appendix~\ref{app:clt} gives a conditional Gaussian limit for $\widetilde d$ standardized by $s(d)$. The null distribution is the distribution of $d+J_{\lambda_d}$. As $\lambda_d\to\infty$, we can express $Z_d$ in the Skellam term as a sum of iid increments whose number grows with $\lambda_d$. The CLT gives $(2/k)Z_d/\sqrt{8\lambda_d/k^2}\Rightarrow\mathcal N(0,1)$. The ratio of $8\lambda_d/k^2=2(d-2)(d-3)/M$ to $s(d)^2$ tends to one, and the $\pm2$ term divided by $s(d)$ tends to zero. Hence $J_{\lambda_d}/s(d)\Rightarrow\mathcal N(0,1)$. Both standardized CDFs converge to the continuous CDF $\Phi$. Weak convergence to a continuous limit implies uniform CDF convergence (P\'olya's theorem), so their Kolmogorov distance tends to zero.

To prove uniformity, suppose that there exist $\delta>0$ and a sequence $(M,d)$ along which the distance is at least $\delta$. Take a subsequence with $d/\sqrt M\to c\in[0,\infty]$. If $c<\infty$, this subsequence eventually satisfies $d\le C\sqrt M$ for some constant, and Theorem~\ref{thm:cross} gives a TV distance tending to zero. If $c=\infty$, the previous paragraph gives convergence to zero. Both cases contradict the lower bound $\delta$.
\end{proof}

\subsection{\texorpdfstring{Proof of Theorem~\ref{thm:coverage} (uniform asymptotic bound on the conditional miss rate)}{Proof of Theorem: uniform bound on the conditional miss rate}}\label{app:proof-coverage}
\begin{proof}
Write $G_n$ for the CDF of the null distribution at candidate $n$ and $Y$ for the observation. The rule rejects $n$ if and only if $G_n(Y)\le u$. We use the standard argument for test inversion and the conservativeness of a discrete p-value \cite{neyman1937}. If $Y\sim G_n$, then $\Pr(G_n(Y)\le u)\le u$. If no lattice point satisfies $G_n(y)\le u$, the rejection event is empty. Otherwise, let $y^*$ be the largest lattice point satisfying $G_n(y)\le u$. The rejection event is then $\{Y\le y^*\}$, with probability $G_n(y^*)\le u$. Since the rejection region $\{Y\le y^*\}$ is a lower set, the rejection probabilities under the true distribution and $G_n$ differ by at most their Kolmogorov distance. Taking $n=d$ and applying Theorem~\ref{thm:uniform}, we obtain a rejection probability of at most $u+o(1)$ for the true candidate, uniformly over all $d\ge2$. The event $\{U<d\}$ implies rejection of the true candidate, giving \eqref{eq:coverage}. The argument applies to arbitrary acceptance sets and allows candidate rejection probabilities to vary nonmonotonically with $n$.
\end{proof}

Using $\le u$ as the rejection criterion with the left limit $G_n(Y-)=\Pr(Y^{\prime}<Y)$ (where $Y^{\prime}\sim G_n$) gives a different rule. The equally weighted $\pm2$ distribution has left limit zero at $-2$. This rejects its smallest support point, which has probability $1/2$, and the conditional miss rate tends to $1/2$ when $d\ge4$ and $d=o(\sqrt M)$. At that point, $G_n$ equals $1/2>u$, so the Skellam rule accepts the candidate.

For $d=3$, the observation may lie outside the support of $G_3$, and $\widetilde d\ge1$ always holds. If all three coefficients have the same sign, $\sum_vC_v^2\ge3k$. Otherwise, label the edge whose sign differs from the other two as edge $3$. The sum of squares is $3k+2(H_{12}-H_{13}-H_{23})$. Among the $k$ cells of edge $3$, at least $H_{13}+H_{23}-H_{12}$ are shared with edge $1$ or edge $2$. Thus $H_{13}+H_{23}-H_{12}\le k$, and the sum of squares is at least $k$. Hence $G_3(\widetilde d)\ge G_3(1)=1/2>u$, and the true candidate is always accepted.

\subsection{Magnitudes of the terms in the two-round risk decomposition}\label{app:risk}
This subsection derives the four terms of Theorem~\ref{thm:protocol} and gives their numerical values.

\paragraph{Sizing error.} The first term is the joint miss rate $\Pr(F_1)\Pr(U_b<d\mid F_1)$. Under the Skellam rule, Theorem~\ref{thm:coverage} makes it asymptotically at most $u\Pr(F_1)$ uniformly over all $d\ge2$; for a one-sided difference $U_b\ge|b|=d$, so this term is zero. The Gaussian rule's conditional miss rate can far exceed $u$ when $c$ is bounded, but failure itself is then rare: by Proposition~\ref{prop:plant}, when $d=o(M_1)$, $\Pr(F_1)$ is asymptotically equal to the expected number of duplicate edge pairs $\binom d2/\binom{M_1}k$, which is about $2.8\times10^{-6}$ for $k=3$, $M_1=1024$, and $d=\sqrt{M_1}=32$, so the joint miss rate remains small.

\paragraph{Peeling failure with sufficient capacity.} The second term comes from the randomness of peeling: even with sufficient second-round capacity, the random hypergraph can leave a nonempty 2-core with small probability. The second-round load is at most $0.8\alpha_c$, so by Lemma~\ref{lem:propL}, $\max_{m\ge m(d)}\Pr_m(F)$ decays polynomially in $d$.

\paragraph{Unbounded upper bound.} When the rule gives no finite upper bound, the protocol cannot size the second-round table by $m(U_b)$, and~\eqref{eq:risk} counts this case as a failure. The Skellam rule always accepts candidate $2$ and returns $+\infty$ only when $M_1\le4\log(1/u)$ ($M_1\le18$ for $u=0.01$) (Appendix~\ref{app:search-cut}); at practical table sizes this term is zero. The Gaussian rule's acceptance lower bound $t(n)$ increases with $n$ and returns $+\infty$ only when the observation lies below $t(2)$ and no candidate is accepted; for $M_1=1024$ and $u=0.01$, $t(2)\approx1.86$. A post-failure observation in this range essentially occurs only when $d=2$ or $3$ and the duplicate edge pair has opposite signs, and then $\Pr(F_1)$ is at most $2\times10^{-8}$ for $k=3$ and $M_1=1024$ (Appendix~\ref{app:gauss-mono}).

\paragraph{Checksum error.} The checksum takes a non-pure cell as pure, or indicates termination on a nonempty core, making the recovered result incorrect. Each check passes falsely with probability at most $2^{-64}$, and the number of checks is at most the number of cells plus $k$ times the number of deleted keys; the keySum and count fields of a nonempty core can both cancel to zero, and the probability that all checksums are then zero is likewise at most $2^{-64}$, which can occur at most once at termination of each round. Hence $\Pr(E)$ is of order $2(M_1+m+kd)\cdot2^{-64}$, negligible at practical scales. With 32-bit checksums this term is of order $2(M_1+m+kd)\cdot2^{-32}$, of order $10^{-4}$ for $M_1=16384$ and load at most 3; since $\widetilde d$ uses only the count field, the first term is unaffected by checksum width, and at deployment the target failure probability can be split between the sizing and checksum terms.

\subsection{Extension to multiple sign counts}\label{app:multi}
Suppose each cell stores $L$ copies of the count, using independent sign hashes and a shared mapping. Conditional on $\mathcal G$, the estimates $\widehat d$ from these copies are mutually independent, each with mean $d$. Averaging independent estimates is the AMS approach to reducing variance \cite{alon1999spacecomplexity}. From $\E[\widehat d^2\mid\mathcal G]\le3d^2$ and $e^{-x}\le1-x+x^2/2$, we obtain
\begin{equation*}
\E\Bigl[\exp\Bigl(-\frac{\widehat d}{5d}\Bigr)\Bigm|\mathcal G\Bigr]\le\frac{43}{50},
\end{equation*}
and the exponential Markov inequality (Chernoff bound) applied to $U=\frac2L\sum_\ell\widehat d^{(\ell)}$ gives $\Pr(U<d\mid\mathcal G)\le(\frac{43}{50}e^{1/10})^L$. This bound holds for every deterministic mapping, hence for arbitrarily rare failure. For $L=272$, the right-hand side is below $10^{-6}$. When only an extra scalar AMS accumulator is stored, the $L$ count sums are independent given $d$, and their joint distribution is a product of binomial distributions. This permits an exact discrete candidate test with the finite-sample guarantee $\Pr(U<d\mid F)\le u$.

Repeating counts reduces absolute error but preserves the ratio of the conditional scale to the unconditional scale. Fix $d\ge4$ and $L$, and let $M\to\infty$. If we apply Gaussian inversion to the average of $L$ copies using the unconditional standard deviation $s(n)/\sqrt L$, the conditional miss rate tends to $2^{-L}\sum_{j<L/2}\binom Lj$. This equals $1/2,1/4,5/16,93/256$ for $L=1,2,4,8$, respectively. The reason is that for fixed small $d$, failure is caused by a duplicate edge pair, which contributes $\pm2$ independently in each copy, while the acceptance lower bound tends to $d$ as $s(d)\to0$. Whenever more than half of the copies take $-2$, the average falls below $d$ and the true candidate is rejected. This probability tends to $1/2$ as $L$ grows. Increasing the number of copies therefore cannot replace the correct post-failure distribution, which is why the main text uses a single count with the discrete limit distribution.

\section{Exact evaluation and search cutoffs}\label{app:search}
This appendix describes the computation of the Skellam rule (Definition~\ref{def:skellam}). The guarantee of the Skellam rule does not require the accepted candidates to form an interval, and by definition $U$ is the largest accepted candidate. Rejected candidates may lie between accepted ones, so a binary search can miss the largest accepted candidate. The implementation therefore scans downward from an explicit cutoff and returns the first accepted candidate; under the condition of Proposition~\ref{prop:cut}, the result equals that of an unbounded scan. Appendix~\ref{app:search-int} evaluates $G_n(\widetilde d)$ exactly using integer thresholds; Appendix~\ref{app:search-cut} gives the cutoff and proves that every candidate above it is rejected; Appendices~\ref{app:search-pre} and~\ref{app:search-pois} give the per-candidate prescreening and the exact CDF computation; the comparisons confirm that none of these measures changes the output of the rule. Appendix~\ref{app:search-lut} explains how to precompute a table of $U$ when the parameters are fixed.

\subsection{\texorpdfstring{Computing $G_n(\widetilde d)$ with integer thresholds}{Computing the candidate CDF with integer thresholds}}\label{app:search-int}
Expanding \eqref{eq:aux} over the two cases $\varepsilon_1\varepsilon_2=\pm1$, we obtain
\begin{equation*}
G_n(\widetilde d)=\frac12\Pr\Bigl(Z_n\le\frac k2\bigl(\widetilde d-n-2\bigr)\Bigr)+\frac12\Pr\Bigl(Z_n\le\frac k2\bigl(\widetilde d-n+2\bigr)\Bigr).
\end{equation*}
Let $R=\sum_vC_v^2$, an integer, so that $\widetilde d=R/k$. A random variable with CDF $G_n$ is at most the observation if and only if $Z_n\le\frac k2(R/k-n\mp2)$, or equivalently $Z_n\le\lfloor(R-kn\mp2k)/2\rfloor$. Thus $G_n(\widetilde d)$ equals
\begin{equation*}
\frac12\Pr\Bigl(Z_n\le\Bigl\lfloor\frac{R-kn-2k}2\Bigr\rfloor\Bigr)+\frac12\Pr\Bigl(Z_n\le\Bigl\lfloor\frac{R-kn+2k}2\Bigr\rfloor\Bigr),
\end{equation*}
where we compute both floors using integer arithmetic, so floating point does not change the decision at lattice boundaries. The second term corresponds to $\varepsilon_1\varepsilon_2=-1$, the component with shift $-2$; its threshold is larger and so is its CDF.

\subsection{Search cutoff}\label{app:search-cut}
\begin{proposition}[Search cutoff]\label{prop:cut}
If $M>4\log(1/u)$, then every candidate satisfying
\begin{equation}
n\ge\frac{\widetilde d+2+\frac{4}{3k}\log(1/u)}{1-2\sqrt{\log(1/u)/M}}\label{eq:cut}
\end{equation}
is rejected.
\end{proposition}

For $u=0.01$, the condition is $M>18$. If the condition fails, the implementation returns $+\infty$, which only enlarges the upper bound; for any fixed $u$, this branch disappears as $M\to\infty$ and does not affect Theorem~\ref{thm:coverage}. At $M=16384$ and $u=0.01$, the coefficient of the observation $\widetilde d$ in \eqref{eq:cut} is about $1.035$, and only one or two hundred candidates lie between the cutoff and $U$.

\begin{proof}
$G_n(\widetilde d)$ is the average of two components, so it suffices to prove that the component with shift $-2$ is at most $u$. This component equals $\Pr(Z_n\le-x)$ with $x=\frac k2(n-R/k-2)$, and \eqref{eq:cut} guarantees $x>0$. When $\lambda_n=0$, this probability is zero and the conclusion follows directly; assume $\lambda_n>0$. Multiplying the moment generating functions of two independent $\Pois(\lambda_n)$ variables gives the moment generating function of the Skellam distribution,
\begin{equation}
\E e^{-tZ_n}=\exp\{2\lambda_n(\cosh t-1)\},\qquad t\in\mathbb R.\label{eq:skellam-mgf}
\end{equation}
For $0<t<3$, $2\lambda_n(\cosh t-1)=2\lambda_n\sum_{j\ge1}t^{2j}/(2j)!\le\lambda_nt^2\sum_{j\ge0}(t/3)^{j}=\lambda_nt^2/(1-t/3)$, because $2/(2j)!\le3^{-(2j-2)}$. Hence $-Z_n$ is a sub-gamma variable with variance factor $2\lambda_n$ and scale $1/3$. The Chernoff tail bound for sub-gamma variables \cite[Section~2.4]{boucheron2013concentration}, converted to Bernstein form as in the proof of Corollary~2.11 of that book, gives
\begin{equation*}
\Pr(Z_n\le-x)\le\exp\Bigl[-\frac{x^2}{4\lambda_n+2x/3}\Bigr].
\end{equation*}
A sufficient condition for the right-hand side to be at most $u$ is $x\ge2\sqrt{\lambda_n\log(1/u)}+\frac23\log(1/u)$. By~\eqref{eq:lambda}, $\lambda_n\le k^2n^2/(4M)$, so the sufficient condition becomes $\frac k2(n-R/k-2)\ge kn\sqrt{\log(1/u)/M}+\frac23\log(1/u)$, which rearranges to \eqref{eq:cut}; the condition $M>4\log(1/u)$ ensures that the denominator is positive.
\end{proof}

Numerically, we compute the cutoff at risk $u(1-10^{-9})$, round up, and add one. For the experimental values $u=0.01,0.05$, this risk is below the acceptance level $u-10^{-12}$, so the cutoff leaves a margin toward larger candidates and floating-point rounding cannot skip a candidate that should be accepted. The scan starts one below the cutoff and proceeds downward.

\subsection{Per-candidate Chernoff prescreening}\label{app:search-pre}
Every candidate between the cutoff and $U$ is rejected, and a per-candidate Chernoff bound can reject some of them before the exact CDF is computed. Let $j=\lfloor(R-kn+2k)/2\rfloor$; the component with shift $-2$ is $\Pr(Z_n\le j)$. The variable $Z_n$ has mean $0$, so for $j\ge0$ the event $\{Z_n\le j\}$ contains the mean and the Chernoff bound only gives the trivial value $1$; prescreening is therefore used only for $j<0$. When $\lambda_n=0$, $Z_n=0$ and $\Pr(Z_n\le j)=0$, so the candidate is rejected directly; assume $\lambda_n>0$. Applying the exponential Markov inequality to~\eqref{eq:skellam-mgf} gives, for $t>0$, $\Pr(Z_n\le j)\le\exp\{tj+2\lambda_n(\cosh t-1)\}$. Differentiating in $t$, the exponent is minimized at $\sinh t=r$ with $r=-j/(2\lambda_n)$; substituting $t=\operatorname{asinh}r$ and $\cosh t=\sqrt{1+r^2}$ gives
\begin{equation*}
\Pr(Z_n\le j)\le\exp\Bigl\{-\Bigl[(-j)\operatorname{asinh}r-2\lambda_n\bigl(\sqrt{1+r^2}-1\bigr)\Bigr]\Bigr\}.
\end{equation*}
If this bound is at most $u(1-10^{-9})$, the larger component is at most $u$ and so is the average, and the candidate is rejected directly. Prescreening only rejects candidates certified for rejection; disabling it and computing the exact CDF for every candidate gives the same result case by case.

\subsection{Windowed Poisson table for the exact CDF}\label{app:search-pois}
We use $\Pr(Z_n\le j)=\sum_m\Pr(P_2=m)\Pr(P_1\le j+m)$, where $P_1,P_2$ are independent $\Pois(\lambda_n)$ variables. By~\eqref{eq:lambda}, $\lambda_n\approx k^2n^2/(4M)$ grows quadratically with the candidate: at $M_1=16384$ and load $3$, $\lambda_n$ ranges from about $2\times10^5$ to $4\times10^5$. The initial implementation built the table from zero to $\lambda_n+14\sqrt{\lambda_n}$; a single sizing took about $1.3$ s and exhausted memory in early two-round runs. The implementation therefore builds, for each candidate, a Poisson table restricted to the window centered at $\lambda_n$ with half-width $14\sqrt{\lambda_n}+40$: starting at the mode, we compute the pmf recursively in both directions using $p(i+1)=p(i)\lambda_n/(i+1)$, normalize within the window, and accumulate the CDF. Substituting this half-width into the Chernoff bound for the two-sided Poisson tail gives an exponent of at least $98$ for every $\lambda_n$, so the mass beyond each end of the window is below $10^{-40}$, smaller than double-precision rounding error; normalization also avoids the relative error from evaluating $\log\Gamma$ term by term at large $\lambda_n$. The upper bounds before and after windowing agree case by case. Both the number of candidates requiring an exact CDF and the window length are proportional to $c$, so the cost of a single sizing is roughly proportional to $\lambda_d$; Appendix~\ref{app:timing} reports measurements.

The Poisson table depends only on $(M,k)$ and $n$, so it can be shared across syncs and risk levels. The implementation maintains a cache keyed by $(M,k)$. When full, it evicts the least recently used tables until half its capacity is free. Capacity is measured in entries, each containing two float64 values.

Under the condition of Proposition~\ref{prop:cut}, the cutoff, prescreening, and windowed tables preserve the output of the rule. We compare case by case against a reference implementation with 70 decimal digits of precision on 4536 general cases and 286 bounded-$c$ cases, covering lattice points, parity boundaries, and $+\infty$. We also compare the windowed implementation with cutoffs against the full-table implementation on 9779 cases. The value $M=16$ lies on both sides of the condition of Proposition~\ref{prop:cut}: the rule returns $+\infty$ at $u=0.01$ and a finite upper bound at $u=0.05$. On the 2406 cases with $M=16$, $k=3,4,5$, $u=0.01,0.05$, and $R=0,\dots,400$, we compare all three implementations. All three comparisons agree case by case.

\subsection{Table lookup with fixed parameters}\label{app:search-lut}
When $(M,k,u)$ is fixed, $U$ depends only on the integer $R$. We can precompute $U$ for every $R$ up to a chosen limit and store the results in a table indexed by $R$. For a fixed candidate $n$, both floors are nondecreasing in $R$, so $G_n(\widetilde d)$ is nondecreasing in $R$ as well. The Chernoff bound used for prescreening is also nondecreasing in $R$. Thus, the values of $R$ for which candidate $n$ is accepted form an integer interval unbounded above, whose left endpoint can be found by binary search. Recording each candidate at its left endpoint and then taking prefix maxima in increasing order of $R$ gives the largest accepted candidate, $U$, at every $R$ in the table. Candidates need only be checked up to the cutoff at the table's upper limit: the cutoff increases with $R$, and every candidate above it is rejected at each $R$ in the table (Proposition~\ref{prop:cut}). Table lookup and the sequential search use the same acceptance test and agree entry by entry. Observations beyond the table are still computed online as described above. Appendix~\ref{app:timing} reports table sizes and construction times.

\section{Experimental setup and supplementary results}\label{app:exp}
This appendix gives the setup and supplementary results of the experiments in Section~\ref{sec:exp}.

\subsection{Validation goals and experiment overview}\label{app:exp-overview}
The experiments fall into four groups according to the conclusion they test, and the following subsections follow this order. Experiment tags identify them in the main text and appendix.

\paragraph{Source of failure samples.} The statistical experiments (\exptag{STAT}) use the 189 configurations of Section~\ref{sec:exp-setup}, plus 42 boundary configurations at $M=64$, with about 20{,}000 failures per configuration. Configurations with rare failures sample from the planted distribution in place of the failure-conditioned distribution (Lemma~\ref{lem:plant}). The planted sampling control (\exptag{PLANT}) tests the finite-size error of this substitution: the low-load group directly checks the planted samples used by the statistical experiments in 12 configurations, and the critical-load group uses 21 configurations with $M=64$ to 1024 to locate the range of the approximation. Appendix~\ref{app:sampling-impl} and Appendix~\ref{app:planted} give the sampling method and the control results.

\paragraph{Properties of the random-sign encoding.} Random-sign encoding preserves decoding (Proposition~\ref{prop:decode}), and the decoding-equivalence test (\exptag{DEC}) verifies this on the full IBLT with more than 500 million paired trials, counting rejected candidates under both encodings (Appendix~\ref{app:decode-exp}). The post-failure estimate is conditionally unbiased for any sign composition (Propositions~\ref{prop:moments} and~\ref{prop:composition}), and the conditional-unbiasedness test (\exptag{UNB}) verifies this on 203 configurations (Appendix~\ref{app:dist}).

\paragraph{Lower-tail shape and conditional miss rate.} The failure records of the statistical experiments (\exptag{STAT}) compare the post-failure lower tail with the discrete and Gaussian limits (Theorems~\ref{thm:cross} and~\ref{thm:clt}, Appendix~\ref{app:limit}) and measure the conditional miss rate and extra capacity of the two rules at finite size (Theorem~\ref{thm:coverage} and Corollary~\ref{cor:gauss}, Appendix~\ref{app:cover-more}). The worst-configuration test (\exptag{WORST}) uses independent samples to test the Skellam rule's conditional miss rate on the 12 configurations closest to the risk bound, with 56{,}000 to 330{,}000 failures per configuration.

\paragraph{Two-round protocol and computation cost.} The protocol experiments (\exptag{PROTO}) run the two-round protocol on a real IBLT over 120 configurations with $M_1=256$ to 16384, first-round load 0.95 to 3, and one-sided and balanced sign compositions, measuring reliability and communication cost (Theorem~\ref{thm:protocol}, Appendix~\ref{app:protocol}). The computation-cost measurement (\exptag{COST}) gives the time and memory of sizing and decoding at $M=16384$ (Appendix~\ref{app:timing}).

\paragraph{Design fixed before sampling.} Before launch, we fix the configurations, seeds, sampling sources, stopping rules, and budgets for the statistical experiments (\exptag{STAT}), the worst-configuration test (\exptag{WORST}), the planted sampling control (\exptag{PLANT}), and the protocol experiments (\exptag{PROTO}) using a single feasibility pilot. We fix the selection rule for the worst-configuration test before the statistical experiments begin, then derive the specific configurations and sample sizes from that rule after the statistical experiments finish (Appendix~\ref{app:cover-more}). The conditional-unbiasedness test (\exptag{UNB}) and the decoding-equivalence test (\exptag{DEC}) come from separate independent runs, and the computation-cost measurement (\exptag{COST}) is conducted outside the other runs (Appendix~\ref{app:timing}).

\subsection{Sampling method}\label{app:sampling-impl}
There are two sources of failure samples. Direct sampling generates mappings independently and runs peeling, retaining only the failure samples; obtaining one failure requires $1/\Pr(F)$ mappings on average. Planted sampling picks a pair of difference keys uniformly at random, makes them share a uniformly random $k$-cell set, generates the mappings of the remaining keys independently and uniformly, and draws the signs independently; this construction produces a peeling failure every time, and the resulting mapping follows $Q$. Under iid random signs the statistic is invariant under label permutation, so the implementation equivalently fixes the last two labels as the duplicate edge pair. Lemma~\ref{lem:plant} gives the asymptotic relation between the planted distribution and the true failure-conditioned distribution, and Appendix~\ref{app:planted} compares them at finite size.

The statistical experiments (\exptag{STAT}) choose between the two sources according to the cost estimated in the pilot: configurations whose 20{,}000 direct failures are expected to cost at most 20 CPU core-hours use direct sampling with $\lceil2\times10^4/\widehat{\Pr}(F)\rceil$ trials, and the remaining configurations use 20{,}000 planted samples. Figure~\ref{fig:exp-grid} shows each configuration's location in $c$ and load and its sampling source. The six points at load 3 ($d=\mathrm{round}(3\alpha_cM)$, with rounding making the load slightly above 3) run separately with independent seeds after the statistical experiments finish, using the same sampling method and target failure count as the other configurations.

\begin{figure}[!htbp]
\centering\includegraphics[width=\linewidth]{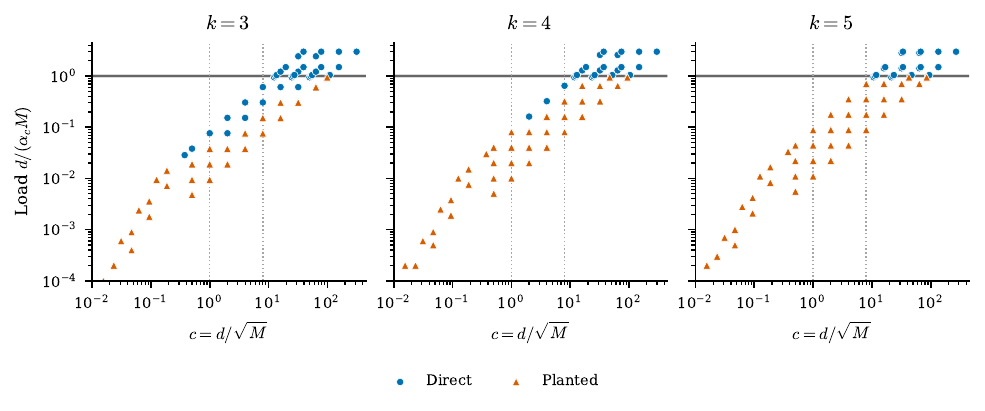}
\caption{The 189 configurations of the statistical experiments. Point shapes indicate the sampling source. Vertical lines mark the grouping boundaries of $c$, and horizontal shading marks load 0.95--1.05.}
\label{fig:exp-grid}
\end{figure}

All 203 configurations of the conditional-unbiasedness test (\exptag{UNB}) use direct sampling on CPUs with sequential peeling, with table sizes up to $M=262144$.

\subsection{Finite-size error of the planted distribution}\label{app:planted}
The planted sampling control (\exptag{PLANT}) tests whether the planted samples used in the main text can replace direct failure samples. The conclusion has two parts: in the low-load region where the statistical experiments actually use planted samples, the two samples show no detectable difference; near the critical load, the planted approximation still holds at $M=1024$ and degrades as $M$ decreases, with the planted sample giving higher miss at $M=64$.

\paragraph{Compared quantities and criteria.} Each configuration compares three quantities: the fraction of pure duplicate edge pairs among direct failures, the full-range KS distance between the post-failure $X$ distributions of the two samples, and the miss/u of the Skellam rule on each sample. The difference in miss/u is judged by an equivalence test with bound $\pm0.1$, with a one-sided level Bonferroni-corrected over the configurations in the same batch. The KS distance is compared with the asymptotic 95\% critical value: 0.0075 with 65{,}000 failures on each side and 0.011 with 20{,}000 on the direct side. Because $X$ is discrete, this comparison is conservative. The control has two groups, and Figure~\ref{fig:planted} summarizes both.

\paragraph{Low-load group.} Most of the 105 planted configurations in the statistical experiments lie far from the critical load, so the expected number of duplicate edge pairs $\mu$ (Appendix~\ref{app:graph}) is small, and Lemmas~\ref{lem:plant} and~\ref{lem:propL} make the failure probability of the same order as $\mu$. Direct sampling is infeasible for most of these configurations: among the 77 planted configurations with $k=4,5$, 76 have $\mu$ below $10^{-6}$, with a minimum of about $10^{-19}$, and the feasibility pilot observes no failure in 68 of 75 configurations; collecting 65{,}000 failures requires at least $6.5\times10^{10}$ trials, and about half the configurations require more than $10^{15}$, beyond the feasible range of direct sampling. On the other hand, the distance between the planted distribution and the failure-conditioned distribution is controlled by Lemma~\ref{lem:plant}, with an error term that tends to zero when $d=o(M)$. The region where direct sampling is hardest is precisely where this approximation is most accurate. The low-load group selects the 12 configurations for which direct sampling remains affordable, covering $c=0.5$ to 64 and load 0.04 to 0.71, with the number of direct trials fixed by $\mu$ before sampling. Each planted side has 65{,}000 failures; the direct side has about 65{,}000 failures for $k=3$ and for $(k,M,d)=(4,256,16)$ and $(5,256,128)$, and about 20{,}000 for the other 6 configurations, where failures are rarer.

Table~\ref{tab:v-low} gives the results, and the two samples agree on all 12 configurations. The fraction of pure duplicate edge pairs among direct failures ranges from 0.988 to 1.000, consistent with the mechanism in Lemma~\ref{lem:plant}. At $u=0.05$, all 95\% intervals of the difference contain zero; the equivalence test passes for all 4 configurations with $k=3$ and for 5 of the 8 configurations with $k=4,5$, while the other 3 ($(4,1024,512)$, $(5,256,32)$, $(5,256,64)$) have only about 20{,}000 direct failures, with interval endpoints within $\pm0.13$. At $u=0.01$, only $(4,1024,512)$ has a difference interval excluding zero, equal to $-0.16$ (interval $[-0.30,-0.02]$), meaning higher miss in the planted sample. Only $(3,4096,1024)$ has a KS distance of 0.0079 slightly above the critical value ($p=0.034$), and no difference is detectable after correction for 12 tests.

\begin{table}[!htbp]
\centering
\caption{Direct and planted sampling in the low-load group (Skellam rule, $u=0.05$). miss/u lists direct / planted. Difference gives the 95\% interval for direct minus planted. KS is the full-range KS distance of $X$, with 95\% critical value 0.0075 for 65{,}000 direct failures and 0.011 for 20{,}000.}
\label{tab:v-low}
\small
\setlength{\tabcolsep}{3pt}
\begin{tabular}{@{}rrrrrrcccc@{}}
\toprule
& & & & & direct & pure duplicate edge pairs & & & \\
$k$ & $M$ & $d$ & $c$ & load & failure count & fraction & miss/u & difference & KS\\
\midrule
3 & 1024 & 32 & 1 & 0.04 & 65061 & 1.000 & 0.94 / 0.95 & $[-0.05,+0.04]$ & 0.0059\\
3 & 4096 & 512 & 8 & 0.15 & 64977 & 1.000 & 0.95 / 0.96 & $[-0.06,+0.03]$ & 0.0048\\
3 & 4096 & 1024 & 16 & 0.31 & 65136 & 0.999 & 0.92 / 0.94 & $[-0.07,+0.03]$ & 0.0079\\
3 & 16384 & 8192 & 64 & 0.61 & 65530 & 0.999 & 0.98 / 0.99 & $[-0.06,+0.04]$ & 0.0051\\
\addlinespace
4 & 256 & 8 & 0.5 & 0.04 & 20071 & 0.998 & 0.54 / 0.52 & $[-0.03,+0.07]$ & 0.0071\\
4 & 256 & 16 & 1 & 0.08 & 65839 & 0.995 & 0.68 / 0.69 & $[-0.05,+0.03]$ & 0.0057\\
4 & 1024 & 128 & 4 & 0.16 & 20067 & 0.997 & 0.95 / 0.95 & $[-0.06,+0.07]$ & 0.0049\\
4 & 1024 & 256 & 8 & 0.32 & 20265 & 0.994 & 0.98 / 0.97 & $[-0.06,+0.08]$ & 0.0048\\
4 & 1024 & 512 & 16 & 0.65 & 20038 & 0.988 & 0.94 / 0.97 & $[-0.10,+0.04]$ & 0.0062\\
5 & 256 & 32 & 2 & 0.18 & 20137 & 0.999 & 1.00 / 0.98 & $[-0.05,+0.09]$ & 0.0068\\
5 & 256 & 64 & 4 & 0.36 & 19885 & 0.997 & 0.91 / 0.95 & $[-0.11,+0.03]$ & 0.0048\\
5 & 256 & 128 & 8 & 0.71 & 65470 & 0.995 & 0.99 / 1.01 & $[-0.07,+0.02]$ & 0.0063\\
\bottomrule
\end{tabular}
\end{table}

\paragraph{Critical-load group.} This group locates the range of the planted approximation. The statistical experiments include five planted configurations at load 0.95, all with $M\ge4096$. This group uses $k=3,4,5$, $M=64$ to 1024, load 0.85 to 0.95, and direct failure probability at least $10^{-3}$, comparing the two samples in a region where larger stopping sets cause most failures. The sample size is set by the equivalence test on the miss/u difference at $u=0.05$: bound $\pm0.1$, one-sided level $0.05/21$, 80\% power when the difference is zero, giving 65{,}000 failures on each side. Table~\ref{tab:v} reports the 21 configurations.

\begin{table}[!htbp]
\centering
\caption{Direct and planted sampling in the critical-load group (Skellam rule). miss/u lists direct / planted. KS is the full-range KS distance of $X$, with 95\% critical value 0.0075.}
\label{tab:v}
\small
\setlength{\tabcolsep}{3pt}
\begin{tabular}{@{}rrrrcccc r@{}}
\toprule
$k$ & $M$ & load & pure duplicate edge pairs & \multicolumn{2}{c}{$u=0.01$} & \multicolumn{2}{c}{$u=0.05$} & KS\\
\cmidrule(lr){5-6}\cmidrule(lr){7-8}
& & & fraction & miss/u & difference & miss/u & difference & \\
\midrule
3 & 64 & 0.86 & 0.09 & 0.44 / 0.48 & $-0.04\pm0.07$ & 0.87 / 0.88 & $-0.01\pm0.04$ & 0.0026\\
3 & 64 & 0.90 & 0.05 & 0.49 / 0.52 & $-0.03\pm0.08$ & 0.82 / 0.87 & $-0.05\pm0.04$ & 0.0085\\
3 & 64 & 0.95 & 0.02 & 0.44 / 0.51 & $-0.07\pm0.07$ & 0.87 / 0.94 & $-0.07\pm0.05$ & 0.0088\\
4 & 64 & 0.85 & 0.01 & 0.56 / 0.49 & $+0.08\pm0.08$ & 0.89 / 0.92 & $-0.03\pm0.05$ & 0.0105\\
4 & 64 & 0.89 & 0.01 & 0.54 / 0.65 & $-0.11\pm0.08$ & 0.89 / 0.96 & $-0.07\pm0.05$ & 0.0097\\
4 & 64 & 0.95 & 0.00 & 0.39 / 0.55 & $-0.16\pm0.07$ & 0.81 / 0.93 & $-0.13\pm0.04$ & 0.0131\\
5 & 64 & 0.85 & 0.00 & 0.58 / 0.69 & $-0.11\pm0.09$ & 0.96 / 1.07 & $-0.11\pm0.05$ & 0.0130\\
5 & 64 & 0.89 & 0.00 & 0.57 / 0.69 & $-0.12\pm0.09$ & 1.00 / 1.13 & $-0.13\pm0.05$ & 0.0142\\
5 & 64 & 0.96 & 0.00 & 0.55 / 0.72 & $-0.16\pm0.09$ & 0.89 / 1.04 & $-0.15\pm0.05$ & 0.0118\\
3 & 256 & 0.85 & 0.63 & 0.75 / 0.62 & $+0.13\pm0.09$ & 0.89 / 0.88 & $+0.01\pm0.04$ & 0.0063\\
3 & 256 & 0.90 & 0.13 & 0.75 / 0.66 & $+0.09\pm0.09$ & 0.95 / 0.88 & $+0.07\pm0.05$ & 0.0080\\
3 & 256 & 0.95 & 0.02 & 0.77 / 0.64 & $+0.14\pm0.09$ & 0.92 / 0.91 & $+0.01\pm0.05$ & 0.0060\\
4 & 256 & 0.90 & 0.01 & 0.71 / 0.70 & $+0.01\pm0.09$ & 0.93 / 0.98 & $-0.05\pm0.05$ & 0.0045\\
4 & 256 & 0.95 & 0.00 & 0.72 / 0.73 & $-0.01\pm0.09$ & 0.93 / 0.97 & $-0.05\pm0.05$ & 0.0052\\
5 & 256 & 0.90 & 0.00 & 0.71 / 0.75 & $-0.04\pm0.09$ & 0.92 / 0.94 & $-0.02\pm0.05$ & 0.0052\\
5 & 256 & 0.95 & 0.00 & 0.73 / 0.80 & $-0.07\pm0.09$ & 0.95 / 0.94 & $+0.00\pm0.05$ & 0.0027\\
3 & 1024 & 0.85 & 0.97 & 0.79 / 0.78 & $+0.01\pm0.10$ & 0.92 / 0.91 & $+0.02\pm0.05$ & 0.0033\\
3 & 1024 & 0.90 & 0.96 & 0.83 / 0.86 & $-0.03\pm0.10$ & 0.93 / 0.95 & $-0.02\pm0.05$ & 0.0037\\
3 & 1024 & 0.95 & 0.08 & 0.84 / 0.84 & $+0.01\pm0.10$ & 0.96 / 0.96 & $-0.00\pm0.05$ & 0.0091\\
4 & 1024 & 0.95 & 0.00 & 0.89 / 0.91 & $-0.02\pm0.10$ & 0.97 / 0.94 & $+0.03\pm0.05$ & 0.0054\\
5 & 1024 & 0.95 & 0.00 & 0.85 / 0.79 & $+0.06\pm0.10$ & 0.96 / 0.95 & $+0.01\pm0.05$ & 0.0039\\
\bottomrule
\end{tabular}
\end{table}

The difference between the two samples decreases as $M$ grows. All five $M=1024$ configurations pass the equivalence test at $u=0.05$, and the difference intervals at both levels of $u$ contain zero; among them, the fraction of pure duplicate edge pairs is 0.97 in some configurations and close to zero in others. For 3 of the 7 configurations with $M=256$, the test does not establish that the difference lies within the bound: at $k=3$ and load 0.90 the direct sample has higher miss, and at $k=4$ and loads 0.90 and 0.95 the planted sample has higher miss. The two configurations with $k=3$ and loads 0.85 and 0.95 agree at $u=0.05$, while at $u=0.01$ the direct sample's miss/u is higher by 0.13 and 0.14, with intervals excluding zero. For 7 of the 9 configurations with $M=64$, the test does not establish the bound; the deviation has the same direction in all of them, with the planted sample giving higher miss and the difference interval reaching as far as $-0.20$.

The KS distance shows the same trend: 1 configuration above the critical value at $M=1024$, 1 at $M=256$, and 8 at $M=64$. The $M=1024$ configuration is $(k,M,d)=(3,1024,796)$, where the difference lies in the middle of the distribution; the maximum difference in the lower tail between the 0.2\% and 15\% quantiles of the pooled sample is only 0.0039. At $M=1024$, the configuration with a pure duplicate edge pair fraction of 0.97 has as small a KS distance as the configurations with fractions close to zero: the distributions of $X$ can be close even when the failure mechanisms differ, so the fraction of pure duplicate edge pairs alone does not predict the distribution error of $X$, which also depends on overlap fluctuations on the $s(d)$ scale.

\begin{figure}[!htbp]
\centering\includegraphics[width=\linewidth]{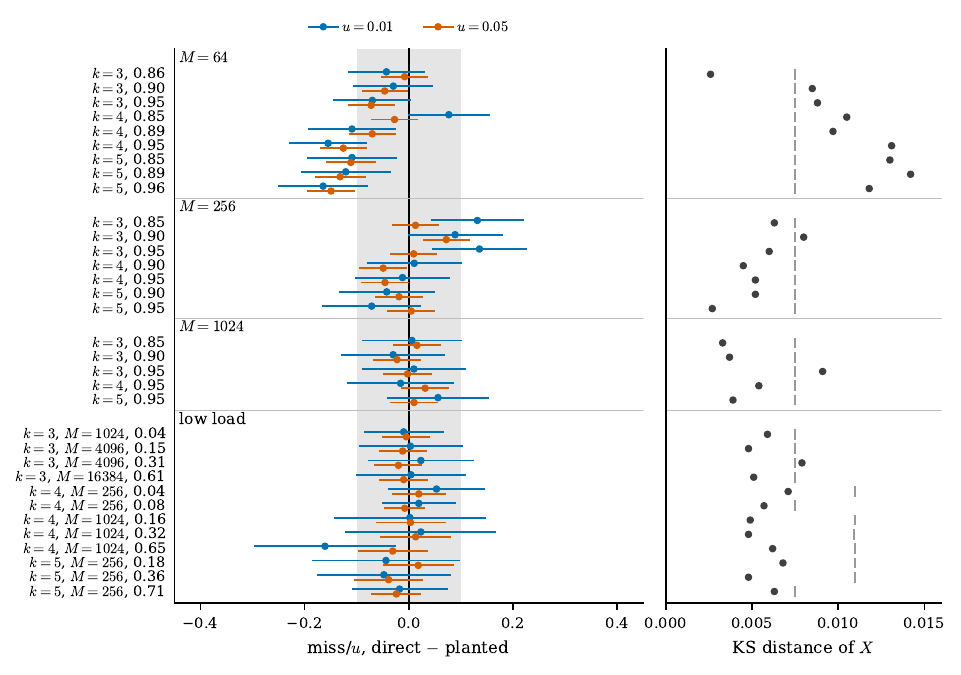}
\caption{Direct versus planted sampling. The upper three panels show the critical-load group, grouped by $M$; the bottom panel shows the low-load group. The vertical axis labels $k$ and load. Left: Skellam rule miss/u difference (direct minus planted) with 95\% intervals and shading at $\pm0.1$. Right: full-range KS distance of $X$; gray ticks mark the 95\% critical value for each sample size.}
\label{fig:planted}
\end{figure}

\subsection{Random-sign encoding does not change decoding}\label{app:decode-exp}
The decoding-equivalence test (\exptag{DEC}) compares first-round decoding under standard encoding and random-sign encoding. Across all paired trials, the two encodings agree on success and failure, the recovered results agree with the true residual core, and the decoder never falsely accepts a non-pure cell. The test uses the two-round protocol's actual IBLT encoder and decoder, builds the two parties' tables, and subtracts them.

The test has three groups, each answering one question. The first group covers the full load range and all sign compositions, confirming equivalence over the whole parameter range with a large sample. The second group samples loads more densely near the peeling threshold, where the failure probability changes fastest with load and agreement is most sensitive; Table~\ref{tab:decode} of the main text also comes from this group. The third group examines the intermediate steps of decoding, counting rejected candidates and comparing how the checksum workload of the two encodings varies with sign composition. All three groups use the same construction and verification.

\paragraph{Construction and verification.} Each trial regenerates distinct keys and fresh seeds for the cell hash, sign hash, and checksum hash; both encodings share the same keys and hashes. We first verify cell by cell that the three fields after cancellation of the common elements match those obtained by encoding the difference directly, then run decoding. An independent reference uses only the unsigned incidence graph of the true difference and computes its residual 2-core by the degree-one rule. Whenever the actual decoder accepts a key, we check the true residual degree, key, and direction, then update the reference state; these checks cover trajectories ending in both success and failure. After termination, we verify that the recovered keys are exactly the difference outside the 2-core, with all directions correct, reconstruct the three fields from the true residual, compare them cell by cell, and check that all fields are zero if and only if the residual core is empty.

\paragraph{First group: full load grid.} The configurations use $k\in\{3,4,5,6,10\}$ and $M\in\{64,1024,16384\}$. The load is $d/(\alpha_cM)\in\{0.5,0.95,1,1.05,3\}$, and the positive-direction ratio $d_+/d$ with target values $0,1/4,1/2,3/4,1$; we round $d$ and $d_+$ to integers from the target values. We use the corresponding peeling threshold for each $k$, giving 375 configurations in total. The number of common elements is $\min\{1000,M\}$. For $M=64,1024,16384$, we run 3,000,000, 900,000, and 100,000 trials per configuration, respectively, for a total of 500,000,000 trials, with 100,000,000 for each $k$. Table~\ref{tab:decode-paired} reports the results by $k$. Across all configurations, we observe zero paired success/failure disagreements, non-pure acceptances, incorrect recoveries, or residual-core disagreements. The two decoders together check 211,516,840,988 acceptance operations, all satisfying the true degree-one and direction conditions.

\paragraph{Second group: near the threshold.} The configurations use $k\in\{3,4,5,6,10\}$, $M\in\{1024,16384,65536\}$, 10 load values between 0.90 and 1.05, and $d_-/d\in\{0,1/2\}$. For $M=1024,16384,65536$, we run 100,000, 20,000, and 10,000 trials per configuration, respectively, for a total of 13,000,000, and all checks pass. At $M=16384$, the failure probability ranges from 0.020 to 0.060 at load 0.99 and from 0.993 to 0.998 at load 1.01; at $M=65536$, almost no trials fail at load 0.99 and almost all fail at load 1.01. After failure, the decoder still recovers the difference outside the 2-core. At $M=65536$ and load 1.02, the average recovery fraction ranges from 0.18 ($k=10$) to 0.50 ($k=3$).

\paragraph{Third group: rejected candidates.} Rejected candidates are cells with count $\pm1$ that fail the checksum or the cell-ownership check. Paired trials count them separately for each encoding. The configurations use $M=16384$, $k\in\{3,4,5,6,10\}$, load $0.5,0.95,1$, and nine values of $d_-/d$ from 0 to 1 in steps of $1/8$, with 50,000 trials per configuration for a total of 6,750,000, and all checks pass. Figure~\ref{fig:decode} reports the results for $k=3$ to 6. Standard encoding has no rejected candidates under one-sided differences because the count equals the degree; as the difference approaches balanced, more opposite-sign cancellation occurs and more candidates are rejected. Under random-sign encoding, the number of rejected candidates differs by at most 0.3\% across the nine sign compositions and equals the standard value in the balanced case, consistent with Proposition~\ref{prop:composition}: the coding coefficient is an iid random sign independent of $(d_+,d_-)$. The results for $k=10$ follow the same pattern, with both encodings yielding about 1.27 rejected candidates per difference key at load 1.

Across the three groups, the two decoders together reject $2.8\times10^{11}$ non-pure candidates with count $\pm1$, with zero false acceptances; the cell-ownership check rejects zero candidates, and the checksum rejects every incorrect candidate.

\begin{figure}[!htbp]
\centering\includegraphics[width=\linewidth]{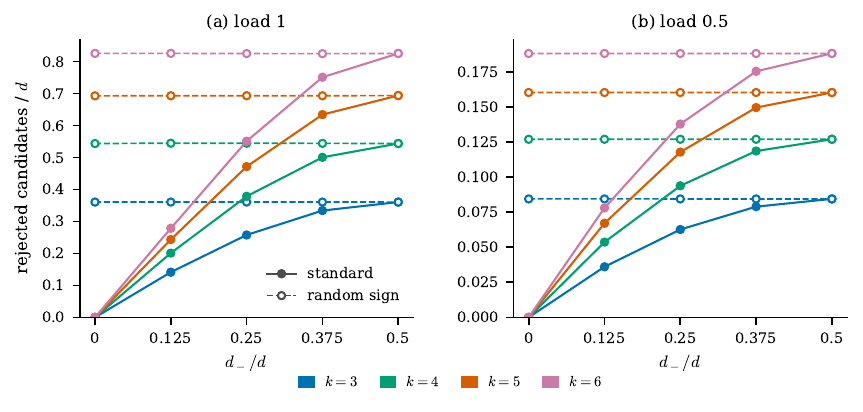}
\caption{Rejected candidates per first-round decoding, divided by $d$ ($M=16384$; both encodings share the same mappings for each configuration). Solid lines with filled markers show standard encoding; dashed lines with open markers show random-sign encoding.}
\label{fig:decode}
\end{figure}

\begin{table}[!htbp]
\centering
\caption{First-round paired decoding results.}
\label{tab:decode-paired}
\begin{tabular}{rrrr}
\toprule
$k$ & paired trials & successes & failures\\
\midrule
3 & 100,000,000 & 36,532,515 & 63,467,485\\
4 & 100,000,000 & 35,900,798 & 64,099,202\\
5 & 100,000,000 & 33,530,250 & 66,469,750\\
6 & 100,000,000 & 32,553,378 & 67,446,622\\
10 & 100,000,000 & 30,233,256 & 69,766,744\\
\bottomrule
\end{tabular}
\end{table}

\subsection{The post-failure estimate is conditionally unbiased for every sign composition}\label{app:dist}
Under random-sign encoding, the conditional mean of the post-failure $X$ agrees with zero on all tested configurations, whereas the conditional mean under standard encoding varies with sign composition and $k$. We give the results for both encodings and then describe the data of Figure~\ref{fig:encoding} in the main text.

\paragraph{Random-sign encoding.} The conditional-unbiasedness test (\exptag{UNB}) uses $k=3$ to 6 and $M=64$ to 262144, with loads covering the critical window, its left side, and 0.4 to 0.95. Of the 203 configurations, 196 have at least 1000 failure samples. For these, $|\E[X\mid F]|$ is at most 0.032, with 93\% within two standard errors, consistent with the exact unbiasedness in Proposition~\ref{prop:moments}; the conditional variance to the left of the critical window is 0.972--1.049. The 84 direct-sampling configurations of the statistical experiments (\exptag{STAT}) give the same result: $|\E[X\mid F]|$ is at most 0.019, with 98\% within two standard errors.

\paragraph{Composition-by-composition comparison of the two encodings.} Table~\ref{tab:mean-by-k} lists the values of Figure~\ref{fig:composition} in the main text. For each $(k,M,d)$, the same batch of mappings records $X$ under both encodings at five compositions, and both encodings use all failure records; the four configurations have failure probabilities of 0.099--0.102. All 20 values under random-sign encoding have absolute value at most 0.005 and lie within two standard errors of zero. Under standard encoding, the bias is largest for one-sided differences and decreases in absolute value as $d_-/d$ grows, approaching zero at balanced; for one-sided differences its sign changes with $k$, being positive for $k=3,4$ and negative for $k=5,6$. For one-sided differences $b^2/d=d$, so by~\eqref{eq:projection} of Proposition~\ref{prop:projection} the standard-encoding bias has the same sign as $c_F$, and the change of sign with $k$ therefore comes from $c_F$. The statistical experiments give the same sign change at load 1.0 and failure probability about 0.67: the conditional mean under one-sided differences is $+0.20$, $+0.02$, and $-0.08$ for $k=3,4,5$ (standard error 0.007).

\paragraph{Data of Figure~\ref{fig:encoding}.} The data of this figure are generated by the sampling code of the interactive experiment page Peeling Lens: on each random mapping it computes $X$ for three encodings, with random signs from an independent generator. The configuration has $c\approx104$. Under standard encoding, $\E[X\mid F]=+0.599\pm0.001$ for one-sided differences and $-0.002$ for balanced differences; under random-sign encoding, it is $+0.001\pm0.001$, with conditional variance 1.007, and the failure and full samples have similar distributions, consistent with the large-$c$ behavior in Theorem~\ref{thm:clt}.

\begin{table}[!htbp]
\centering
\caption{Values in Figure~\ref{fig:composition}: mean post-failure $X$ ($M=16384$, failure probability about 0.1). Standard errors are 0.0018 for $k=3$ and 0.0031--0.0032 for $k=4,5,6$.}
\label{tab:mean-by-k}
\begin{tabular}{rrlrrrrr}
\toprule
& & & \multicolumn{5}{c}{$d_-/d$}\\
\cmidrule(lr){4-8}
$k$ & $d$ & encoding & 0 (one-sided) & 0.1 & 0.25 & 0.4 & 0.5 (balanced)\\
\midrule
3 & 13306 & standard & $+0.638$ & $+0.405$ & $+0.158$ & $+0.028$ & $+0.003$\\
  &       & random sign & $+0.002$ & $-0.001$ & $0.000$ & $-0.004$ & $+0.002$\\
\addlinespace
4 & 12567 & standard & $+0.084$ & $+0.052$ & $+0.016$ & $-0.002$ & $-0.003$\\
  &       & random sign & $-0.001$ & $-0.002$ & $+0.005$ & $+0.001$ & $-0.001$\\
\addlinespace
5 & 11418 & standard & $-0.218$ & $-0.137$ & $-0.054$ & $-0.009$ & $+0.002$\\
  &       & random sign & $+0.003$ & $-0.005$ & $0.000$ & $-0.002$ & $-0.004$\\
\addlinespace
6 & 10363 & standard & $-0.387$ & $-0.244$ & $-0.096$ & $-0.021$ & $-0.004$\\
  &       & random sign & $+0.002$ & $+0.003$ & $+0.005$ & $+0.001$ & $+0.001$\\
\bottomrule
\end{tabular}
\end{table}

\subsection{Distance between the post-failure lower tail and the limit distribution}\label{app:limit}
This section uses the failure records of the statistical experiments (\exptag{STAT}) to compare the post-failure distribution with the two limit distributions. For small $c$, the discrete limit $G_d$ is close to the empirical distribution while the Gaussian shows a clear deviation; for $c\ge8$ the Gaussian is already accurate enough. The error of both approximations decreases as $M$ grows.

\paragraph{Statistics used in the comparison.} The comparison with $G_d$ uses the raw integer sum of squares $\widetilde d=\sum_vC_v^2/k$. By~\eqref{eq:raw-decomp} of Appendix~\ref{app:ams}, the difference between $\widehat d$ and $\widetilde d$ is about $0.02s(d)$ at $M=16384$. This perturbation is small relative to the standard deviation but shifts the discrete atoms, so we validate the distribution using $\widetilde d$, as in Theorem~\ref{thm:cross}. Figure~\ref{fig:crosscdf} retains all empirical values of $\widetilde d$ and standardizes them on the same scale as $G_d$; the KS distance is computed on both sides of every value taken by either distribution.

\paragraph{KS distance.} Table~\ref{tab:limit-ks} computes the full KS distances from the raw integer records, with the two columns corresponding to the two limits. The discrete-limit column compares $\widetilde d$ with $G_d$ for $c=1,2,4$ and $k=3,4,5$, giving nine configurations at each $M$; $M\le1024$ includes direct and planted samples, and $M\ge4096$ uses planted samples. The same column also lists, for $c=1$, the KS distance between $(\widetilde d-d)/s(d)$ and the Gaussian as a control. The Gaussian column compares $X=(\widehat d-d)/s(d)$ with $N(0,1)$ over the 75 direct-sampling configurations with $c\ge8$, including the critical window and overload. Each configuration has about 20{,}000 failures, and at this fixed sample size the per-configuration 95\% DKW error bound is about 0.010. The KS distance of the discrete limit decreases from 0.019--0.034 at $M=256$ to 0.003--0.012 at $M=16384$, while the Gaussian KS distance at $c=1$ stays above 0.21 throughout. For $c\ge8$, the Gaussian KS distance falls to the order of the DKW error bound at $M=16384$, and the 1\% quantile error decreases from $0.18s(d)$ to $0.05s(d)$.

\begin{table}[!htbp]
\centering
\caption{Full post-failure KS distance. The discrete limit uses $\widetilde d$; the Gaussian uses $\widehat d$. Quantile error is in units of $s(d)$ and gives the maximum absolute error across configurations.}
\label{tab:limit-ks}
\small
\begin{tabular}{rcccccc}
\toprule
& \multicolumn{3}{c}{discrete limit ($c=1,2,4$)} & \multicolumn{3}{c}{Gaussian ($c\ge8$, direct sampling)}\\
\cmidrule(lr){2-4}\cmidrule(lr){5-7}
$M$ & KS & KS of Gaussian ($c=1$) & configs & KS & 1\% quantile error & configs\\
\midrule
256 & 0.019--0.034 & 0.237--0.252 & 9 & 0.010--0.021 & 0.18 & 23\\
1024 & 0.008--0.023 & 0.224--0.244 & 9 & 0.005--0.015 & 0.10 & 23\\
4096 & 0.004--0.016 & 0.215--0.230 & 9 & 0.005--0.014 & 0.06 & 17\\
16384 & 0.003--0.012 & 0.214--0.230 & 9 & 0.003--0.010 & 0.05 & 12\\
\bottomrule
\end{tabular}
\end{table}

\paragraph{Conditional variance.} The conditional variance reflects the same structure. Among the 84 direct-sampling configurations of the statistical experiments, the conditional variance under random-sign encoding is 0.98--1.06 for $c\ge8$ and 0.985--1.020 within the critical window. For small $c$, it is well above one: about 1.12 at $c=4$, about 1.5 at $c=2$, and above 10 when $c<1$, consistent with $1+2/c^2$, obtained by converting the variance $4+2c^2$ of the limit distribution in Theorem~\ref{thm:cross} to the $X$ scale. The excess comes from the $\pm2$ shift of the duplicate edge pair, which is non-negligible relative to $s(d)$.

\subsection{Conditional miss rate and extra capacity at finite size}\label{app:cover-more}
This section supplements the numerical results of Section~\ref{sec:exp-cover}. The main text concludes that the Skellam rule keeps the conditional miss rate near $u$ on all configurations, the Gaussian rule does so only for large $c$, and the extra capacity $(U-d)/d$ of the two rules is almost the same and decreases as $M^{-1/2}$. We first give the complete results of the statistical experiments (\exptag{STAT}), the miss/u and extra capacity of both rules, and then recheck the configurations closest to the risk bound with independent samples (\exptag{WORST}).

\paragraph{Conditional miss rate of the Skellam rule.} Of the 189 configurations of the statistical experiments (\exptag{STAT}), 84 use direct sampling and 105 use planted sampling. Before judging whether miss/u exceeds 1, we consider sampling error: each configuration has about 20{,}000 failures, corresponding to about 200 misses at $u=0.01$, so the standard error of miss/u is about 0.07. At this precision, the Skellam rule's maximum miss/u on direct-sampling configurations is 1.04 (the same at both levels of $u$, occurring within the critical window), indistinguishable from 1; the maximum on planted configurations is 1.03, computed under the planted distribution. On direct-sampling configurations with $c\ge8$, the medians at the two levels of $u$ are 0.79--0.86 and 0.95--0.97. Low-load configurations with $c<8$ are more conservative, with a median of at most 0.76 at $u=0.01$, because of the discrete structure: the lower tail is concentrated on a few lattice points and the rule cannot use the whole risk budget exactly. Under the Skellam rule every sample has at least one accepted candidate; if no candidate is accepted, $U=+\infty$, which does not count as a miss.

\paragraph{Excess of the Gaussian rule.} For $c\ge8$, the Gaussian rule's miss/u is 0.83--1.19 across all configurations, at the same level as the Skellam rule. As $c$ decreases, the excess grows rapidly. Table~\ref{tab:gauss-c} lists the values at fixed $c$ for $M=16384$, corresponding to the triangles with $c\le16$ in Figure~\ref{fig:coverage} of the main text: at $u=0.01$ it is about 8--11 at $c=1$ and above 33 at $c=0.5$. Comparing table sizes, the values at $M=4096$ and 16384 differ by at most 10\% for $c\le4$, so the excess does not vanish as the table grows, consistent with the finite limit in Theorem~\ref{thm:cross}. For very small $c$ there is another case: when the observation lies near $d-2$, observations for $d=2,3$ fall below the acceptance threshold of every candidate, the Gaussian rule returns $+\infty$, and no miss occurs.

\paragraph{Extra capacity.} The price of controlling miss is an upper bound above $d$, and Table~\ref{tab:extra} summarizes this extra capacity. For $c\ge8$ the extra capacity is determined mainly by the table size and roughly halves when $M$ grows by a factor of 4, consistent with $s(d)/d\approx\sqrt{2/M}$; the medians of the two rules differ by at most 2 percentage points in the same configuration. When $c$ is bounded, $d$ itself is small and the ratio is less meaningful, so the right half of the table lists the excess number of elements $U-d$. Here the Gaussian rule has a lower upper bound, corresponding to its excess miss at small $c$. For fixed small $d$, the 90\% quantile of the upper bound at $d=2$ exceeds $d$ by 4--5.

\begin{table}[!htbp]
\centering
\caption{Extra capacity of the upper bound in the statistical experiments (\exptag{STAT}). Left: $(U-d)/d$ for $c\ge8$ and $u=0.01$, with the median taken over configuration medians and the 90\% quantile taken as the maximum over configuration 90\% quantiles; right: median of $U-d$ when $c$ is bounded.}
\label{tab:extra}
\begin{tabular}{rccc}
\toprule
& \multicolumn{2}{c}{Skellam rule} & Gaussian rule\\
\cmidrule(lr){2-3}\cmidrule(lr){4-4}
$M$ & median & 90\% quantile & median\\
\midrule
256 & 25\% & 41\% & 23.6\%\\
1024 & 11.3\% & 18\% & 10.9\%\\
4096 & 5.4\% & 8.5\% & 5.3\%\\
16384 & 2.6\% & 4.1\% & 2.6\%\\
\bottomrule
\end{tabular}
\qquad
\begin{tabular}{rcc}
\toprule
$c$ & Skellam rule & Gaussian rule\\
\midrule
1 & 5 & 3\\
2 & 8 & 6\\
4 & 15 & 14\\
\bottomrule
\end{tabular}
\end{table}

\paragraph{Case $M=64$.} When the table size drops to $M=64$, the overlap between the duplicate edge pair and the remaining edges, omitted by the null distribution, is no longer negligible (Section~\ref{sec:exp-cover}). Among the 42 configurations with $M=64$, the Skellam rule's maximum miss/u is 1.08 at $u=0.05$ ($k=4$, $d=6$, 95\% Wilson interval 1.02--1.15), beyond sampling error, and 1.04 at $u=0.01$.

\begin{table}[!htbp]
\centering
\caption{Gaussian rule miss/u at fixed $c$ ($M=16384$).}
\label{tab:gauss-c}
\begin{tabular}{rrrrrrr}
\toprule
& \multicolumn{3}{c}{$u=0.01$} & \multicolumn{3}{c}{$u=0.05$}\\
\cmidrule(lr){2-4}\cmidrule(lr){5-7}
$c$ & $k=3$ & $k=4$ & $k=5$ & $k=3$ & $k=4$ & $k=5$\\
\midrule
0.5 & 36.4 & 33.2 & 33.2 & 9.37 & 8.81 & 9.39\\
1 & 11.3 & 8.74 & 8.27 & 4.09 & 4.24 & 4.60\\
2 & 3.06 & 2.48 & 2.60 & 2.00 & 1.70 & 1.94\\
4 & 1.52 & 1.41 & 1.50 & 1.21 & 1.19 & 1.20\\
8 & 1.12 & 1.04 & 1.08 & 1.07 & 1.08 & 1.06\\
16 & 0.98 & 0.93 & 1.00 & 0.98 & 1.01 & 1.00\\
\bottomrule
\end{tabular}
\end{table}

\paragraph{Worst-configuration test (\exptag{WORST}).} The maximum miss/u of the statistical experiments lies within sampling error, but each configuration has only about 20{,}000 failures, which is insufficient to confirm that the highest-risk configurations also meet the target. This test uses samples independent of the statistical experiments (\exptag{STAT}) to recheck the Skellam rule's conditional miss rate on the configurations closest to the risk bound, with both the selection rule and the sample sizes fixed before sampling. We first select the 5 configurations whose miss/u exceeds 1 in exploration; the remaining slots are filled in descending order of the larger of the two one-sided 95\% exact upper bounds divided by $u$ in the statistical experiments (\exptag{STAT}). Four highly ranked configurations have only planted samples, because direct sampling cannot obtain failures at feasible cost, and the rule skips them in favor of later configurations; their miss/u under $Q$ is at most 1.02. Only one of the selected 12 configurations uses $M=256$.

The criteria are also fixed before sampling. The 12 configurations and two levels of $u$ give 24 judgments. A judgment passes if its one-sided simultaneous confidence upper bound is at most $1.1u$, using the Clopper--Pearson interval with a Bonferroni correction over the 24 judgments. Sample sizes are chosen by power: for each configuration and level of $u$, we use the previously estimated 95\% Wilson upper bound (at most 1) as the planning value, choose the failure count that gives a passing probability of at least $1-0.2/24$, and take the larger count over the two levels, so that at the planning value the probability that all 24 judgments pass is at least 80\%. We sample in blocks and analyze all blocks through the first point at which the cumulative failure count in block order reaches the target; stopping depends only on the failure count.

All 24 judgments pass, with a maximum upper bound of $1.024u$; 22 of the 24 two-sided 95\% intervals lie entirely below 1 and the other 2 cover 1. Table~\ref{tab:confirm} lists the load, failure count, miss/u, and one-sided simultaneous confidence upper bound of each test configuration, and the segments in Figure~\ref{fig:confirm} of the main text are two-sided 95\% intervals. As a control, at the two test configurations with $k=3$ and $M=1024$, the Gaussian rule's miss/u is 11.1 ($u=0.01$) and 4.1 ($u=0.05$) at $c=1$, and 2.9 and 1.8 at $c=2$.

\begin{table}[!htbp]
\centering
\caption{Skellam rule miss/u on the \exptag{WORST} configurations, with one-sided simultaneous confidence upper bounds in parentheses (per-item level $0.05/24$).}
\label{tab:confirm}
\begin{tabular}{lrrcc}
\toprule
configuration & load & failures & $u=0.01$ & $u=0.05$\\
\midrule
$k=3$, $M=1024$, $d=32$ & 0.04 & 55{,}509 & 0.53 (0.62) & 0.97 (1.02)\\
$k=3$, $M=1024$, $d=64$ & 0.08 & 289{,}001 & 0.90 (0.95) & 0.91 (0.94)\\
$k=3$, $M=1024$, $d=192$ & 0.23 & 289{,}394 & 0.79 (0.83) & 0.95 (0.97)\\
$k=4$, $M=256$, $d=128$ & 0.65 & 105{,}689 & 0.74 (0.82) & 0.93 (0.97)\\
$k=4$, $M=1024$, $d=751$ & 0.95 & 298{,}439 & 0.85 (0.90) & 0.95 (0.98)\\
$k=3$, $M=4096$, $d=3352$ & 1.00 & 333{,}476 & 0.93 (0.98) & 0.97 (1.00)\\
$k=5$, $M=4096$, $d=3018$ & 1.05 & 332{,}350 & 0.90 (0.95) & 0.99 (1.01)\\
$k=5$, $M=1024$, $d=755$ & 1.05 & 332{,}353 & 0.83 (0.87) & 0.96 (0.98)\\
$k=4$, $M=16384$, $d=16384$ & 1.29 & 293{,}250 & 0.95 (1.01) & 0.98 (1.01)\\
$k=3$, $M=4096$, $d=5029$ & 1.50 & 332{,}350 & 0.91 (0.96) & 0.98 (1.00)\\
$k=4$, $M=16384$, $d=18980$ & 1.50 & 290{,}808 & 0.96 (1.02) & 0.98 (1.00)\\
$k=4$, $M=16384$, $d=37959$ & 3.00 & 289{,}952 & 0.95 (1.01) & 0.97 (1.00)\\
\bottomrule
\end{tabular}
\end{table}

\paragraph{Data of Figure~\ref{fig:scale}.} Of the 12 configurations in the upper row of Figure~\ref{fig:scale}, two ($k=3$, $M=4096$ and $k=4$, $M=16384$) are test configurations; the other 10 follow the same design, with sample size and stopping rule fixed before sampling, and each has 291{,}000 to 332{,}000 direct-sampling failures (target about 290{,}000, stopping in blocks). Under random-sign encoding, common elements cancel in the difference table and sign composition does not affect the post-failure distribution, and the $|b|\le1$ used by the protocol under balanced differences also leaves the upper bound unchanged; the protocol experiments (\exptag{PROTO}), the statistical experiments (\exptag{STAT}), and these direct samples therefore estimate the same conditional miss rate, and we pool them by configuration in the figure. For these 10 configurations, the Skellam rule's one-sided upper bound is at most $1.025u$ at $u=0.01$ and $1.003u$ at $u=0.05$.

\subsection{Reliability and communication cost of the two-round protocol}\label{app:protocol}
The protocol experiments (\exptag{PROTO}) run the complete two-round synchronization on a real IBLT, corresponding to Section~\ref{sec:exp-e2e} of the main text. They test whether, when the second-round table length is set from the upper bound after a first-round failure, both the final reliability and the communication cost approach those of an oracle that knows the true $d$.

\paragraph{Trial construction.} Each trial generates common elements and difference elements separately, with all differences on the $A$ side for one-sided differences and half on each side for balanced differences. Common elements cancel in the difference table but participate as usual in table construction and cell-by-cell subtraction. The two rounds use independent cell hashes and share the sign hash and checksum hash. Trials that succeed in the first round require no upper bound, so only first-round-failure trials enter the comparison; the two levels of $u$ are evaluated separately on the same batch of trials. All upper bounds also share the same keys and hashes: for each first-round failure we record $U_b$, the second-round size, and the second-round outcome for each bound, so different bounds can be compared in paired fashion by trial. The second round decodes independently on the new table, consistent with Theorem~\ref{thm:protocol}; beyond the table length of Section~\ref{sec:flow} it imposes a minimum of 64 cells. All second-round recovery results match the true difference element by element.

\paragraph{Communication volume and sample size.} We count 20 bytes per cell: 4 bytes for the transmitted count and 8 bytes each for keySum and checksum, plus an 8-byte global counter for random-sign encoding. Both levels of $u$ use the same number of trials per configuration, set to give about 20{,}000 first-round failures, $\lceil2\times10^4/\widehat{\Pr}(F_1)\rceil$, with a cap of 3{,}000{,}000; the 13 configurations with no first-round failure in the pilot receive a fixed 100{,}000 trials. Among them, the 8 configurations with $k=4,5$ at load 0.95 and $M_1\ge4096$ still have no first-round failure in these trials and contribute only to the failure-rate and byte statistics. Intervals for the success-rate difference use the Newcombe hybrid score method for paired proportion differences.

Table~\ref{tab:e2e} summarizes the protocol results by $M_1$ and $u$. The final failure rate of both rules is at most 0.3\% even at the smallest $M_1=256$, and the conditional miss rate on balanced configurations agrees with the statistical experiments (\exptag{STAT}) of Section~\ref{sec:exp-cover}. Even when the upper bound is sufficient ($U_b\ge d$), second-round peeling itself can still fail, accounting for 0.06--0.08\% at $M_1=256$. The oracle's upper bound equals $d$, so by~\eqref{eq:risk} of Theorem~\ref{thm:protocol} its two-round failure comes only from the second term, the failure of second-round peeling with sufficient capacity. The oracle's second-round table is therefore smaller than those of the two rules, and in paired comparisons by trial the Skellam rule's success rate is up to 0.35 percentage points higher than the oracle's on some $M_1=256$ configurations. The bounds of the two rules occasionally underestimate $d$, and the second-round margin $1/(0.8\alpha_c)$ absorbs this slight underestimate (Section~\ref{sec:protocol}), so their final failure rates stay at the oracle's level.

\begin{table}[!htbp]
\centering
\caption{Two-round protocol summary ($k=3,4,5$). miss/u gives the median / maximum over balanced configurations with at least 1000 first-round failures. Final failure is the maximum over configurations. Second-round size and total bytes are medians relative to the oracle. The oracle's upper bound equals $d$ and is independent of $u$; both levels share the same trials, so the oracle results are identical and appear only in the upper half.}
\label{tab:e2e}
\begin{tabular}{@{}llrrrr@{}}
\toprule
$M_1$, $u$ & upper bound & miss/u & final failure & second-round size & total bytes\\
\midrule
256, 0.01 & Skellam & 0.73 / 0.87 & 0.24\% & +25.6\% & +14.6\%\\
& Gaussian & 1.00 / 1.13 & 0.25\% & +23.9\% & +13.7\%\\
& oracle & -- & 0.51\% & 0 & 0\\
1024, 0.01 & Skellam & 0.88 / 1.00 & 0.06\% & +11.4\% & +6.5\%\\
& Gaussian & 1.05 / 1.17 & 0.08\% & +11.1\% & +6.3\%\\
& oracle & -- & 0.09\% & 0 & 0\\
4096, 0.01 & Skellam & 0.92 / 1.01 & 0.02\% & +5.4\% & +3.1\%\\
& Gaussian & 1.01 / 1.09 & 0.03\% & +5.3\% & +3.0\%\\
& oracle & -- & 0.03\% & 0 & 0\\
16384, 0.01 & Skellam & 0.97 / 1.07 & 0.01\% & +2.6\% & +1.5\%\\
& Gaussian & 1.02 / 1.13 & 0.01\% & +2.6\% & +1.5\%\\
& oracle & -- & 0.02\% & 0 & 0\\
\addlinespace
256, 0.05 & Skellam & 0.93 / 1.01 & 0.29\% & +16.9\% & +9.7\%\\
& Gaussian & 1.01 / 1.07 & 0.27\% & +16.3\% & +9.4\%\\
1024, 0.05 & Skellam & 0.95 / 1.05 & 0.13\% & +7.8\% & +4.5\%\\
& Gaussian & 1.02 / 1.09 & 0.11\% & +7.7\% & +4.4\%\\
4096, 0.05 & Skellam & 0.97 / 1.04 & 0.02\% & +3.8\% & +2.2\%\\
& Gaussian & 1.01 / 1.07 & 0.02\% & +3.8\% & +2.1\%\\
16384, 0.05 & Skellam & 1.00 / 1.05 & 0.01\% & +1.9\% & +1.1\%\\
& Gaussian & 1.01 / 1.06 & 0.02\% & +1.9\% & +1.1\%\\
\bottomrule
\end{tabular}
\end{table}

\paragraph{Comparison with the calibration rule~\cite{wu2026selfsizing}.} The existing self-sizing approach enlarges the first-round estimate by an offline calibration factor that must cover the worst sign composition. This comparison examines how sensitive that approach is to sign composition under the two encodings. The baseline is the calibration rule of~\cite{wu2026selfsizing}: for $k=3$ that work groups results by first-round size and takes the worst 1\% quantile $q_{0.01}$ of post-failure $\widehat d/d$ over several loads and sign compositions (0.8045, 0.8997, 0.9482 for $M_1=256,1024,4096$). We use only its upper-bound component $U=\lceil\widehat d/q_{0.01}\rceil$ with the same second-round table length as our two rules; the capacity coefficient of that protocol, the deduction of elements recovered in the first round, and joint peeling across the two rounds are outside the comparison. Table~\ref{tab:vldb} compares this rule under both encodings, with and without accounting for $|b|$.

\begin{table}[!htbp]
\centering
\caption{Calibration rule versus the Skellam rule ($k=3$, $u=0.01$). The conditional miss rate /u gives the median / maximum for balanced configurations and the range for one-sided configurations (at least 1000 first-round failures). Final failure is the maximum. Total bytes are the median relative to the oracle.}
\label{tab:vldb}
\small
\setlength{\tabcolsep}{4pt}
\begin{tabular}{@{}rlcccc@{}}
\toprule
$M_1$ & rule & balanced & one-sided & final failure & total bytes\\
\midrule
256 & standard encoding, calibration factor & 0.92 / 1.02 & 0.09--0.90 & 0.27\% & +13.9\%\\
& standard encoding, calibration factor, with $|b|$ & 0.92 / 1.02 & 0 & 0.26\% & +13.9\%\\
& random sign, calibration factor & 0.87 / 0.97 & 0.85--0.93 & 0.23\% & +13.9\%\\
& random sign, Skellam rule & 0.72 / 0.80 & 0 & 0.24\% & +14.6\%\\
1024 & standard encoding, calibration factor & 0.88 / 0.94 & 0.05--1.03 & 0.07\% & +6.4\%\\
& standard encoding, calibration factor, with $|b|$ & 0.88 / 0.94 & 0 & 0.07\% & +6.4\%\\
& random sign, calibration factor & 0.98 / 1.03 & 0.80--1.04 & 0.08\% & +6.4\%\\
& random sign, Skellam rule & 0.89 / 0.92 & 0 & 0.06\% & +6.5\%\\
4096 & standard encoding, calibration factor & 0.84 / 0.86 & 0.36--0.89 & 0.02\% & +3.1\%\\
& standard encoding, calibration factor, with $|b|$ & 0.84 / 0.86 & 0 & 0.02\% & +3.1\%\\
& random sign, calibration factor & 0.90 / 0.94 & 0.77--0.95 & 0.02\% & +3.1\%\\
& random sign, Skellam rule & 0.97 / 1.01 & 0 & 0.01\% & +3.1\%\\
\bottomrule
\end{tabular}
\end{table}

With the same calibration factor, the standard encoding's conditional miss rate varies widely with sign composition, whereas the random-sign encoding's range is much smaller. For one-sided configurations, standard encoding ranges from $0.05u$ to $1.03u$ and random-sign encoding from $0.77u$ to $1.04u$. At fixed load 0.95--1.0, the difference is clearer: for one-sided and balanced differences, standard encoding gives $0.05$--$0.40u$ and $0.86$--$0.94u$, and random-sign encoding gives $0.80$--$0.95u$ and $0.85$--$1.01u$. Accounting for $|b|$ makes the one-sided miss zero under both encodings. In reliability and communication volume the two encodings hardly differ: in paired comparisons by trial, with the calibration factor and the choice of whether to account for $|b|$ fixed, the success-rate difference between random-sign and standard encoding lies within $\pm1\times10^{-3}$, with 2 of the 30 intervals excluding zero (both at $M_1=256$ and load 1.0), and the median total-byte difference is about 0.01\%.

\subsection{Computation cost}\label{app:timing}
The computation-cost measurement (\exptag{COST}) covers three parts: table size, construction time, and protocol time for the Skellam rule's lookup implementation; sizing time and Poisson-table memory for online computation; and the extra first-round decoding time of random-sign encoding.

\paragraph{Measurement and configurations.} The host has dual Intel Xeon Silver 4410Y processors (nominal 2.0 GHz, 24 physical cores and 48 logical CPUs), runs Linux, and uses a single-threaded Go implementation pinned to one core. Timing uses two runs. Before the first, we pause the host's task queue and wait for all tasks in the current batch to finish. We record that core's frequency and system load every 10 seconds; the frequency remains at 2.001 GHz throughout the run. Figure~\ref{fig:sizing} and Table~\ref{tab:dec-ratio} come from this run. In the second run, with the task queue idle, we repeat protocol timing with an added lookup branch. Table~\ref{tab:proto} comes from this run; its online and Gaussian protocol times differ from the first run by less than 3\%. Protocol timing uses $M_1=16384$ and 100,000 common elements, with the same 40 trials per configuration repeated 10 times and branch order rotated across repetitions. The second run includes lookup, online, Gaussian rule, oracle, and standard-encoding branches. First-round decoding time is counted per trial.

\paragraph{Table lookup.} Following Appendix~\ref{app:search-lut}, we build 24 tables for $M_1=1024,16384$, $k=3,4,5$, and $u=0.01,0.05$, with $R$ covered up to first-round loads of 1.5 and 3. Every $U$ is below $2^{16}$, so each entry occupies 2 bytes. At $M_1=16384$, tables covering load up to 1.5 occupy 121--172 KB and take 5--7 s to construct on a single thread; those covering load up to 3 occupy 241--345 KB and take 20--25 s. At $M_1=1024$, the tables occupy 8--22 KB and take less than 0.6 s to construct. For every $R$ in each table, we compare the stored result with the online implementation used in the experiments. All 24 tables agree entry by entry. A separately measured lookup at a random index takes about 1 ns. For protocol timing, we build three additional tables with $M_1=16384$ and $u=0.01$, covering first-round loads up to 4, and likewise verify them entry by entry. All first-round failure observations fall within these tables. In the protocol, each lookup-branch sizing takes 0.40--0.88 $\mu$s, compared with 0.75--1.33 $\mu$s for the Gaussian rule measured in the same way. Both include timing overhead. Total protocol time for the lookup branch differs from the Gaussian rule by less than 0.5\%.

\paragraph{Online sizing time and memory.} The Skellam rule computes the upper bound online when parameters vary across synchronizations or observations exceed the lookup table's range. It uses one Poisson table per candidate $n$, and whether these tables can be reused determines the sizing time, so the timing distinguishes cache conditions. Poisson tables are kept in a main-memory cache managed by the program. A cache entry corresponds to one value in the Poisson table of some candidate $n$ and stores its probability and cumulative probability as two double-precision numbers, 16 bytes in total; the capacity is 16,000,000 entries (about 256 MB). When the cache exceeds its limit, the program evicts tables in least-recently-used order until the occupancy plus the size of the incoming table is at most half the capacity, then inserts the new table. We time sizing under three cache conditions: cold starts each sizing with an empty cache; session has the same receiver perform consecutive synchronizations at fixed $(M,k)$, each possibly involving one or two protocol rounds and each processing a new observation, with the cache retained between synchronizations; warm immediately recomputes the same observation once. For sizing timings, we use the first 200 failure records per configuration from the statistical experiments (\exptag{STAT}) at $M=16384$, repeating each of the four modes, cold, session, warm, and Gaussian rule, 10 times, with their order rotated across repetitions.

Appendix~\ref{app:search-pois} predicts that sizing cost is roughly proportional to $\lambda_d$, that is, roughly $c^2$, and the measurements agree. Figure~\ref{fig:sizing} fixes $M$, with load proportional to $c$. Cold time grows roughly as load to the power 1.9, or $c^{1.9}$, from 0.06--0.12 ms at load about 0.08 ($c=8$) to 69--84 ms at load 3. Cache reuse reduces the time markedly: session is about one fifth of cold, and warm about one sixth to one seventh. At load 3, the Poisson tables requested cumulatively in session exceed the cache limit, the program evicts some old tables, and session rises to 41--55 ms. In the full protocol, peak session memory for Poisson tables is 49--86 MB at load 1.05, 128--174 MB at load 1.5, and reaches the 256 MB limit at load 3. Gaussian rule sizing takes only 0.06--0.29 $\mu$s.

The exact Skellam CDF can also be replaced by a Gaussian approximation at large parameters: in the exploration phase, across the 33 configurations with Poisson mean above 200, the miss count is unchanged, and the fraction of samples whose $U$ changes is at most 1.8\% in 32 of these configurations, and 10\% in the other configuration with only 10 samples. All results in this paper use the exact rule.

\begin{figure}[!htbp]
\centering\includegraphics[width=\linewidth]{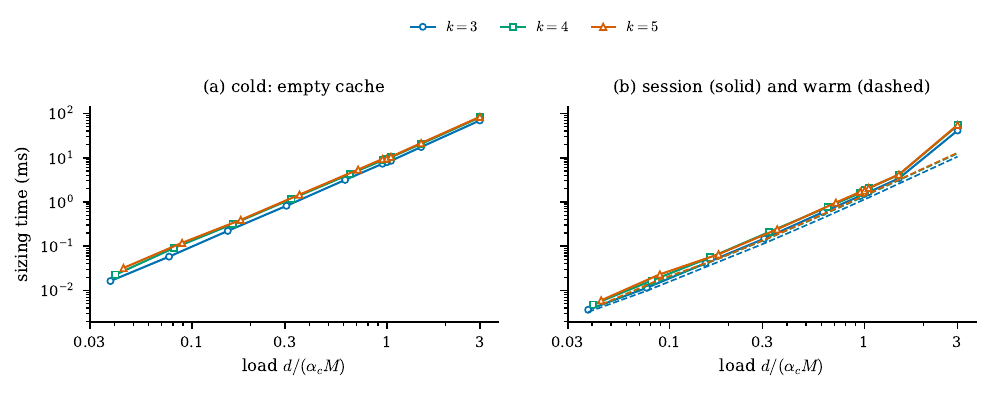}
\caption{Skellam rule time per sizing ($M=16384$, single thread, cache limit 16,000,000 entries). (a) cold; (b) session (solid) and warm (dashed).}
\label{fig:sizing}
\end{figure}

\paragraph{Extra decoding cost of random-sign encoding.} This cost arises in the pure-cell candidate check of first-round decoding. The decoder treats cells with count $\pm1$ as candidate pure cells and confirms them with the checksum. Under standard encoding with one-sided differences, all keys have the same sign, so a cell with count $\pm1$ necessarily contains exactly one key and every candidate passes; with both positive and negative differences, for example three keys with signs $+1,+1,-1$ in one cell also give count $+1$, and the checksum rejects such candidates. Under random-sign encoding the key signs are iid, so the situation equals that of balanced differences regardless of sign composition. A separate run (30 trials per configuration, 3 repetitions) counted the candidates checked per recovered key: random-sign encoding checks 1.77 ($k=3$) and 2.53 ($k=4$); standard encoding checks 1.00 under one-sided differences and matches random-sign encoding under balanced differences.

Table~\ref{tab:dec-ratio} gives the ratio of first-round decoding times of the two encodings. Under balanced differences the two are almost the same. Under one-sided differences, the extra time of random-sign encoding comes from the additional candidate checks, and the ratio grows with load; at load 3 it is about 3.1, because standard encoding stops almost immediately under one-sided differences. At load 3 the absolute times of both are only 0.04--0.16 ms, less than 1\% of total protocol time.

\begin{table}[!htbp]
\centering
\caption{Ratio of first-round decoding times (random-sign encoding / standard encoding; $M_1=16384$, 100,000 common elements, each trial timed as the median of 10 repetitions).}
\label{tab:dec-ratio}
\begin{tabular}{llccc}
\toprule
load & difference & $k=3$ & $k=4$ & $k=5$\\
\midrule
1.0 & one-sided & 1.26 & 1.41 & 1.40\\
1.05 & one-sided & 1.41 & 1.54 & 1.70\\
1.5 & one-sided & 1.75 & 2.24 & 2.41\\
3 & one-sided & 3.17 & 3.12 & 3.11\\
\addlinespace
1.05 & balanced & 1.06 & 1.04 & 1.03\\
3 & balanced & 1.00 & 0.99 & 1.00\\
\bottomrule
\end{tabular}
\end{table}

\end{document}